%% file: main.tex
\documentclass[11pt,letterpaper]{article}

\usepackage[
  left=1.25in,
  right=1.25in,
  top=1in,
  bottom=1in
]{geometry}

\usepackage{setspace}
\usepackage[T1]{fontenc}
\usepackage{microtype}
\usepackage{enumitem}
\usepackage[sort]{natbib}
\usepackage{mathrsfs}
\usepackage{amsmath,amssymb,amsfonts,amsthm,mathtools,bm}
\usepackage{bbm}
\usepackage{booktabs}
\usepackage{xcolor}
\usepackage{hyperref}
\usepackage[nameinlink,noabbrev]{cleveref}
\hypersetup{
  colorlinks=true,
  linkcolor=blue!55!black,
  citecolor=blue!55!black,
  urlcolor=blue!55!black
}

\theoremstyle{plain}
\newtheorem{theorem}{Theorem}[section]
\newtheorem{proposition}[theorem]{Proposition}
\newtheorem{lemma}[theorem]{Lemma}

\theoremstyle{definition}
\newtheorem{assumption}{Assumption}

\theoremstyle{remark}

\newcommand{\E}{\mathbb{E}}
\newcommand{\Pbb}{\mathbb{P}}

\newcommand{\indep}{\perp\!\!\!\perp}
\newcommand{\Mcal}{\mathcal{M}}

\newcommand{\Ltwo}{L_2}
\newcommand{\Range}{\operatorname{Range}}
\newcommand{\Ker}{\operatorname{Ker}}

\newcommand{\argmin}{\operatorname*{arg\,min}}

\newcommand{\keywords}[1]{%
  \par\medskip
  \noindent\textbf{Keywords:} #1
}

\title{Bridging Local and Population Causal Effects: A \\Proximal Instrumental Variable Approach}

\author{
Zixuan Yao and Guosheng Yin\\
Department of Statistics and Actuarial Science\\ 
University of Hong Kong
}
\date{September 18, 2026}

\begin{document}

\maketitle

\begin{abstract}
Instrumental variable (IV) methods address treatment endogeneity, but with noncompliance and heterogeneous treatment effects a binary instrument generally identifies the local average treatment effect (LATE) among compliers rather than the population average treatment effect (ATE). When treatment effects and compliance probabilities are heterogeneous and dependent through latent factors, the ATE need not be identified by IV variation alone. We develop a proximal instrumental variable framework that uses proxies for these factors to adjust for the compliance weighting in LATE and recover the ATE. We show that the ATE is identified through either an outcome bridge capturing latent treatment effect heterogeneity or a compliance bridge representing inverse compliance probabilities, and combining the two yields a doubly robust representation. We derive the efficient influence function under known and unknown instrument propensities. For estimation, we propose regularized kernel minimax bridge estimators, with orthogonalized conditional moments for estimated propensity, and a three-way sequential cross-fitting scheme. We establish projected bridge-risk bounds, and give conditions for asymptotic normality and semiparametric efficiency. The framework provides a proxy-based route from local IV effects to population causal effects under latent treatment effect and compliance heterogeneity. 
\end{abstract}


\section{Introduction}
\label{sec:introduction}

Instrumental variable methods exploit exogenous variation in treatment uptake
to identify causal effects when treatment is confounded. For a binary
instrument and binary treatment, the standard conditional independence,
exclusion, relevance, and monotonicity conditions identify the local average treatment
effect (LATE) among
compliers, whose treatment status is changed by the instrument, rather than the
population average treatment effect (ATE) \citep{imbens1994identification,angrist1996identification,
abadie2003semiparametric}. This local interpretation allows for unrestricted treatment effect heterogeneity, but it also restricts the
estimand to an instrument-specific subpopulation
\citep{imbens2014instrumental}. More generally,
instrumental variable estimands weight treatment effects according to
compliance with the instrument \citep{angrist2004treatment}. If common latent factors influence both
treatment effects and compliance probabilities, compliers may therefore have
systematically different treatment effects from the rest of the population. We
refer to this as the simultaneous heterogeneity problem. The no-simultaneous-heterogeneity condition of
\citet{hartwig2023average} resolves the problem by restricting the relevant
dependence in conditional means. We instead consider settings in which such
latent dependence may persist.

Existing methods for estimating population-relevant effects with instrumental variables typically rely on additional observed information or structural restrictions. Covariate-based methods extrapolate or reweight local effects using observed pretreatment covariates that capture relevant differences between compliers and
the target population
\citep{angrist2013extrapolate,aronow2013beyond}. Related work uses observed covariates to predict compliance and obtain bounds on effects in identifiable subgroups that are more generalizable than the complier population \citep{kennedy2020sharp}. Other approaches impose restrictions on the joint heterogeneity of treatment effects and instrument responsiveness, including additive no-interaction and mean-independence conditions \citep{wang2018bounded,hartwig2023average}. Latent-index and marginal treatment effect models instead characterize unobserved heterogeneity through resistance to treatment and use richer instrument support or functional restrictions to aggregate local effects toward population- and policy-relevant targets \citep{heckman2005structural,brinch2017beyond,mogstad2018using}. However, these approaches do not generally identify the population ATE when treatment effects and compliance are jointly driven by latent factors that are not captured by observed covariates. To address this gap, we develop a new identification strategy for settings in which
these latent heterogeneity factors are
imperfectly measured by auxiliary pretreatment variables. We refer to such variables as proxies, and show that combining proxy information with the binary IV can identify the ATE.

Beyond the IV literature, negative-control and proximal methods identify the ATE using proxies for unmeasured confounders and bridge functions that connect the observed proxies to the latent confounding structure \citep{miao2018identifying,miao2024confounding}. Proximal causal inference develops this idea into a general framework for identification and semiparametric inference based on outcome and treatment bridges \citep{cui2024semiparametric,tchetgen2024introduction}. Estimating these bridge functions leads to ill-posed conditional moment problems, for which flexible minimax and nonparametric methods have been developed \citep{dikkala2020minimax,bennett2023variational,singh2019kernel,bennett2023minimax}. Recent work further shows that target functionals may remain identified even when the bridge functions themselves are nonunique \citep{zhang2023proximal,bennett2026inference}.

Proxy and negative control ideas have also been used in problems adjacent to
instrumental variable analysis, although for purposes different from ours.
Within the IV literature, negative controls have been used to diagnose or
select valid instruments
\citep{davies2017compare,orihara2024valid,danieli2026negative}, to recover
causal effects when a proposed instrument may violate exogeneity or exclusion
\citep{dukes2025using}, and to study settings in which the instrument is latent
or measured with error by using observed proxies for the instrument
\citep{chalak2017instrumental,kedagni2023identifying}. Related proxy-based
methods have been developed for other latent structure problems, including
identification of causal effects within principal strata under unmeasured
confounding \citep{luo2024identification} and transport of treatment effects
across randomized trials in the presence of unobserved effect modifiers
\citep{su2026proximal}. The problem studied here is distinct. We instead consider a setting with a valid, randomized binary instrument in which simultaneous heterogeneity prevents the instrument alone from identifying the population ATE, and show that proxies for the shared latent heterogeneity can
restore identification. Our use of proxies also differs from that in standard
proximal causal inference: the proxies measure latent treatment effect and
compliance heterogeneity rather than serving as negative control variables. In
particular, we do not require the proxies to satisfy the conditional independence
restrictions with treatment or outcome that are typically imposed in standard
proximal methods \citep{tchetgen2024introduction}. 

Building on the proxy idea, we develop a proximal instrumental variable framework for identifying population effects with a conditionally randomized binary instrument and simultaneous heterogeneity. An outcome-side proxy represents latent treatment effect heterogeneity through an outcome bridge, while a treatment-side proxy represents inverse latent compliance probability
through a compliance bridge. Under instrument validity, monotonicity, positivity, proxy separation, and bridge existence, these restrictions yield primal and adjoint observed-data bridge equations, providing two complementary routes to identification. To our knowledge, this is the first
nonparametric identification framework that combines a binary instrument
with proxies for latent treatment effect and compliance heterogeneity to
identify the population ATE.

Our contributions are fourfold. First, we show that either observed
bridge equation identifies the ATE and that
combining the two equations yields a doubly robust representation (identification does not require either bridge solution to be
unique). Second, under completeness, we derive the
efficient influence function (EIF) and semiparametric efficiency bound in models
with known and unknown instrument propensity. We further show that the EIF admits an augmented inverse propensity weighted (AIPW) form involving the instrument propensity and the conditional means of a bridge-generated pseudo-outcome given each instrument arm, yielding additional robustness to misspecification of either component. Third, we construct regularized kernel
minimax estimators for the two bridge equations \citep{kallus2021causal,ghassami2022minimax,dikkala2020minimax}. When the instrument propensity is unknown, direct plug-in estimation introduces first-order nuisance errors into the conditional moments. We therefore construct residualized conditional moments that are Neyman orthogonal to first-order errors in the estimated nuisances, building on the orthogonal statistical learning literature \citep{chernozhukov2018double,foster2023orthogonal}. Because the pseudo-outcome regressions use outcomes generated from estimated bridges,
we propose a sequential cross-fitting scheme that assigns disjoint samples
to bridge and nuisance estimation, generated-outcome regression, and final
score evaluation \citep{fisher2023three,kennedy2023towards}. Fourth, we establish
projected bridge-risk bounds, and give conditions for
asymptotic normality and semiparametric efficiency. Overall, these results establish a new
framework for identifying and efficiently estimating population causal effects
from local instrument-induced variation under simultaneous heterogeneity.


\section{Problem Setup}
\label{sec:setup}

\subsection{Notation and 
estimand}

Let \(O=(Y,A,R,Z,W,X)\)
denote the observed data, where \(Y\) is the outcome,
\(A\in\{0,1\}\) is the treatment, \(R\in\{0,1\}\) is a randomized binary instrument,
\(Z\) and \(W\) are treatment-side and outcome-side proxies, respectively,
and \(X\) is a vector of observed baseline covariates. We observe independent
and identically distributed draws from an unknown distribution \(P_0\).

Let \(A(r)\) be the treatment that would be received under \(R=r\), and under the exclusion restriction, let \(Y(a)\)
be the potential outcome under treatment \(a\). Our target
estimand is the population ATE,
\begin{equation*}
	\theta_0=\E\{Y(1)-Y(0)\}.
\end{equation*}

Let \(\pi_0(X)=\Pbb(R=1\mid X)\)
denote the instrument propensity score, and define
	\[
	\rho_0(R,X)
	=
	\frac{R}{\pi_0(X)}
	-
	\frac{1-R}{1-\pi_0(X)}
	=
	\frac{R-\pi_0(X)}
	{\pi_0(X)\{1-\pi_0(X)\}}.
	\]
The propensity \(\pi_0(X)\) may be known by design or estimated from data.
The known-propensity case includes randomized encouragement experiments, in
which encouragement is assigned while treatment receipt remains voluntary,
such as encouragement to receive an influenza vaccine or to follow
depression-treatment guidelines
\citep{hirano2000assessing,ten2004causal}. The conditional formulation
allows the randomization probability to depend on prespecified baseline
covariates.

By the definition of \(\pi_0(X)\), for any integrable random variable \(V\),
	\[
	\E\{\rho_0(R,X)V\mid X\}
	=
	\E(V\mid R=1,X)-\E(V\mid R=0,X).
	\]
Thus \(\rho_0(R,X)\) extracts the conditional contrast induced by the
instrument. 


Let \(U\) denote latent effect-relevant heterogeneity, possibly multivariate and of mixed
discrete--continuous types. Under monotonicity, let
\(C=A(1)-A(0)\in\{0,1\}\) denote the compliance indicator. Write
\(\tau(U,X)=\E(Y(1)-Y(0)\mid U,X)\) and \(c(U,X)=\E(C\mid U,X)\). Thus \(\tau(U,X)\) and \(c(U,X)\) represent latent heterogeneity in
treatment effects and instrument responsiveness, respectively. 	By iterated expectations,
	\(\theta_0 = \E\{\tau(U,X)\}\). Additionally, under
the assumptions stated in the next section, the conditional Wald
estimand satisfies
	\[
	\theta_\mathrm{Wald}(X)
	=
	\frac{\E\{c(U,X)\tau(U,X) \mid X\}}
	{\E\{c(U,X) \mid X\}}
	=
	\E\{\tau(U,X) \mid X\}
	+
	\frac{
		\operatorname{Cov}\{\tau(U,X),c(U,X) \mid X\}
	}{
		\E\{c(U,X) \mid X\}
	}.
	\]
This decomposition shows that, after conditioning on \(X\), the local-to-population discrepancy arises from residual unobserved heterogeneity represented by \(U\). The Wald estimand weights treatment effects by  latent compliance probabilities, whereas the population ATE does not, so the two differ when \(\operatorname{Cov}\{\tau(U,X),c(U,X) \mid X\}\neq 0\). We use simultaneous heterogeneity to describe  settings in which latent treatment effect heterogeneity and instrument responsiveness remain dependent conditional on observed covariates. Existing approaches that identify population effects from a binary instrument typically impose restrictions that rule out this dependence \citep{hartwig2023average}. We instead allow the latent dependence to persist and use proxies for \(U\) to identify the population ATE.

Unlike the latent variable in standard proximal causal inference, which typically represents the unmeasured confounding structure \citep{tchetgen2024introduction}, \(U\) in our framework represents latent heterogeneity in treatment effects and instrument responsiveness and need not be a confounder between treatment and outcome. The proxies \(Z,W\) therefore need not capture the entire latent confounding structure, and additional unmeasured treatment--outcome confounding may exist without precluding identification of the ATE. This is because our identification relies primarily on a valid binary instrument, while the proxies address unobserved simultaneous heterogeneity. For simplicity, however, we allow \(U\) to serve both as the source of unobserved simultaneous heterogeneity and as an unmeasured treatment--outcome confounder.

\begin{figure}[t]
    \centering
    \includegraphics[width=0.5\linewidth]{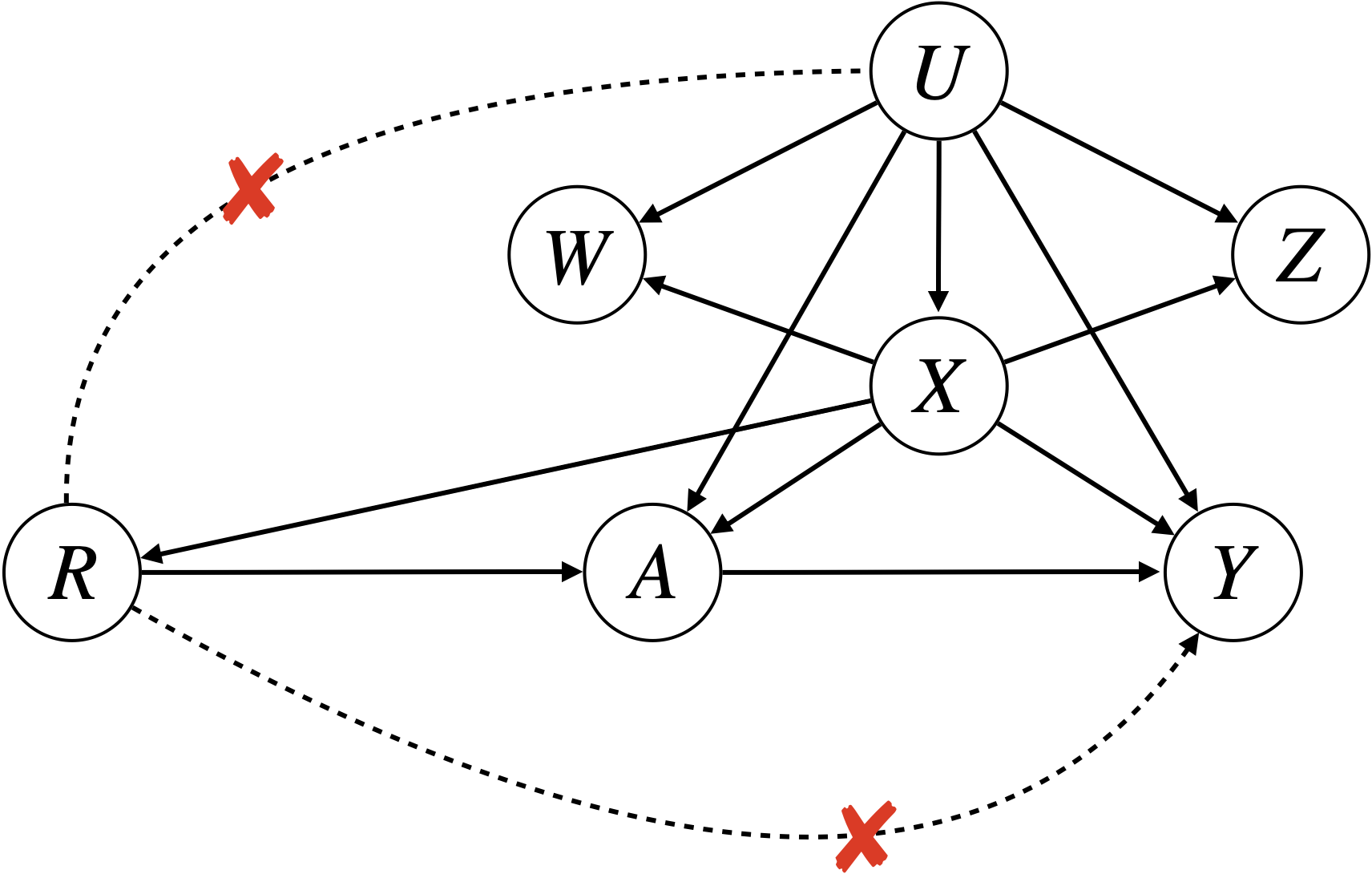}
    \caption{Schematic causal directed acyclic graph (DAG) for the proximal
    instrumental variable framework. 
    Crossed dashed arrows indicate IV independence and exclusion.}
    \label{fig:dag}
\end{figure}

Figure~\ref{fig:dag} summarizes the motivating causal structure.
Conditional on \(X\), the instrument \(R\) is independent of the latent
heterogeneity \(U\) and affects \(Y\) only through \(A\), while \(Z\) and \(W\)
serve as proxies for \(U\). The formal restrictions represented
schematically by the graph are stated in the next subsection.

\subsection{Identifying assumptions}
We now formalize the identifying assumptions.

\begin{assumption}[Consistency and exclusion]
\label{ass:consistency}
For \(r,a\in\{0,1\}\), \(A=A(R)\), \(Y=Y(A)\), 
and \(Y(a,r)=Y(a)\) almost surely.
\end{assumption}

Assumption~\ref{ass:consistency} states the standard consistency and exclusion
conditions for IV. The observed treatment and outcome equal the potential values under
the realized instrument and treatment, and the instrument affects the outcome
only through treatment receipt
\citep{imbens1994identification,angrist1996identification}.

\begin{assumption}[Instrument validity and uniform overlap]
\label{ass:iv}
The instrument satisfies
\[
R\indep\{Y(0),Y(1),A(0),A(1),U,Z,W\}\mid X.
\]
Moreover, for some \(\underline{\pi}>0\),
it holds that \(\underline{\pi}<\pi_0(X)<1-\underline{\pi}\)
almost surely.

\end{assumption}

The conditional independence condition formalizes conditional randomization of the instrument. It holds by design in randomized encouragement studies when the variables on the right-hand side are measured before randomization, and becomes an as-if-randomization assumption in observational IV studies \citep{angrist1996identification}. Uniform overlap strengthens the usual instrument positivity condition
\(0<\pi_0(X)<1\) by requiring the instrument propensity to be bounded away
from zero and one. This condition controls the inverse weights in
\(\rho_0(R,X)\) and is commonly imposed for regular inverse-weighted
estimation and inference \citep{abadie2003semiparametric,tan2006regression}.

\begin{assumption}[Monotonicity and compliance positivity]
\label{ass:monotonicity}
The instrument satisfies \(A(1)\geq A(0)\) almost surely. In addition, it holds that \(c(U,X)>0\)
almost surely.
\end{assumption}

Monotonicity rules out units whose treatment moves in the direction opposite to
the instrument and is the standard restriction underlying the LATE
interpretation \citep{imbens1994identification,angrist1996identification}.
The condition \(c(U,X)>0\) is a latent analogue of the compliance positivity condition used in inverse-compliance-score weighting approaches
\citep{aronow2013beyond}. It requires positive compliance probability within latent \((U,X)\) strata
but allows \(c(U,X)\) to be arbitrarily close to zero.

\begin{assumption}[Latent treatment effect and compliance heterogeneity]
\label{ass:latent}
The conditional treatment effect satisfies
\[
\E(Y(1)-Y(0)\mid C,U,Z,W,X)=\tau(U,X).
\]
Furthermore, the proxies do not reveal realized compliance type after conditioning on \((U,X)\),
\[
\E(C\mid U,Z,W,X)=c(U,X).
\]
\end{assumption}

Assumption~\ref{ass:latent} states that \((U,X)\) are sufficient for the relevant
treatment effect and compliance heterogeneity. Because \(Z\) and \(W\) are
proxies for \(U\), they may be associated with treatment effects and compliance
marginally, but provide no additional information once \((U,X)\) are given.
Likewise, the realized compliance type does not further modify the conditional
treatment effect after conditioning on \((U,X)\). \

\begin{assumption}[Proxy separation]
\label{ass:proxy}
The proxies satisfy \(Z\indep W\mid U,X\).
\end{assumption}

Assumption~\ref{ass:proxy} is a standard multi-proxy conditional independence
restriction in proximal causal inference \citep{tchetgen2024introduction}, plausible when \(Z\) and \(W\) arise from distinct measurement
mechanisms \citep{kuroki2014measurement,miao2018identifying}. 

\begin{assumption}[Existence of latent bridges]
\label{ass:latent-bridges}
There exist square-integrable functions \(\mu_0(W,X)\) and \(q_0(Z,X)\) such that
\begin{equation}
\E\{\mu_0(W,X)\mid U,X\}=\tau(U,X),
\quad
\E\{q_0(Z,X)\mid U,X\}=\frac{1}{c(U,X)}.
\label{eq:latent-bridges}
\end{equation}
\end{assumption}

Assumption~\ref{ass:latent-bridges} is a pair of range conditions: the relevant
latent functions must be representable through conditional averages of functions
of the observed proxies. Similar bridge existence conditions underlie proxy
identification and proximal causal inference
\citep{miao2024confounding,tchetgen2024introduction,cui2024semiparametric}. We refer to \(\mu_0(W,X)\) and \(q_0(Z,X)\) as the outcome bridge and compliance bridge, respectively. The former represents latent treatment effect heterogeneity, whereas the latter represents the
inverse of latent compliance probability. The square-integrability of the compliance bridge imposes an additional moment condition on \(c(U,X)\). That is, by conditional Jensen's inequality,
\[
\E\!\left\{\frac{1}{c(U,X)^2}\right\}
=
\E\!\left[
\left\{\E\bigl(q_0(Z,X)\mid U,X\bigr)\right\}^2
\right]
\leq
\E\!\left\{q_0(Z,X)^2\right\}
<\infty.
\]

\section{Nonparametric Identification}
\label{sec:identification}
We now eliminate the latent variable \(U\) from the identification problem by translating the latent bridge restrictions into conditional moment equations in observed variables. These equations lead to two identification formulas for the ATE and to an augmented representation that is robust to misspecification of either bridge.

Define the primal operator
\(T:\Ltwo(P_{W,X})\to\Ltwo(P_{Z,X})\) and the adjoint operator
\(T^\ast:\Ltwo(P_{Z,X})\to\Ltwo(P_{W,X})\) by
\[
\begin{aligned}
(T\mu)(Z,X)
&=
\E\{\rho_0(R,X)A\mu(W,X)\mid Z,X\},
\\
(T^\ast q)(W,X)
&=
\E\{\rho_0(R,X)Aq(Z,X)\mid W,X\}.
\end{aligned}
\]
The operator \(T\) maps a candidate outcome bridge into the corresponding
instrument-induced treatment contrast, while \(T^\ast\) is its
\(\Ltwo\)-adjoint. In particular,
	\[
\E\{q(Z,X)(T\mu)(Z,X)\}
=
\E\{\mu(W,X)(T^\ast q)(W,X)\}.
\]

Let \(b(Z,X)=\E\{\rho_0(R,X)Y\mid Z,X\}\),
which is the instrument-induced outcome contrast. 

\begin{proposition}[Observed bridge equations]
\label{thm:observed-bridges}
Under Assumptions~\ref{ass:consistency}--\ref{ass:latent-bridges}, the latent
bridge functions \((\mu_0,q_0)\) solve the primal and adjoint equations,
\begin{equation}
T\mu_0=b,
\quad
T^\ast q_0=1.
\label{eq:observed-bridges}
\end{equation}
Equivalently, the equations in \eqref{eq:observed-bridges} can be rewritten as
\begin{align}
\E\!\left[
\rho_0(R,X)\{Y-A\mu_0(W,X)\}
\mid Z,X
\right]
&=0,
\label{eq:primal-observed-bridge}
\\
\E\!\left[
\rho_0(R,X)Aq_0(Z,X)
\mid W,X
\right]
&=1.
\label{eq:adjoint-observed-bridge}
\end{align}
almost surely.
\end{proposition}

Proposition~\ref{thm:observed-bridges} replaces the latent bridge conditions
with observed-data conditional moments. The primal equation
characterizes the outcome bridge using \(Z\) as the conditioning proxy, whereas the adjoint equation characterizes the compliance bridge using \(W\). Accordingly, we later refer to \(\mu_0\) interchangeably as the outcome or primal bridge, and to
\(q_0\) as the compliance or adjoint bridge. These
inverse problems may be ill-posed and need not have unique solutions, but uniqueness is not needed for the identification result below.

The compliance bridge \(q_0\) also plays a natural role in Riesz-based debiasing. The linear functional \(g\mapsto\E\{g(W,X)\}\) on \(\Ltwo(P_{W,X})\) has Riesz representer \(1\). Since \(T^\ast q_0=1\), \(q_0\) is an adjoint preimage of this representer and therefore satisfies
\[
\E\{g(W,X)\}
=
\langle 1,g\rangle
=
\langle T^\ast q_0,g\rangle
=
\langle q_0,Tg\rangle
\]
for every \(g\in\Ltwo(P_{W,X})\) for which the expressions are well defined. Thus, \(q_0\) represents the target functional through the adjoint relation \(T^\ast q_0=1\) \citep{severini2012efficiency}. This motivates the residual correction in the augmented score below, which is a debiasing adjustment for the target functional \citep{chernozhukov2022automatic}.

Define the observed bridge solution sets,
\[
\mathcal S(T,b)
=
\{\mu\in\Ltwo(P_{W,X}):T\mu=b\},
\quad
\mathcal S(T^\ast,1)
=
\{q\in\Ltwo(P_{Z,X}):T^\ast q=1\}.
\]

To combine
the two bridges, for generic square-integrable \((\mu,q)\), define 
\begin{equation}
\psi_{\mathrm{DR}}(O;\mu,q)
=
\mu(W,X)
+
\rho_0(R,X)q(Z,X)\{Y-A\mu(W,X)\}.
\label{eq:raw-score}
\end{equation}
This doubly robust score augments the outcome bridge plug-in functional with a residual correction
weighted by the adjoint bridge, in a form similar to classical doubly robust proximal estimators \citep{cui2024semiparametric,tchetgen2024introduction}. The next theorem formalizes the identification and robustness results.

\begin{theorem}[Nonparametric identification and double robustness]
\label{thm:identification}
Under Assumptions~\ref{ass:consistency}--\ref{ass:latent-bridges}, for every
\(\mu\in\mathcal S(T,b)\) and \(q\in\mathcal S(T^\ast,1)\),
\[
\theta_0
=
\E\{\mu(W,X)\}
=
\E\{\rho_0(R,X)q(Z,X)Y\}.
\]
Moreover,
\[
\E\{\psi_{\mathrm{DR}}(O;\mu,q)\}
=
\theta_0
\]
whenever either
\(\mu\in\mathcal S(T,b)\) or \(q\in\mathcal S(T^\ast,1)\).

Let \((\mu_0,q_0) \in \mathcal S(T,b)\times\mathcal S(T^\ast,1)\).
Then, for arbitrary square-integrable \((\mu,q)\),
\begin{equation}
\E\{\psi_{\mathrm{DR}}(O;\mu,q)\}-\theta_0
=
-\E\!\left[
\{q(Z,X)-q_0(Z,X)\}
\{T\mu(Z,X)-T\mu_0(Z,X)\}
\right].
\label{eq:mixed-bias}
\end{equation}
\end{theorem}

Theorem~\ref{thm:identification} shows that either bridge equation identifies the ATE and the doubly robust score is unbiased whenever either bridge candidate belongs to its corresponding population solution set. Equation~\eqref{eq:mixed-bias} further shows that, with the true instrument propensity, the bias of the augmented functional is bilinear in the two bridge errors. These population identities motivate the primal, adjoint, and doubly robust estimators developed in Section~\ref{subsubsec:primal-adjoint-dr}.

\section{Semiparametric Efficiency}
\label{sec:efficiency}

For efficiency analysis, we adopt the observed-data semiparametric model defined by the primal bridge restriction. This is the standard formulation for nonparametric conditional moment models with unknown functions \citep{ai2003efficient} and is also used in earlier proximal causal inference work \citep{cui2024semiparametric}. Accordingly, regular parametric submodels and tangent spaces are defined relative to this conditional moment model rather than the full latent causal model. We then derive the EIF for both known- and unknown-propensity models and show that it admits an AIPW representation.

\subsection{Efficient influence function}

Let \(P\) denote a generic observed-data law, and write
\(\pi_P(X)=P(R=1\mid X)\) and \(\rho_P(R,X)=R/\pi_P(X)-(1-R)/\{1-\pi_P(X)\}\).
We consider two semiparametric observed-data models:
	\begin{enumerate}[label=(\roman*)]
		\item
		The unknown-propensity model \(\Mcal\) allows \(\pi_P\) to vary with \(P\) and requires the primal bridge equation to admit a square-integrable solution,
		\[
		\Mcal
		=
		\left\{
		P:
		\exists\,\mu_P\in L_2(P_{W,X}),\ 
		\E_P\!\left[
		\rho_P(R,X)\{Y-A\mu_P(W,X)\}
		\mid Z,X
		\right]
		=0
		\right\}.
		\]
		
		\item
		The known-propensity model \(\Mcal_{\pi_0}\) is the submodel of \(\Mcal\) in which the instrument propensity is fixed at its true value,
		\[
		\Mcal_{\pi_0}
		=
		\left\{
		P\in\Mcal:
		\pi_P(X)=\pi_0(X)\ \text{almost surely}
		\right\}.
		\]
	\end{enumerate}
In both cases, the model is characterized by the primal bridge equation. The adjoint bridge equation is not imposed along regular submodels, while its solution \(q_0\) at the true law \(P_0\) serves as the Riesz representer used to characterize the pathwise derivative.

We now impose an additional completeness condition for the efficiency analysis.

	\begin{assumption}[Completeness]
		\label{ass:completeness}
For every \(g\in L_2(P_{W,X})\),
		\[
		\E\{\rho_0(R,X)A g(W,X)\mid Z,X\}=0
		\quad\Longrightarrow\quad
		g(W,X)=0
		\quad\text{almost surely}.
		\]
For every \(h\in L_2(P_{Z,X})\),
		\[
		\E\{\rho_0(R,X)A h(Z,X)\mid W,X\}=0
		\quad\Longrightarrow\quad
		h(Z,X)=0
		\quad\text{almost surely}.
		\]
	\end{assumption}

The two completeness conditions imply injectivity of \(T\) and \(T^\ast\), respectively, and hence uniqueness of the primal and adjoint bridges. In addition, by the Hilbert-space identity, the second completeness condition yields \(\overline{\Range(T)}=\Ker(T^\ast)^\perp=L_2(P_{Z,X})\). Under the standard regularity conditions for nonparametric conditional moment models, this dense-range property implies that the primal bridge restriction is locally just identified and therefore imposes no additional first-order restriction on the observed-data score \citep{chen2018overidentification}. Consequently, \(\Mcal\) is locally nonparametric with tangent space \(L_2^0(P_0)\). The same argument applies within \(\Mcal_{\pi_0}\), where the primal bridge restriction imposes no further first-order restriction beyond the fixed-propensity condition.

The completeness condition has a natural proxy-relevance interpretation: it requires the observed proxies to be sufficiently informative about the latent structure for the relevant conditional expectation operators to distinguish nonzero square-integrable directions \citep{miao2018identifying,tchetgen2024introduction}. It is, however, stronger than is generally necessary for regular inference on the target functional. In proximal and related inverse problems, root-\(n\) regular, and in some cases semiparametrically efficient inference may remain possible without completeness or uniqueness of the bridge functions under alternative technical conditions on the bridge operators or solution sets \citep{bennett2023proximalrl,zhang2023proximal,severini2012efficiency,bennett2026inference}. In this paper, we impose completeness as a convenient sufficient condition for the efficiency analysis, while noting that it may be relaxed under a more general treatment.

Define \(\omega_0(X) = \E\!\left[ q_0(Z,X)\{Y-A\mu_0(W,X)\} \mid X \right]\).
The following theorem characterizes the efficient influence function and efficiency bound in the two observed-data models.

	\begin{theorem}[Efficient influence function and efficiency bound]
		\label{thm:eif}
Suppose Assumptions~\ref{ass:consistency}--\ref{ass:completeness} hold, and \(\E\!\left[ q_0(Z,X)^2\{Y-A\mu_0(W,X)\}^2 \right]<\infty\).
Then, in both \(\Mcal_{\pi_0}\) and \(\Mcal\), \(\theta_0\) is pathwise differentiable with canonical gradient
		\begin{equation}
			\phi_0(O)
			=
			\mu_0(W,X)-\theta_0
			+
			\rho_0(R,X)
			\left[
			q_0(Z,X)\{Y-A\mu_0(W,X)\}
			-
			\omega_0(X)
			\right].
			\label{eq:eif}
		\end{equation}
Consequently, the semiparametric efficiency bound in either model is \(V_{\mathrm{eff}}=\E\{\phi_0(O)^2\}\).
	\end{theorem}

The proof of Theorem~\ref{thm:eif} uses the Riesz representer characterization of regular linear functionals in conditional moment models \citep{ai2003efficient,severini2012efficiency}. Because \(T^\ast q_0=1\), the first-order effect of a bridge perturbation on the target can be represented through the corresponding perturbation of the conditional moment. The pathwise derivative can therefore be characterized without requiring differentiability of the bridge solution as the observed-data law varies. This avoids the stronger operator surjectivity condition imposed in \citet{cui2024semiparametric} to enable such differentiation.

The equality of the two efficiency bounds is an adaptivity property. When the propensity is known, the uncentered bridge gradient is a valid gradient but generally is not the canonical gradient, and projection onto the known-propensity tangent space removes the component associated with instrument assignment. When the propensity is unknown, exactly the same correction arises from the pathwise variation of the propensity-dependent weight \(\rho\).

\subsection{The AIPW representation of EIF}
\label{subsec:aipw-representation}

To expose the role of the instrument propensity, we define the pseudo-outcome \(H_0=q_0(Z,X)\{Y-A\mu_0(W,X)\}\),
and the corresponding instrument-specific regressions \(\nu_r^\star(X)=\E(H_0\mid R=r,X)\),
for \(r\in\{0,1\}\).
The primal bridge restriction implies \(\nu_1^\star(X)=\nu_0^\star(X)=\omega_0(X)\).
The canonical gradient in \eqref{eq:eif} therefore can be written in the AIPW form,
\begin{equation}
    \phi_0(O)
    =
    \mu_0(W,X)-\theta_0
    +
    \frac{R}{\pi_0(X)}
    \{H_0-\nu_1^\star(X)\}
    -
    \frac{1-R}{1-\pi_0(X)}
    \{H_0-\nu_0^\star(X)\}
    +
    \nu_1^\star(X)-\nu_0^\star(X).
    \label{eq:eif-aipw}
\end{equation}
At the true bridge pair, the final difference term vanishes, and \eqref{eq:eif-aipw} is algebraically equivalent to \eqref{eq:eif}. The expanded form is nevertheless useful for estimation because it separates the propensity component from the regression adjustment. For generic bridge candidates, the induced pseudo-outcome regressions may differ across instrument arms, which motivates allowing the working regressions to depend on \(R\).

For a generic candidate bridge pair \((\mu,q)\), define
\(H_{\mu,q}=q(Z,X)\{Y-A\mu(W,X)\}\). For \(r\in\{0,1\}\), let
\(\nu_r[\mu,q](x)=\E(H_{\mu,q}\mid R=r,X=x)\) denote the corresponding
instrument-specific regression, with \([\mu,q]\) indicating its dependence on
the candidate bridge pair.
Let \(\nu_0,\nu_1\in L_2(P_X)\) be arbitrary working regressions, let
\(\pi\) be a propensity candidate, and write
\(\eta=(\mu,q,\pi,\nu_0,\nu_1)\). Replacing the nuisance components in the
AIPW representation by these generic candidates yields the working score,
\begin{equation*}
\Psi(O;\eta)
=
\mu(W,X)
+
\frac{R}{\pi(X)}
\{H_{\mu,q}-\nu_1(X)\}
-
\frac{1-R}{1-\pi(X)}
\{H_{\mu,q}-\nu_0(X)\}
+
\nu_1(X)-\nu_0(X).
\end{equation*}

The next result gives an exact decomposition of the population bias of \(\Psi(O;\eta)\).

\begin{proposition}[Nested robustness]
\label{prop:nested-robustness}
For arbitrary square-integrable \((\mu,q,\nu_0,\nu_1)\) and any propensity
candidate \(\pi\) bounded away from zero and one,
\begin{equation*}
\E\{\Psi(O;\eta)\}-\theta_0
=
-\E\!\left[
(q-q_0)\{T\mu-T\mu_0\}
\right]+
\E\!\left[
(\pi_0-\pi)
\left\{
\frac{\nu_1[\mu,q]-\nu_1}{\pi}
+
\frac{\nu_0[\mu,q]-\nu_0}{1-\pi}
\right\}
\right].
\end{equation*}
Consequently, \(\E\{\Psi(O;\eta)\}=\theta_0\) whenever both of the following
conditions hold:
\begin{enumerate}[label=(\alph*)]
    \item either \(T\mu=T\mu_0\) or
    \(T^\ast q=T^\ast q_0\) almost surely; and
    \item either \(\pi=\pi_0\) or
    \(\nu_r=\nu_r[\mu,q]\) almost surely for \(r\in\{0,1\}\).
\end{enumerate}
\end{proposition}

Proposition~\ref{prop:nested-robustness} decomposes the bias into two terms
with different robustness structures. The first vanishes when either the primal
or adjoint bridge equation is correctly specified, yielding double robustness
with respect to the bridge pair. The second vanishes when either the instrument
propensity is correctly specified or the working regressions equal the
pseudo-outcome regressions induced by \((\mu,q)\), yielding the usual AIPW-type
robustness. Notably, the pseudo-outcome regression targets are
\(\nu_r[\mu,q]\), which depend on the working bridge pair, rather than on
\((\mu_0,q_0)\). This representation motivates the cross-fitted efficient estimator developed
in Section~\ref{subsubsection:efficient}.

\section{Estimation and Inference}
\label{sec:kernel-estimation}
The identification results motivate primal, adjoint, and doubly robust estimators of the ATE, while the canonical gradient leads to an efficient estimator. All four estimators depend on the primal and adjoint bridges. We estimate these bridges by regularized kernel minimax learning \citep{dikkala2020minimax,ghassami2022minimax,kallus2021causal}. When the instrument propensity is known, the observed bridge equations can be solved directly. When it is unknown, we use residualized bridge moments that are orthogonal to first-order errors in preliminary nuisance estimation. We then construct cross-fitted ATE estimators. 

\subsection{Regularized kernel minimax learning}
\label{subsec:stabilised-kernel-minimax}

Because the bridge equations are conditional moment inverse problems, estimation must control both conditional moment fit and bridge complexity. We use a regularized minimax formulation in which an adversarial critic searches for violations of each bridge restriction \citep{dikkala2020minimax,bennett2023variational}. The bridge estimators are motivated by the variational identity,
\begin{equation}
\label{eq:minimax-identity}
\frac{1}{2}
\left\|
\E\{D_f(O)\mid S\}
\right\|_{L_2(P_S)}^2
=
\sup_{g\in L_2(P_S)}
\left[
\E\{D_f(O)g(S)\}
-
\frac{1}{2}\E\{g(S)^2\}
\right],
\end{equation}
where \(D_f(O)\) is a residual indexed by a candidate bridge \(f\) and
\(S\) is the conditioning variable in the corresponding conditional moment
restriction. The maximizer in \eqref{eq:minimax-identity} is
\[
g_f^\star(S)=\E\{D_f(O)\mid S\}.
\]
Thus, minimizing the maximized criterion over \(f\) is equivalent, at the
population level, to minimizing the squared \(L_2(P_S)\) norm of the
conditional moment violation. The variational formulation avoids direct estimation of the conditional expectation by allowing the critic to detect deviation from the bridge restriction. 

In estimation, we replace population expectations by empirical averages and restrict the adversary to a critic RKHS. The empirical second-moment term stabilizes the inner maximization, while the RKHS penalties control the complexity of the critic and regularize the ill-posed bridge problem. Let \(\mathcal H_\mu\) and \(\mathcal H_q\) be bridge RKHSs on \((W,X)\) and \((Z,X)\), with kernels \(k_\mu\) and \(k_q\), respectively. Let \(\mathcal G_\mu\) on \((Z,X)\) and \(\mathcal G_q\) on \((W,X)\) be the corresponding critic RKHSs, with kernels \(k'_\mu\) and \(k'_q\). We denote the critic regularization parameters by \(\kappa_{\mu,n}\) and \(\kappa_{q,n}\) and the bridge regularization parameters by \(\lambda_{\mu,n}\) and \(\lambda_{q,n}\). The representer theorem reduces the resulting saddle-point problems to finite-dimensional optimization over kernel evaluations \citep{scholkopf2001generalized}. 

\subsection{Bridge estimation with known instrument propensity}
\label{subsec:known-pi-estimation}

When \(\pi_0\) is known, \(\rho_0(R,X)\) is also known and the bridge moments require no preliminary nuisance estimation. Write \(\rho_{0i}=\rho_0(R_i,X_i)\). The primal bridge is estimated by
\begin{equation}
	\label{eq:known-primal-bridge}
	\begin{aligned}
		\widehat\mu
		\in\argmin_{\mu\in\mathcal H_\mu}\Bigg\{
		\sup_{g\in\mathcal G_\mu}\Bigg[{}
		&\frac1n\sum_{i=1}^n
		\rho_{0i}\{Y_i-A_i\mu(W_i,X_i)\}g(Z_i,X_i)\\
		&-\frac1{2n}\sum_{i=1}^n g(Z_i,X_i)^2
		-\frac{\kappa_{\mu,n}}2\|g\|_{\mathcal G_\mu}^2
		\Bigg]
		+\frac{\lambda_{\mu,n}}2\|\mu\|_{\mathcal H_\mu}^2
		\Bigg\}.
	\end{aligned}
\end{equation}
The adjoint bridge equation can be written as
	\[
	\E[\rho_0\{R-Aq_0(Z,X)\}\mid W,X]=0,
	\]
because conditional randomization of \(R\) implies \(\E\{\rho_0(R,X)R\mid W,X\}=1\). We therefore estimate the adjoint bridge by
\begin{equation}
	\label{eq:known-adjoint-bridge}
	\begin{aligned}
		\widehat q
		\in\argmin_{q\in\mathcal H_q}\Bigg\{
		\sup_{g\in\mathcal G_q}\Bigg[{}
		&\frac1n\sum_{i=1}^n
		\rho_{0i}\{R_i-A_iq(Z_i,X_i)\}g(W_i,X_i)\\
		&-\frac1{2n}\sum_{i=1}^n g(W_i,X_i)^2
		-\frac{\kappa_{q,n}}2\|g\|_{\mathcal G_q}^2
		\Bigg]
		+\frac{\lambda_{q,n}}2\|q\|_{\mathcal H_q}^2
		\Bigg\}.
	\end{aligned}
\end{equation}
In each criterion, the critic measures the largest regularized violation of the corresponding conditional moment, while the outer minimization selects a regularized bridge that makes this violation small.

\subsection{Bridge estimation with unknown instrument propensity}
\label{subsec:unknown-pi-estimation}

When \(\pi_0\) is unknown, directly replacing it in \(\rho_0\) by an estimated propensity generally makes the bridge moments sensitive to first-order propensity error. Following the orthogonal statistical learning literature, we instead construct residualized moments that are Neyman orthogonal to the preliminary nuisance functions \citep{chernozhukov2018double}. Let \(\gamma_0(Z,X)=\E(Y\mid Z,X)\)
and \(\alpha_0(W,Z,X)=\E(A\mid W,Z,X)\)
denote auxiliary regressions used to orthogonalize the bridge moments. For generic working functions \(\pi\), \(\gamma\), and \(\alpha\), define
	\[
	\begin{aligned}
		\xi_\mu(\mu;\pi,\gamma,\alpha)
		&=
		\{R-\pi(X)\}
		\{Y-\gamma(Z,X)-[A-\alpha(W,Z,X)]\mu(W,X)\},
		\\
		\xi_q(q;\pi,\alpha)
		&=
		\{R-\pi(X)\}
		\{R-\pi(X)-[A-\alpha(W,Z,X)]q(Z,X)\}.
	\end{aligned}
	\]
Write \(\xi_{\mu,0}(\mu)=\xi_\mu(\mu;\pi_0,\gamma_0,\alpha_0)\)
and \(\xi_{q,0}(q)=\xi_q(q;\pi_0,\alpha_0)\).

\begin{proposition}[Orthogonality of the residualized bridge moments]
\label{prop:orthogonal-residualized-moments}
Under Assumption~\ref{ass:iv}, for every square-integrable candidate
\(\mu\) and \(q\),
\begin{align*}
\E\{\xi_{\mu,0}(\mu)\mid Z,X\}
&=
\pi_0(X)\{1-\pi_0(X)\}\{b-T\mu\}(Z,X),\\
\E\{\xi_{q,0}(q)\mid W,X\}
&=
\pi_0(X)\{1-\pi_0(X)\}\{1-T^\ast q\}(W,X).
\end{align*}
Consequently, the residualized and original bridge restrictions have the
same population solution sets.

Moreover, for arbitrary working functions \(\pi\), \(\gamma\), and \(\alpha\),
write \(\Delta_\pi(X)=\pi(X)-\pi_0(X), \Delta_\gamma(Z,X)=\gamma(Z,X)-\gamma_0(Z,X)\),
and \(\Delta_\alpha(W,Z,X) = \alpha(W,Z,X)-\alpha_0(W,Z,X)\).
Then,
\begin{align}
\E\{\xi_\mu(\mu;\pi,\gamma,\alpha)
      -\xi_{\mu,0}(\mu)\mid Z,X\}&=
\Delta_\pi(X)
\left[
\Delta_\gamma(Z,X)
-
\E\{\Delta_\alpha(W,Z,X)\mu(W,X)\mid Z,X\}
\right],
\label{eq:primal-orthogonality-identity}\\
\E\{\xi_q(q;\pi,\alpha)-\xi_{q,0}(q)\mid W,X\}
 &=
\Delta_\pi(X)
\left[
\Delta_\pi(X)
-
\E\{\Delta_\alpha(W,Z,X)q(Z,X)\mid W,X\}
\right].
\label{eq:adjoint-orthogonality-identity}
\end{align}
Thus, the residualized moments are Neyman orthogonal at the truth for every
fixed square-integrable candidate bridge.
\end{proposition}

Proposition~\ref{prop:orthogonal-residualized-moments} shows that residualization preserves the population bridge solution sets and makes preliminary nuisance errors enter the bridge moments only through products involving the propensity error. This second-order sensitivity therefore permits flexible nuisance estimation under product-rate conditions.

To preserve this orthogonality in estimation, the nuisance predictions entering the bridge criteria are obtained out-of-fold. For observation \(i\), let \(\widehat\pi_i\), \(\widehat\gamma_i\), and \(\widehat\alpha_i\) denote predictions from nuisance learners (e.g., neural networks) trained without the inner fold containing observation \(i\). The resulting unknown-propensity bridge estimators retain the kernel minimax structure of the known-propensity case, with the original residuals replaced by the orthogonalized ones. The primal bridge estimator is

\begin{equation}
	\label{eq:unknown-primal-bridge}
	\begin{aligned}
		\widehat\mu
		\in\argmin_{\mu\in\mathcal H_\mu}\Bigg\{
		\sup_{g\in\mathcal G_\mu}\Bigg[{}
		&\frac1n\sum_{i=1}^n
		(R_i-\widehat\pi_i)
		\{Y_i-\widehat\gamma_i
		-(A_i-\widehat\alpha_i)\mu(W_i,X_i)\}
		g(Z_i,X_i)\\
		&-\frac1{2n}\sum_{i=1}^n g(Z_i,X_i)^2
		-\frac{\kappa_{\mu,n}}2\|g\|_{\mathcal G_\mu}^2
		\Bigg]
		+\frac{\lambda_{\mu,n}}2\|\mu\|_{\mathcal H_\mu}^2
		\Bigg\}.
	\end{aligned}
\end{equation}
Similarly, the adjoint bridge estimator is
\begin{equation}
	\label{eq:unknown-adjoint-bridge}
	\begin{aligned}
		\widehat q
		\in\argmin_{q\in\mathcal H_q}\Bigg\{
		\sup_{g\in\mathcal G_q}\Bigg[{}
		&\frac1n\sum_{i=1}^n
		(R_i-\widehat\pi_i)
		\{R_i-\widehat\pi_i
		-(A_i-\widehat\alpha_i)q(Z_i,X_i)\}
		g(W_i,X_i)\\
		&-\frac1{2n}\sum_{i=1}^n g(W_i,X_i)^2
		-\frac{\kappa_{q,n}}2\|g\|_{\mathcal G_q}^2
		\Bigg]
		+\frac{\lambda_{q,n}}2\|q\|_{\mathcal H_q}^2
		\Bigg\}.
	\end{aligned}
\end{equation}

	\subsection{Cross-fitted ATE estimators}
\label{subsec:sequential-cross-fitting}

We now combine the estimated nuisance functions with the identifying and efficient representations to construct the ATE estimators. The primal, adjoint, and doubly robust estimators require only ordinary cross-fitting. The efficient estimator additionally requires estimation of the conditional means of a bridge-generated pseudo-outcome, for which we propose a three-way sequential sample splitting approach.

\subsubsection{Primal, adjoint, and doubly robust estimators}
\label{subsubsec:primal-adjoint-dr}

Partition the sample into a fixed number \(L\ge2\) of folds
\(\mathcal I_1,\ldots,\mathcal I_L\). For each evaluation fold
\(\mathcal I_\ell\), define the training sample
\(\mathcal T_\ell=\{1,\ldots,n\}\setminus\mathcal I_\ell\). Let \(\widehat\mu_\ell\) and \(\widehat q_\ell\) denote the primal and
adjoint bridge estimates obtained from \(\mathcal T_\ell\). When \(\pi_0\) is
known, we estimate the bridges using the known-propensity criteria in
\eqref{eq:known-primal-bridge} and \eqref{eq:known-adjoint-bridge}. When
\(\pi_0\) is unknown, we instead use the orthogonalized criteria in
\eqref{eq:unknown-primal-bridge} and \eqref{eq:unknown-adjoint-bridge} with
inner cross-fitting within \(\mathcal T_\ell\). Specifically, we further partition \(\mathcal T_\ell\) into a
fixed number \(D\ge2\) of inner folds
\(\mathcal T_{\ell,1},\ldots,\mathcal T_{\ell,D}\). For each \(d\), we
estimate \(\pi_0\), \(\gamma_0\), and \(\alpha_0\) using
\(\mathcal T_\ell\setminus\mathcal T_{\ell,d}\) and evaluate the fitted
functions on \(\mathcal T_{\ell,d}\). To simplify notation, we suppress the
outer-fold index \(\ell\) and denote these inner-fold fits by
\(\widehat\pi_d\), \(\widehat\gamma_d\), and \(\widehat\alpha_d\). We combine
their out-of-fold predictions across \(d\) to form the orthogonalized bridge
criteria on \(\mathcal T_\ell\), from which we obtain
\(\widehat\mu_\ell\) and \(\widehat q_\ell\).

For evaluating the ATE scores, when \(\pi_0\) is unknown, we additionally
fit a propensity estimator \(\widehat\pi_\ell\) on the full
\(\mathcal T_\ell\) and evaluate it on \(\mathcal I_\ell\). When the
propensity is known, we set \(\widehat\pi_\ell=\pi_0\). We fit two separate propensity estimators,
\(\widehat\pi_d\) and \(\widehat\pi_\ell\), because they serve different
sample-splitting roles. The former enters the bridge criteria evaluated
within \(\mathcal T_\ell\) and therefore requires inner cross-fitting,
whereas the latter enters only the ATE scores evaluated on the disjoint fold
\(\mathcal I_\ell\) and can be trained on the full \(\mathcal T_\ell\).
Define \(\widehat\rho_\ell(R,X) = R/\widehat\pi_\ell(X) - (1-R)/\{1-\widehat\pi_\ell(X)\}\).

The cross-fitted primal and adjoint estimators are
\begin{align*}
	\widehat\theta_\mu
	&=
	\frac{1}{n}
	\sum_{\ell=1}^L
	\sum_{i\in\mathcal I_\ell}
	\widehat\mu_\ell(W_i,X_i),
	\\
	\widehat\theta_q
	&=
	\frac{1}{n}
	\sum_{\ell=1}^L
	\sum_{i\in\mathcal I_\ell}
	\widehat\rho_\ell(R_i,X_i)
	\widehat q_\ell(Z_i,X_i)Y_i.
\end{align*}
Augmenting the two gives the cross-fitted doubly robust estimator
\[
\widehat\theta_{\mathrm{DR}}
=
\frac{1}{n}
\sum_{\ell=1}^L
\sum_{i\in\mathcal I_\ell}
\left[
\widehat\mu_\ell(W_i,X_i)
+
\widehat\rho_\ell(R_i,X_i)
\widehat q_\ell(Z_i,X_i)
\{Y_i-A_i\widehat\mu_\ell(W_i,X_i)\}
\right].
\]
Thus, each observation is evaluated using bridge and propensity estimates
constructed without its evaluation fold, thereby separating nuisance
estimation from score evaluation.

\subsubsection{Efficient estimator}
\label{subsubsection:efficient}

The efficient estimator additionally requires the conditional means of the
bridge-generated pseudo-outcome. Because this pseudo-outcome depends on the
estimated bridges, we separate bridge estimation, pseudo-outcome regression,
and final score evaluation using a three-way sequential cross-fitting
procedure \citep{kennedy2023towards,fisher2023three}. We partition the sample
into a fixed number \(L\ge3\) of outer folds
\(\mathcal I_1,\ldots,\mathcal I_L\). For an evaluation fold
\(\mathcal I_\ell\), let \(\ell'=\ell+1\) if \(\ell<L\) and
\(\ell'=1\) otherwise. We reserve
\(\mathcal J_\ell=\mathcal I_{\ell'}\) for the pseudo-outcome regression and
use the remaining observations \(\mathcal T_\ell = \{1,\ldots,n\}\setminus(\mathcal I_\ell\cup\mathcal J_\ell)\)
for bridge estimation. On the training fold \(\mathcal T_\ell\), we construct
\(\widehat\mu_\ell\), \(\widehat q_\ell\), and
\(\widehat\pi_\ell\) using the same bridge and propensity estimation procedures described
in Section~\ref{subsubsec:primal-adjoint-dr}.

For each outer fold \(\ell\), define the fitted pseudo-outcome \(\widehat H_\ell(O) = \widehat q_\ell(Z,X) \{Y-A\widehat\mu_\ell(W,X)\}\)
and write \(\widehat H_{i,\ell} = \widehat H_\ell(O_i)\).
We then regress \(\{\widehat H_{i,\ell}:i\in\mathcal J_\ell\}\)
on \((R_i,X_i)\), and for \(r\in\{0,1\}\), let
\(\widehat\nu_{r,\ell}(x)\) denote the resulting prediction evaluated at
\(R=r\) and \(X=x\). 

For \(i\in\mathcal I_\ell\), define the cross-fitted AIPW score
\begin{align}
	\widehat\psi_{i,\ell}
	&=
	\widehat\mu_\ell(W_i,X_i)
	+
	\frac{R_i}{\widehat\pi_\ell(X_i)}
	\{\widehat H_{i,\ell}-\widehat\nu_{1,\ell}(X_i)\}
	-
	\frac{1-R_i}{1-\widehat\pi_\ell(X_i)}
	\{\widehat H_{i,\ell}-\widehat\nu_{0,\ell}(X_i)\}
	\nonumber\\
	&\quad
	+
	\widehat\nu_{1,\ell}(X_i)
	-
	\widehat\nu_{0,\ell}(X_i).
	\label{eq:cross-fitted-score}
\end{align}
The efficient estimator is
\begin{equation}
	\widehat\theta_{\mathrm{eff}}
	=
	\frac{1}{n}
	\sum_{\ell=1}^L
	\sum_{i\in\mathcal I_\ell}
	\widehat\psi_{i,\ell}.
	\label{eq:final-estimator}
\end{equation}
We estimate the asymptotic variance of the efficient estimator by
\begin{equation}
	\widehat V_\mathrm{eff}
	=
	\frac{1}{n-1}
	\sum_{\ell=1}^{L}
	\sum_{i\in\mathcal I_\ell}
	\left(
	\widehat\psi_{i,\ell}
	-
	\widehat\theta_\mathrm{eff}
	\right)^2.
	\label{eq:variance-estimator}
\end{equation}
A nominal \(95\%\) asymptotic confidence interval is \(\widehat\theta_\mathrm{eff} \pm 1.96\sqrt{\widehat V_\mathrm{eff}/n}\).

As \(\ell\) ranges from \(1\) to \(L\), the roles of the outer folds rotate
cyclically, so every observation contributes to the final estimator exactly
once, with the corresponding nuisance functions estimated on disjoint
samples. The resulting estimator in \eqref{eq:final-estimator} uses the full
sample for evaluating the AIPW score in \eqref{eq:cross-fitted-score}, which preserves the required sample separation without sacrificing first-order efficiency.

\section{Large Sample Theory}
\label{sec:large-sample}

This section establishes the large-sample properties of the bridge estimators
and the sequentially cross-fitted efficient estimator. We first derive
projected convergence rates for the primal and adjoint bridge estimators.
We then translate these projected rates into ordinary \(L_2(P_0)\) rates
through local measures of ill-posedness. Finally, we combine the projected
and ordinary bridge rates with the outer-stage nuisance rates to establish
asymptotic linearity and semiparametric efficiency of
\(\widehat\theta_{\mathrm{eff}}\) in \eqref{eq:final-estimator}.

The bridge-rate analysis builds on regularized minimax estimation for
conditional moment models
\citep{dikkala2020minimax,kallus2021causal,ghassami2022minimax}.
When the instrument propensity is unknown, however, the bridge criteria also depend on out-of-fold estimates of \(\pi_0\), \(\gamma_0\), and \(\alpha_0\). We show that errors in these preliminary nuisance estimates enter the projected bridge rates only through second-order product terms.

For a measurable function \(f\), write
\(\|f\|_p=(P_0|f|^p)^{1/p}\). Let
\(\vartheta_\mu,\vartheta_q\in(0,1)\) denote the interpolation exponents of the
primal and adjoint critic RKHSs, respectively. Thus,
\(\|g\|_\infty\lesssim \|g\|_{\mathcal G_\mu}^{\vartheta_\mu} \|g\|_2^{1-\vartheta_\mu}\) for \(g\in\mathcal G_\mu\), and
\(\|g\|_\infty\lesssim \|g\|_{\mathcal G_q}^{\vartheta_q} \|g\|_2^{1-\vartheta_q}\) for \(g\in\mathcal G_q\).
Such interpolation inequalities follow from spectral conditions on the kernel
eigenvalues and eigenfunctions and are standard in kernel analyses of
orthogonal learners, including the R-learner and proximal P-learner
\citep{nie2021quasi,sverdrup2023proximal}. 

Let \(\delta_{\mu,n}\) and \(\delta_{q,n}\) denote the high-probability
complexity radii for the primal and adjoint critic and coupled bridge--critic
classes, respectively (the precise definitions are given in
Appendix~\ref{subsec:appendix-projected-rates}). For \(r\in\{0,1\}\), define \(\bar\nu_{r,\ell}(x) = \E\!\big( \widehat H_\ell \mid R=r,X=x,\mathscr F_{\mathcal T_\ell} \big)\),
where \(\mathscr F_{\mathcal T_\ell}\) is the sigma-field generated by
the bridge-training sample. Because \(\widehat H_\ell\) depends on the
fold-specific bridge estimates, the relevant regression target for
\(\widehat\nu_{r,\ell}\) is \(\bar\nu_{r,\ell}\), rather than the conditional
mean induced by the true bridges. We impose the following convergence
conditions on the nuisance estimators used in the inner and outer stages.

\begin{assumption}[Nuisance convergence rates]
\label{ass:nuisance-rates}
The following conditions hold.

\begin{enumerate}[label=(\roman*)]

\item Uniformly in \(d\), \(\|\widehat\pi_d-\pi_0\|_2=O_p(\varepsilon_{\pi,n})\),
\(\|\widehat\gamma_d-\gamma_0\|_2=O_p(\varepsilon_{\gamma,n})\), and
\(\|\widehat\alpha_d-\alpha_0\|_2=O_p(\varepsilon_{\alpha,n})\), where \(\varepsilon_{\pi,n},\varepsilon_{\gamma,n}, \varepsilon_{\alpha,n}=o(1)\).

\item Uniformly in \(\ell\),
\(\|\widehat\pi_\ell-\pi_0\|_2=O_p(\varepsilon_{\pi,n})\) and, for each \(r\in\{0,1\}\),
\(\|\widehat\nu_{r,\ell}-\bar\nu_{r,\ell}\|_2 =O_p(\varepsilon_{\nu,r,n})\), where \(\varepsilon_{\nu,0,n},\varepsilon_{\nu,1,n}=o(1)\).

\item Uniformly in \(\ell\),
\(\|\widehat\mu_\ell-\mu_0\|_2=o_p(1)\) and \(\|\widehat q_\ell-q_0\|_2=o_p(1)\).

\end{enumerate}

\end{assumption}

The first two parts of Assumption~\ref{ass:nuisance-rates} distinguish the
inner-fold nuisance estimators used in the orthogonalized bridge criteria from
the outer-fold nuisance estimators entering the AIPW score. We use
\(\varepsilon_{\pi,n}\) as a common upper bound for the inner- and outer-fold propensity estimation rates, with \(\varepsilon_{\pi,n}=0\) when
\(\pi_0\) is known. The third part
imposes consistency of the bridge estimators themselves. We next sharpen this
consistency result in the projected norms induced by the bridge operators.

\begin{theorem}[Projected rates for the bridge estimators]
\label{thm:bridge-rates}
Suppose Assumption~\ref{ass:nuisance-rates} and
Assumptions~\ref{ass:large-sample-regularity}--%
\ref{ass:ls-penalty-calibration} in Appendix~\ref{app:large-sample} hold.
Uniformly in \(\ell\), the unknown-propensity bridge estimators satisfy
\begin{equation}
\begin{aligned}
\|T(\widehat\mu_\ell-\mu_0)\|_2
&=
O_p\!\left(
\delta_{\mu,n}
+
\left[
\varepsilon_{\pi,n}
(\varepsilon_{\gamma,n}+\varepsilon_{\alpha,n})
\right]^{1/(1+\vartheta_\mu)}
\right),\\
\|T^\ast(\widehat q_\ell-q_0)\|_2
&=
O_p\!\left(
\delta_{q,n}
+
\left[
\varepsilon_{\pi,n}
(\varepsilon_{\pi,n}+\varepsilon_{\alpha,n})
\right]^{1/(1+\vartheta_q)}
\right).
\end{aligned}
\label{eq:chi-rates}
\end{equation}
When \(\pi_0\) is known, the corresponding projected rates are
\(O_p(\delta_{\mu,n})\) and \(O_p(\delta_{q,n})\).
\end{theorem}

The rates in \eqref{eq:chi-rates} have two components. The complexity radii
\(\delta_{\mu,n}\) and \(\delta_{q,n}\) capture the sampling error of the
regularized minimax estimators. The remaining terms capture the effect of
preliminary nuisance estimation. By Proposition~\ref{prop:orthogonal-residualized-moments},
these perturbations enter at orders
\(\varepsilon_{\pi,n} (\varepsilon_{\gamma,n}+\varepsilon_{\alpha,n})\) and
\(\varepsilon_{\pi,n} (\varepsilon_{\pi,n}+\varepsilon_{\alpha,n})\), respectively.
The critic interpolation inequalities translate these
perturbations into the powered terms appearing in \eqref{eq:chi-rates}.

For compactness, we denote \(\chi_{\mu,n} = \delta_{\mu,n} + [ \varepsilon_{\pi,n} (\varepsilon_{\gamma,n}+\varepsilon_{\alpha,n}) ]^{1/(1+\vartheta_\mu)}\)
and \(\chi_{q,n} = \delta_{q,n} + [ \varepsilon_{\pi,n} (\varepsilon_{\pi,n}+\varepsilon_{\alpha,n}) ]^{1/(1+\vartheta_q)}\).
If
\(\varepsilon_{\pi,n} (\varepsilon_{\gamma,n}+\varepsilon_{\alpha,n}) = O(\delta_{\mu,n}^{1+\vartheta_\mu})\) and
\(\varepsilon_{\pi,n} (\varepsilon_{\pi,n}+\varepsilon_{\alpha,n}) = O(\delta_{q,n}^{1+\vartheta_q})\), then
\(\chi_{\mu,n}=O(\delta_{\mu,n})\) and
\(\chi_{q,n}=O(\delta_{q,n})\). Under these conditions, estimating the
instrument propensity does not change the oracle projected rates, which is known as the quasi-oracle property \citep{nie2021quasi}.

Projected convergence is the natural loss for conditional moment estimation,
but it does not by itself determine the ordinary \(L_2(P_0)\) convergence rate
of the bridge functions. We therefore use the bridge-error localization
condition in Assumption~\ref{ass:bridge-error-localization} of
Appendix~\ref{subsec:appendix-local-process}. Define \(\mathcal H_{\mu,B_\mu} = \{h\in\mathcal H_\mu: \|h\|_{\mathcal H_\mu}^2\le B_\mu\}\)
and \(\mathcal H_{q,B_q} = \{h\in\mathcal H_q: \|h\|_{\mathcal H_q}^2\le B_q\}\).
For \(\delta>0\), let
\(\mathcal H_{\mu,B_\mu}^{\mid\delta} = \{h\in\mathcal H_{\mu,B_\mu}:\|Th\|_2\le\delta\}\) and
\(\mathcal H_{q,B_q}^{\mid\delta} = \{h\in\mathcal H_{q,B_q}:\|T^\ast h\|_2\le\delta\}\), and define
\(\tau_\mu(\delta) = \sup_{h\in\mathcal H_{\mu,B_\mu}^{\mid\delta}}\|h\|_2\) and \(\tau_q(\delta) = \sup_{h\in\mathcal H_{q,B_q}^{\mid\delta}}\|h\|_2\).
These local measures of ill-posedness quantify the largest ordinary bridge
error compatible with a projected error of size \(\delta\) within the
localized RKHS classes.

Because the localized RKHS balls are star-shaped,
\(\tau_\mu(K\delta)\le K\tau_\mu(\delta)\) and
\(\tau_q(K\delta)\le K\tau_q(\delta)\) for every fixed \(K\ge1\).
Consequently, Theorem~\ref{thm:bridge-rates} and
Assumption~\ref{ass:bridge-error-localization} imply, uniformly in \(\ell\),
\begin{equation}
\|\widehat\mu_\ell-\mu_0\|_2
=
O_p\{\tau_\mu(\chi_{\mu,n})\},
\quad
\|\widehat q_\ell-q_0\|_2
=
O_p\{\tau_q(\chi_{q,n})\}.
\label{eq:tau-direct-rates}
\end{equation}

We now combine the projected and direct bridge-rate bounds with the outer-stage
nuisance rates to establish asymptotic linearity and semiparametric efficiency
of \(\widehat\theta_{\mathrm{eff}}\).

\begin{theorem}[Asymptotic linearity and efficiency]
\label{thm:final-asymptotics}
Suppose the conditions of Theorem~\ref{thm:bridge-rates} and
Assumption~\ref{ass:bridge-error-localization} hold.
Assume further that
\begin{equation}
\min\!\left\{
\chi_{\mu,n}\tau_q(\chi_{q,n}),
\chi_{q,n}\tau_\mu(\chi_{\mu,n})
\right\}
=
o(n^{-1/2})
\label{eq:mixed-rate-condition}
\end{equation}
and
\begin{equation}
\varepsilon_{\pi,n}
(\varepsilon_{\nu,0,n}+\varepsilon_{\nu,1,n})
=
o(n^{-1/2}).
\label{eq:outer-rate-condition}
\end{equation}
Then
\begin{equation}
\sqrt n(\widehat\theta_\mathrm{eff}-\theta_0)
=
\frac1{\sqrt n}\sum_{i=1}^n\phi_0(O_i)
+
o_p(1)
\xlongrightarrow{\cal D}
N(0,V_{\mathrm{eff}}),
\label{eq:asymptotic-linearity}
\end{equation}
where \(V_{\mathrm{eff}}=P_0\{\phi_0(O)^2\}\).
Moreover, \(\widehat V_{\mathrm{eff}} \xrightarrow{p} V_{\mathrm{eff}}\).
Since \(\phi_0\) is the EIF in Theorem~\ref{thm:eif},
\(\widehat\theta_{\mathrm{eff}}\) attains the semiparametric efficiency bound.

When \(\pi_0\) is known,
\(\varepsilon_{\pi,n}=0\), condition
\eqref{eq:outer-rate-condition} is automatically satisfied,
\(\chi_{\mu,n}=\delta_{\mu,n}\), and
\(\chi_{q,n}=\delta_{q,n}\). In this case,
\eqref{eq:mixed-rate-condition} becomes \(\min\!\left\{ \delta_{\mu,n}\tau_q(\delta_{q,n}), \delta_{q,n}\tau_\mu(\delta_{\mu,n}) \right\} = o(n^{-1/2})\).
\end{theorem}

Condition~\eqref{eq:mixed-rate-condition} controls the remainder arising from
bridge estimation by pairing the projected convergence rate of one bridge with
the ordinary \(L_2(P_0)\) rate of the other, as in minimax proximal estimation
\citep{ghassami2022minimax}. Condition~\eqref{eq:outer-rate-condition}
controls the outer AIPW remainder by pairing propensity error with
pseudo-outcome regression error. The sequential cross-fitting scheme separates
the three nuisance estimation stages, while orthogonality reduces the
corresponding first-order remainder terms to products of nuisance errors
\citep{chernozhukov2018double}.

\section{Numerical Experiments}
	\label{sec:simulation}

	\subsection{Data generating process}
	\label{sec:sim-setup}

    \subsubsection{Main simulation setup}
    \label{sec:main-sim-setup}

We consider a randomized binary instrument with two-sided noncompliance, i.e., allowing both always-takers and never-takers. A scalar covariate \(X\) is drawn from \(N(0,1)\). Conditional on \(X\), the treatment-side proxy \(Z\), the outcome-side proxy \(W\), and the latent confounder \(U\) are simulated from a multivariate Gaussian distribution, 
\[
	(Z,W,U)\mid X
	\sim
	\mathrm{MVN}
	\left(
	\begin{pmatrix}
		0.25 X \\
		0.25 X \\
		0
	\end{pmatrix},
	\Sigma
	=
	\begin{pmatrix}
		1 & 0.25 & 0.5 \\
		0.25 & 1 & 0.5 \\
		0.5 & 0.5 & 1
	\end{pmatrix}
	\right).
\]
The implied regressions are \(\E(Z\mid U,X)=0.5 U+0.25 X\) and \(\E(W\mid U,X)=0.5 U+0.25 X\). 

The binary instrument is generated from a logistic model
\begin{equation}
\begin{aligned}
R\mid X
&\sim
\mathrm{Bernoulli}\{\pi_0(X)\},\\
\pi_0(X)
&=
0.1+0.8\,\mathrm{expit}(0.5X),
\end{aligned}
\label{eq:baseline-propensity}
\end{equation}
so that \(\pi_0(X)\in[0.1,0.9]\) uniformly in \(X\). 

Treatment \(A\) follows a principal-stratum model with no defiers \citep{imbens1994identification}. We parameterize the strata by the complier probability \(c(U,X)\) and, conditional on being a noncomplier, the always-taker probability \(a(U,X)\): 
\begin{equation*}
\begin{aligned}
c(U,X)
=
\Pr\bigl(A(1)>A(0)\mid U,X\bigr)
&=
\mathrm{expit}(0.25+0.5U+0.5X),
\\
a(U,X)
=
\Pr\bigl(A(0)=1\mid A(1)=A(0),U,X\bigr)
&=
\mathrm{expit}(-0.5+0.25U+0.125X).
\end{aligned}
\end{equation*} 
The resulting treatment probabilities under \(R=0\) and \(R=1\) are \(p_0(U,X)=a(U,X)\{1-c(U,X)\}\) and \(p_1(U,X)=p_0(U,X)+c(U,X)\), respectively. Principal strata are realized through a latent Gaussian shock \(\varepsilon_A\), with \(A(r)=I[\varepsilon_A\leq \Phi^{-1}\{p_r(U,X)\}]\), for \(r\in\{0,1\}\), and observed treatment \(A=A(R)\). Thus, a unit is an always-taker when \(\varepsilon_A\leq\Phi^{-1}(p_0)\), a complier when \(\Phi^{-1}(p_0)<\varepsilon_A\leq\Phi^{-1}(p_1)\), and a never-taker when \(\varepsilon_A>\Phi^{-1}(p_1)\).

The treatment effect heterogeneity and the observed outcome are specified as
	\[
	\tau(U,X)=2+U+X,
	\quad
	Y=-0.5 U+X+A\,\tau(U,X)-0.25 Z+\varepsilon_Y.
	\]
The linear term \(-0.25 Z\) in \(Y\) makes \(Z\) an invalid negative control treatment. We also allow additional treatment endogeneity by drawing \((\varepsilon_A,\varepsilon_Y)\) from a bivariate standard Gaussian distribution with correlation \(0.5\). 

The implied ground truth primal and adjoint bridges are
	\begin{equation*}
		\begin{aligned}
			\mu_0(W,X)
			&=
			2+2W+0.5 X,
			\\
			q_0(Z,X)
			&=
			1+\exp(-0.625-Z-0.25 X).
		\end{aligned}
	\end{equation*}
They satisfy \(\E\{\mu_0(W,X)\mid U,X\}=\tau(U,X)\) and \(\E\{q_0(Z,X)\mid U,X\}=1/c(U,X)\). The population ATE is \(\theta_0=2\). Because compliance increases with both \(U\) and \(X\), while the treatment effect also increases with these variables, compliers have larger average treatment effects and the LATE is approximately \(2.40\). In addition, both the latent variable \(U\) and the correlation between \(\varepsilon_A\) and \(\varepsilon_Y\) induce confounding bias that remains after adjustment for the observed covariates \(X\).

\subsubsection{Complex instrument propensity}
\label{sec:complex-propensity}

To assess the effect of orthogonalizing the bridge moments when the instrument
propensity is difficult to estimate, we consider a variant of the main simulation setup in which \(\pi_0\) depends nonparametrically on a \(p\)-dimensional covariate. Write \(X=(X_1,\ldots,X_p)\), where \(X_1,\ldots,X_p\) are mutually
independent standard Gaussian variables. The additional coordinates
\(X_2,\ldots,X_p\) enter only the instrument propensity. Varying \(p\) therefore changes the difficulty of propensity estimation while leaving the remaining structural equations unchanged, thereby isolating the effect of propensity complexity and orthogonalization. In particular, \(X_2,\ldots,X_p\) do not enter the laws of \(U\), \(Z\), \(W\), \(A(r)\), or \(Y(a)\), which otherwise remain as specified in Section~\ref{sec:main-sim-setup}.

To generate a nonlinear instrument propensity, we draw a random function
\(f_p:\mathbb R^p\to\mathbb R\) from a zero-mean Gaussian process, \(f_p \sim \operatorname{GP}(0,k_p)\).
The covariance function \(k_p\) is the Matérn kernel with smoothness
\(\nu=1.5\), unit marginal variance, and unit length scale
\citep{williams2006gaussian}:
\[
k_p(x,x')
=
\left(1+\sqrt{3}\lVert x-x'\rVert_2\right)
\exp\left\{-\sqrt{3}\lVert x-x'\rVert_2\right\},
\quad x,x'\in\mathbb R^p.
\]
Because the smoothness and length scale are held fixed, increasing \(p\) reduces correlation
across covariate points and makes the propensity more difficult to estimate. For comparison, the \(L_2\) minimax rate for nonparametric
estimation of a \(\beta\)-H\"older regression function in \(p\) dimensions is
\(n^{-\beta/(2\beta+p)}\) \citep{stone1982optimal}, which is also used by \citet{farrell2021deep} to characterize neural network nuisance
estimation.

For computational tractability, we approximate each Gaussian process realization using \(m=256\) random Fourier features \citep{rahimi2007random}:
\[
f_p(x)
=
\sqrt{\frac{2}{m}}
\sum_{j=1}^{m}
\cos(\omega_j^\top x+b_j).
\]
The frequencies \(\omega_j\) are drawn in dimension \(32\) from a
multivariate \(t\) distribution with \(2\nu\) degrees of freedom and unit scale, and are then truncated to their first \(p\)
coordinates. The phases \(b_j\) are independent and uniformly distributed on \([0,2\pi]\).

We set
\begin{equation}
\pi_0(X)
=
0.1+0.8\,\mathrm{expit}\{0.5f_p(X)\}
\label{eq:complex-propensity}
\end{equation}
and draw \(R\mid X\sim\mathrm{Bernoulli}\{\pi_0(X)\}\).

\subsection{Estimators}

We implement two proposed proximal IV estimators. The
\(\mathrm{PIV}_{\pi_0}\) estimator uses the known instrument propensity,
whereas \(\mathrm{PIV}_{\mathrm{orth}}\) estimates the propensity and
uses the orthogonalized bridge moments. Both procedures estimate the primal and adjoint bridges by minimax learning in Gaussian RKHSs \citep{dikkala2020minimax,kallus2021causal,ghassami2022minimax}. For
computational scalability, we approximate the bridge and critic RKHSs by Nystr\"om subspaces generated from landmark observations drawn from the bridge-training sample \citep{williams2000using}. When the propensity is estimated, we use a neural network and truncate \(\widehat\pi_\ell(X)\) to \([0.01,0.99]\) to reduce instability from extreme estimated propensity scores. For the pseudo-outcome regression, we fit a single neural network \(\widehat\nu_\ell(r,x)\) on \(\mathcal J_\ell\) by regressing \(\widehat H_{i,\ell}\) on \((R_i,X_i)\), and define \(\widehat\nu_{r,\ell}(x)=\widehat\nu_\ell(r,x)\) for \(r\in\{0,1\}\). We use this S-learner rather than separate arm-specific regressions (T-learner) because \(\nu_1^\star=\nu_0^\star=\omega_0\), so pooling information across instrument arms is natural and can improve stability when one arm is sparse over parts of the covariate space \citep{kunzel2019metalearners}.

We compare with four existing methods: a naive AIPW estimator that treats \(A\) as unconfounded given \(X\) \citep{bang2005doubly}, a cross-fitted AIPW Wald estimator of
the complier average effect
\citep{imbens1994identification,chernozhukov2018double}, the
inverse-compliance-score weighted (ICSW) estimator of
\citet{aronow2013beyond}, with \((X,Z,W)\) used as observed predictors
of compliance, and the kernel proximal causal learning (PCL) estimator of
\citet{ghassami2022minimax}. Confidence intervals are based on estimated influence functions for all
estimators except ICSW, for which we use the bootstrap procedure of
\citet{aronow2013beyond} (implementation details are given in Appendix~\ref{sec:app-experiments}). For the cross-fitted estimators, we use
\(L=5\) folds. We evaluate each estimator over 100 Monte Carlo replications at sample sizes \(n\in\{200,500,1{,}000,2{,}000,5{,}000\}\).

\subsection{Simulation results}

    \begin{figure}[t]
		\centering
		\includegraphics[width=\textwidth]{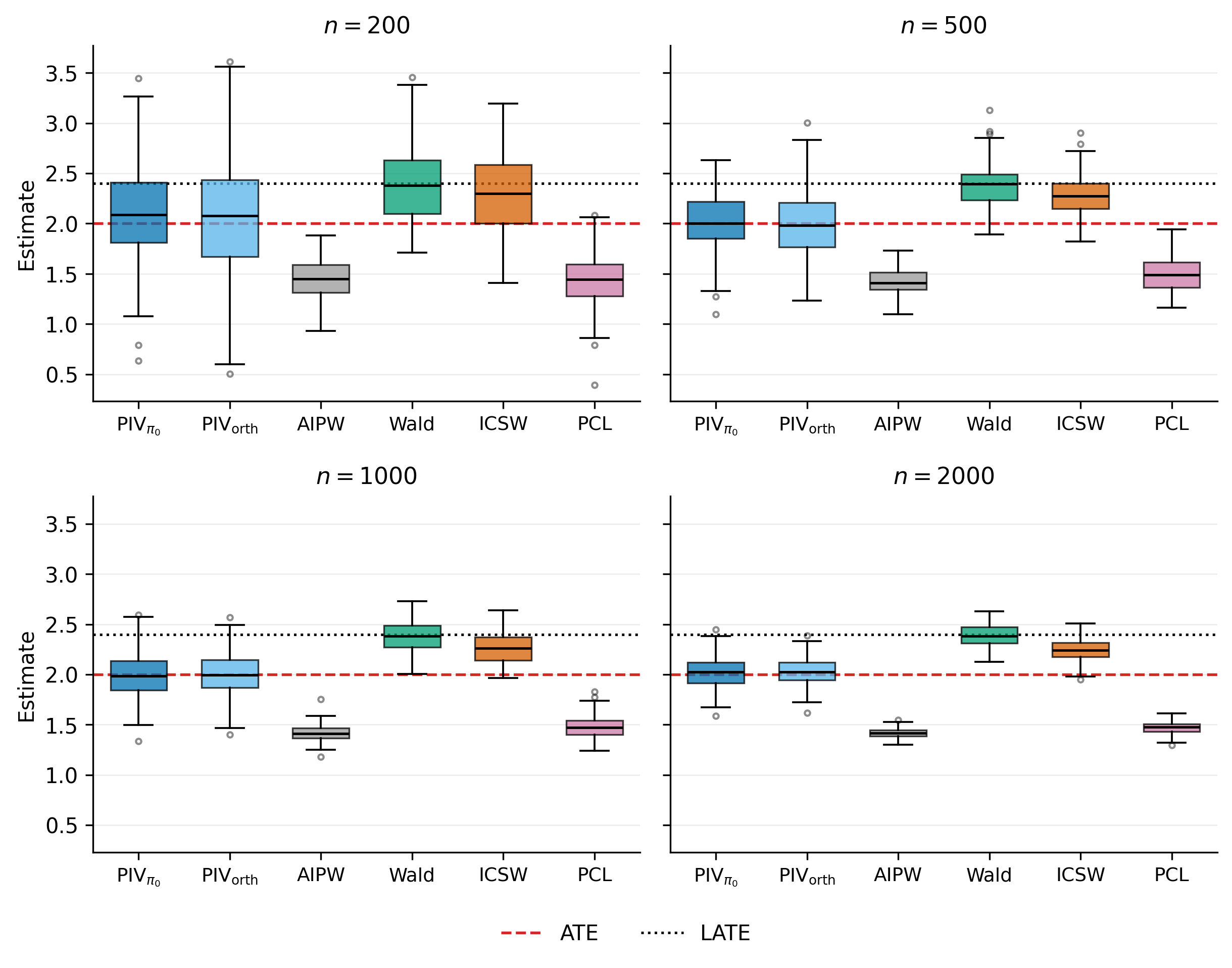}
		\caption{Monte Carlo estimates of the ATE.
			Boxes are the empirical distribution over 100 replications.
			The red dashed line is the ATE \(\theta_0=2\);
			the dotted line is the LATE.}
		\label{fig:sim-ate}
	\end{figure}

    \begin{figure}[t]
		\centering
		\includegraphics[width=\textwidth]{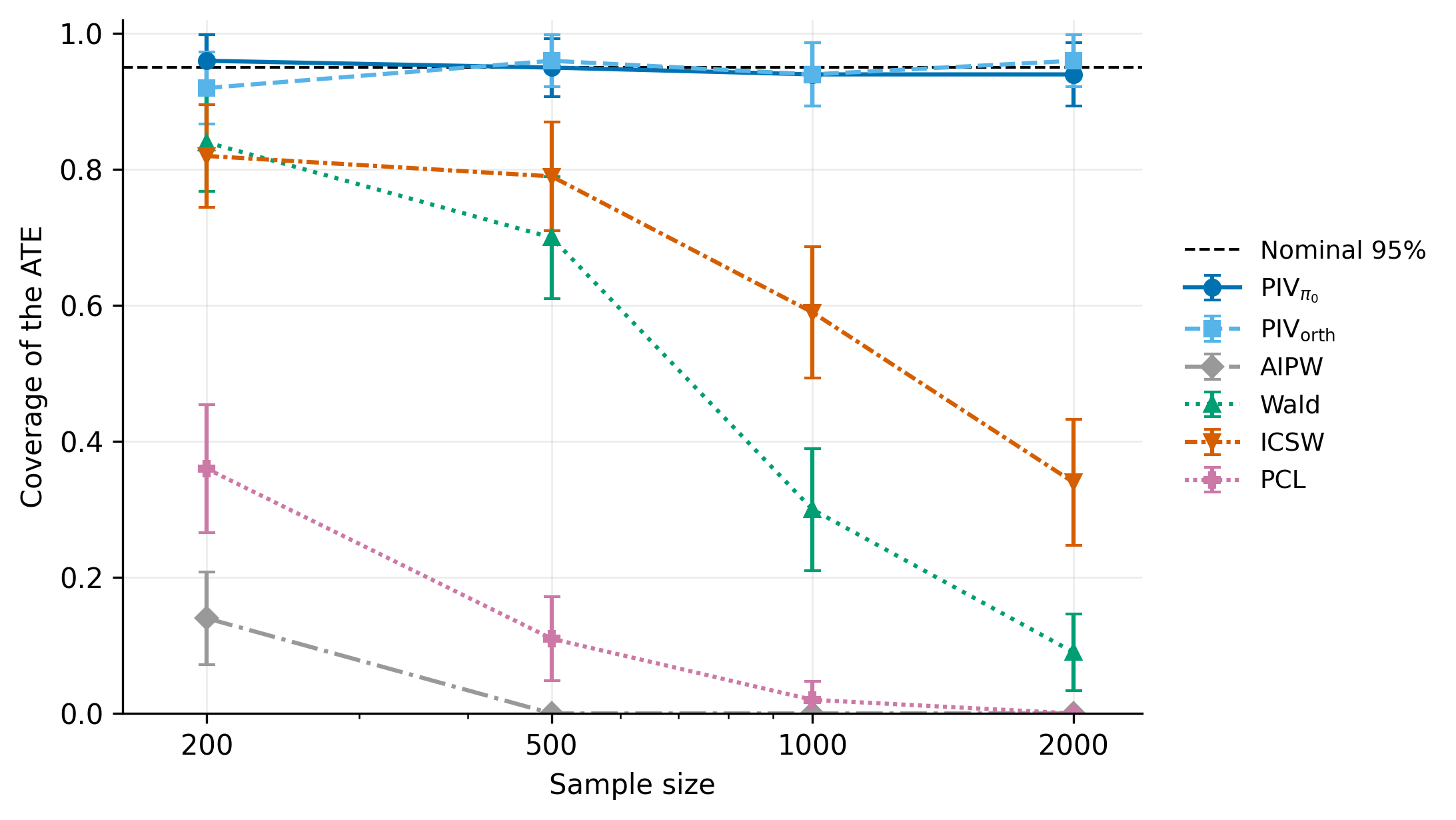}
		\caption{Empirical coverage of nominal \(95\%\) confidence intervals for the
			ATE, with Monte Carlo \(95\%\) error bars. The dashed line is the nominal
			level. 
            }
		\label{fig:sim-coverage}
	\end{figure}
    
Figures~\ref{fig:sim-ate} and~\ref{fig:sim-coverage} report the Monte Carlo
distributions of the estimates and the coverage of nominal \(95\%\) confidence intervals under the instrument propensity in \eqref{eq:baseline-propensity}. Both \(\mathrm{PIV}_{\pi_0}\) and
\(\mathrm{PIV}_{\mathrm{orth}}\) are centered near the ATE at every sample size considered. The AIPW Wald estimator is instead centered near the LATE, as expected because it targets the complier effect rather than transporting that effect to the population. The ICSW estimator, which reweights the complier effect using an estimated compliance score based on \((X,Z,W)\), moves only part of the way toward \(\theta_0\). This indicates that using the two proxies as ordinary observed predictors of compliance does not fully adjust for the relevant latent heterogeneity. The AIPW and PCL estimators are both centered below \(\theta_0\). The former does not adjust for the latent confounder \(U\), while the latter cannot accommodate either the residual endogeneity induced by the dependence between \(\varepsilon_A\) and \(\varepsilon_Y\) or the failure of \(Z\) to satisfy the negative control restriction. Consistent with these patterns, both proximal IV confidence intervals attain coverage near the nominal level, whereas coverage for the
alternative estimators deteriorates as their intervals concentrate around
functionals other than \(\theta_0\). Coverage for AIPW and PCL falls well
below \(95\%\) by \(n=500\); Wald and ICSW remain closer to the nominal level
at \(n=200\) but decline substantially as \(n\) increases.

At \(n=200\), the proximal IV estimators exhibit the largest Monte Carlo
dispersion among the methods considered. This loss of finite-sample
precision has two sources. First, as in conventional IV estimation,
identification uses only treatment variation induced by the instrument, so
limited compliance weakens the first stage and increases sampling variability
\citep{staiger1997instrumental,hartman2024improving}. Second, estimating the primal and adjoint bridges requires solving ill-posed conditional moment problems, which introduces additional variation, particularly when the proxies are weak
\citep{wang2025weakgmm}. The alternative estimators
lack at least one of these sources: AIPW and PCL do not use the instrument,
while the Wald and ICSW estimators do not solve the nonparametric bridge
equations. Their lower dispersion, however, is accompanied by
bias for the ATE functional. As \(n\) increases, the dispersion of the proximal IV estimators decreases substantially, and by \(n=2{,}000\) they attain the lowest MSE among the six procedures considered.

    \begin{figure}
    \centering
    \includegraphics[width=1\linewidth]{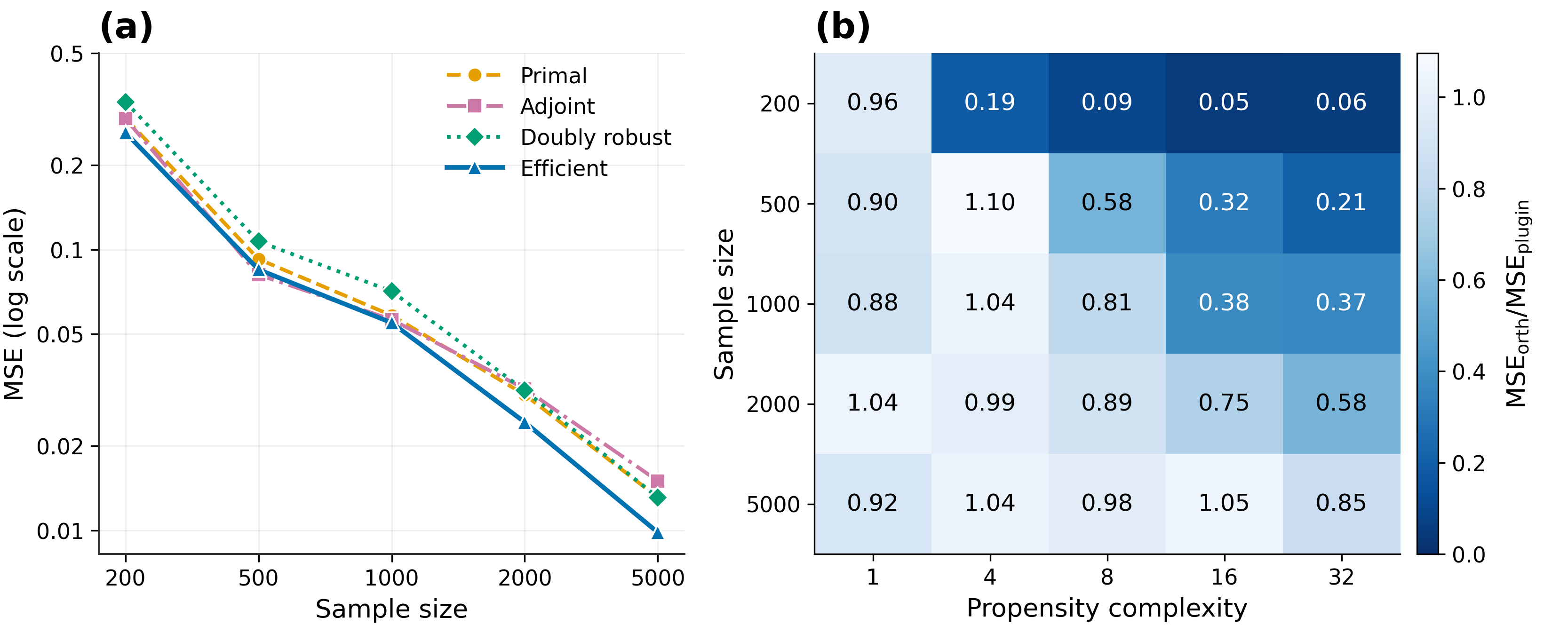}
    \caption{(a)~Monte Carlo MSE of the four proximal IV estimators. (b)~MSE ratio of the orthogonal to the plug-in implementation of the efficient estimator, by sample size and propensity complexity.}
    \label{fig:sim-four-score}
    \end{figure}

Figure~\ref{fig:sim-four-score} compares the finite-sample performance of the proximal IV estimators. Panel~(a) compares the primal, adjoint, doubly robust, and efficient estimators under the known instrument propensity in \eqref{eq:baseline-propensity}. All four estimators are centered near \(\theta_0\), so their MSE differences primarily reflect differences in variance. The efficient estimator has the lowest MSE at most sample sizes, with the advantage becoming clear at \(n=2{,}000\) and \(n=5{,}000\). By contrast, the doubly robust estimator does not improve on the finite-sample efficiency of the primal and adjoint estimators. Panel~(b) evaluates the effect of orthogonalizing the bridge moments under the higher-dimensional propensity in \eqref{eq:complex-propensity}. It reports the ratio of the MSE of the orthogonal implementation to that of the plug-in implementation. Both procedures use the efficient estimator and the same neural network specification for \(\widehat\pi\), differing only in whether the bridge estimators are constructed from plug-in or orthogonalized moments. The MSE gain from orthogonalization is largest when the sample size is small and the propensity is complex. As \(n\) increases, propensity estimation improves and the difference between the two procedures narrows, with nearly identical performance by \(n=5{,}000\).

\section{An application to door-to-door canvassing}
\label{sec:application}

We apply the proximal instrumental variable estimator to the six-city
door-to-door canvassing experiments of \citet{green2003getting}
(Yale {ISPS} Data Archive {D017}), fielded before the local elections of
6~November 2001, which were analyzed as an
encouragement design by \citet{aronow2013beyond}. Registered voters
were randomly assigned to a nonpartisan face-to-face appeal or to no
canvassing attempt. Names at the same address were grouped into
households, and one registered voter from each household was selected for
the study. The pooled sample comprises \(N=18{,}933\) household
records in Bridgeport, Columbus, Detroit, Minneapolis, Raleigh and
St.\ Paul. Since the turnout in 1999 is unrecorded for all Bridgeport
observations, our analysis uses the \(n=17{,}127\) complete records in the
remaining five cities. 

Let \(R\) denote assignment to canvassing and \(Y\) indicate
whether the selected individual voted in the November 2001 election. The treatment \(A\) is successful canvasser contact, defined as a nonpartisan conversation encouraging the named registrant or a voting-age housemate to vote. Canvassers attempted contact only in the assigned arm, so \(A=0\) whenever \(R=0\). Thus, compliance is one-sided: individuals are either compliers or never-takers, with no always-takers or defiers, and monotonicity holds by
design. For identification of the causal effect using \(R\), we
additionally assume the exclusion restriction that assignment affects turnout
only through recorded successful contact.

The households were organized geographically into walk lists. Let \(S\) denote
the city-by-walk-list randomization cell. The randomization procedure varied across
the experimental sites and, in Raleigh, across walk lists
\citep{green2003getting}. Following \citet{aronow2013beyond}, we account
for this design using inverse probability weighting
based on the known cell-level
assignment probabilities \(\pi_0(S)\) \citep{horvitz1952generalization}. Let \(n(s)\) denote the number of study households in cell \(s\) and \(n_1(s)\) the number assigned to canvassing. Under complete randomization of households within cell \(s\), the known design propensity is
\(\pi_0(s)=\frac{n_1(s)}{n(s)}\). The randomization cell \(S\) enters the design
weights only and is not included in the baseline vector \(X\) used to
estimate the remaining nuisance functions.

Under the standard instrumental variable conditions, the Wald estimand identifies
the average effect among compliers, namely, study voters whose household would
be successfully contacted if assigned to canvassing
\citep{imbens1994identification,angrist1996identification}. However, contactability of the household may depend on characteristics that are also associated with heterogeneity in the effect of contact on turnout, such as being at home and willing to engage with canvassers. Consequently, the complier average effect may not coincide with the
corresponding population average effect.

We estimate the population effect by the known-propensity proximal
instrumental variable estimator, using two historical turnout records as
complementary proxy measurements. The baseline covariates \(X\) comprise city, age, household
size, race, registered party, sex and the associated missingness
indicators. Both turnout measures predate the 2001 campaign and are therefore
pre-treatment variables. We take \(Z\) to be turnout
in the 1999 local election, which provides information about contactability and
enters the compliance bridge, and \(W\) to be turnout in the 2000
presidential election, which is predictive of subsequent turnout and enters the
outcome bridge.

The two historical turnout measures are natural proxies for latent civic
engagement because prior participation may reflect persistent features related
both to contactability and to heterogeneity in the effect of canvassing. However,
each record is only a binary and noisy measurement of this underlying
trait, so treating the two indicators as ordinary adjustment variables may not
fully account for the relevant heterogeneity. Our proximal IV approach therefore
uses the two records as proxies for the latent factor, without assuming that
contact is unconfounded or that compliance is conditionally independent of
treatment effect heterogeneity given the observed covariates. By contrast, the
ICSW estimator of \citet{aronow2013beyond} uses the same turnout measures as
observed predictors of compliance and therefore requires the measured covariates
to be sufficiently informative about the heterogeneity needed to transport the
complier effect to the population.

We report the proximal IV (with known propensity), AIPW Wald, ICSW, PCL, and naive AIPW estimators. In this application, the proximal bridge functions are learned with a generalized
product kernel for mixed-type data, combining Gaussian kernels for continuous
covariates with Aitchison--Aitken kernels for categorical and binary variables
\citep{aitchison1976multivariate,racine2004nonparametric}. Hyperparameters are selected on
held-out validation splits by minimizing residual conditional moment risk
\citep{mastouri2021proximal}. Further details are provided in
Appendices~\ref{sec:app-estimators}--\ref{sec:app-tuning}.

	\begin{table}[t]
    \centering
    \caption{Estimates (standard errors) and \(95\%\) confidence intervals
        for the effect of successful canvasser contact on November 2001 turnout.}
    \label{tab:d017}
    \small
    \vspace{0.5em}
    \setlength{\tabcolsep}{4.5pt}
    \begin{tabular}{lccccc}
        \toprule
        & Proximal IV & ICSW & AIPW Wald & PCL & Naive AIPW \\
        \midrule
        Estimates (SEs)
        & \(0.089\) (\(0.024\))
        & \(0.063\) (\(0.019\))
        & \(0.064\) (\(0.022\))
        & \(0.047\) (\(0.009\))
        & \(0.063\) (\(0.009\)) \\
        \(95\%\) CIs
        & \((0.041,\ 0.137)\)
        & \((0.026,\ 0.100)\)
        & \((0.021,\ 0.107)\)
        & \((0.029,\ 0.065)\)
        & \((0.046,\ 0.080)\) \\
        \bottomrule
    \end{tabular}
\end{table}

Table~\ref{tab:d017} shows that the AIPW Wald and ICSW estimates are nearly
identical, \(0.064\) and \(0.063\), respectively. This suggests that reweighting by the estimated compliance score produces little change relative to the complier average effect. In this application, treating the two historical turnout records as ordinary adjustment variables may therefore be insufficient to capture the relevant heterogeneity in compliance and treatment effects. By contrast, the proximal IV analysis treats
the two records as proxies for latent heterogeneity, yielding a higher estimate
of the population average effect \(0.089\). The PCL estimate is lower, at \(0.047\), but may be unreliable because it uses the turnout records as proxies without exploiting the randomized instrument and is therefore sensitive to residual endogeneity or violations of the negative control restrictions. Because the confidence intervals overlap, however, these differences should be interpreted cautiously rather than as evidence of a statistically significant gap between the population and complier average effects. One possible explanation for a higher population effect is that compliers are more
likely to be prior voters and major-party registrants
\citep{aronow2013beyond}. Previous field experiments have found somewhat smaller
canvassing effects among major-party voters than among unaffiliated voters
\citep{gerber1999does}, which could make the population average effect exceed
the complier average effect. The proximal IV analysis thus aligns with this pattern,
subject to the identifying assumptions of the proposed framework.

\section{Discussion}
\label{sec:discussion}

In this paper, we develop a proximal instrumental variable approach to the population ATE when a binary instrument identifies only a local effect among compliers under the standard IV conditions. Recovering a population parameter from that local contrast typically requires that measured covariates sufficiently capture the relevant differences between compliers and noncompliers \citep{aronow2013beyond}, or that treatment effect heterogeneity and instrument responsiveness do not covary \citep{wang2018bounded}. Proximal causal methods instead use complementary proxies to adjust for unmeasured confounding in the absence of an instrument \citep{cui2024semiparametric}. The proximal IV approach combines instrumental variable and proximal causal
frameworks, leveraging proxies for the latent heterogeneity factors to recover the population ATE from the local instrument-based
effect. The resulting primal and adjoint observed-data bridges each identify the population effect, and their combination is doubly robust. The proposed framework therefore provides a practical solution to the local-to-population extrapolation problem when treatment effects and compliance are simultaneously heterogeneous.

Simulation studies show that the proximal IV estimators are centered at the ATE. The Wald estimator is centered near the LATE,
whereas ICSW moves toward but remains biased for the ATE. The estimators that do not use the instrument, including the AIPW estimator on measured covariates and the proximal estimator, are biased for both functionals. In the voter-mobilization application, the proximal IV analysis yields a larger estimated effect than the AIPW Wald and inverse-compliance-score weighted estimators, potentially revealing population-level treatment effect heterogeneity that is not captured when the two historical records are treated
as ordinary measured covariates.

Our analysis has several limitations. First, identification relies on conditions that are difficult to verify in the absence of subject-matter knowledge. For example, the analysis requires the existence of treatment- and outcome-side proxies together with the conditional separation and latent bridge restrictions; violations of these restrictions can bias the proximal functional even when the instrument is valid. Second, weak proxies and severe ill-posedness may substantially hinder bridge estimation and thereby compromise \(\sqrt{n}\)-rate estimation. Also, we have focused on a binary instrument and a binary treatment, so extensions to multivalued instruments and treatments are natural directions for future work.

\section*{Acknowledgment}
OpenAI's GPT-5.6 Sol was used to assist with language editing and polishing of the manuscript, and writing code used in the numerical analyses. The core methodology and mathematical arguments were developed and verified by the authors, who take full responsibility for the content of this paper.

\bibliographystyle{apalike}
\bibliography{references}

\clearpage
\appendix

\input{appendix}

\end{document}

%% file: appendix.tex
\section{Nonparametric Identification}

\subsection{Proof of Proposition~\ref{thm:observed-bridges}}

\begin{proof}[Proof of Proposition~\ref{thm:observed-bridges}.]
Assumption~\ref{ass:iv} implies
\(
P(R=1\mid U,Z,W,X)=\pi_0(X) 
\) and
\(
P(R=0\mid U,Z,W,X)=1-\pi_0(X)
\).
Since \(R\) is conditionally randomized,
\[
\{Y(0),Y(1),A(0),A(1)\}\indep R\mid U,Z,W,X.
\]

We first show the following identity. For any integrable random variable \(V\),
\begin{equation}
\begin{aligned}
&\E\{\rho_0(R,X)V\mid U,Z,W,X\}
\\
={}&
\E\!\left[
\left\{
\frac{R}{\pi_0(X)}-
\frac{1-R}{1-\pi_0(X)}
\right\}V
\,\middle|\,U,Z,W,X
\right]
\\
={}&
\frac{\E(RV\mid U,Z,W,X)}{\pi_0(X)}
-
\frac{\E\{(1-R)V\mid U,Z,W,X\}}{1-\pi_0(X)}
\\
={}&
\frac{P(R=1\mid U,Z,W,X)}{\pi_0(X)}
\E(V\mid R=1,U,Z,W,X)\\
&-
\frac{P(R=0\mid U,Z,W,X)}{1-\pi_0(X)}
\E(V\mid R=0,U,Z,W,X)
\\
={}&
\E(V\mid R=1,U,Z,W,X)
-
\E(V\mid R=0,U,Z,W,X).
\end{aligned}
\label{eq:iv-weighting-identity}
\end{equation}

Take \(V=Y\) in \eqref{eq:iv-weighting-identity}. By consistency and exclusion,
\(Y=Y\{A(R)\}\). Conditional instrument independence, monotonicity, and
Assumption~\ref{ass:latent} then give
\[
\begin{aligned}
&\E\{\rho_0(R,X)Y\mid U,Z,W,X\}
\\
={}&
\E(Y\mid R=1,U,Z,W,X)
-
\E(Y\mid R=0,U,Z,W,X)
\\
={}&
\E\{Y(A(1))\mid R=1,U,Z,W,X\}
-
\E\{Y(A(0))\mid R=0,U,Z,W,X\}
\\
={}&
\E\{Y(A(1))\mid U,Z,W,X\}
-
\E\{Y(A(0))\mid U,Z,W,X\}
\\
={}&
\E\!\left[
Y(A(1))-Y(A(0))
\,\middle|\,U,Z,W,X
\right]
\\
={}&
\E\!\left[
\{A(1)-A(0)\}\{Y(1)-Y(0)\}
\,\middle|\,U,Z,W,X
\right]
\\
={}&
\E\!\left[
C\{Y(1)-Y(0)\}
\,\middle|\,U,Z,W,X
\right]
\\
={}&
\E\!\left[
C\,\E\{Y(1)-Y(0)\mid C,U,Z,W,X\}
\,\middle|\,U,Z,W,X
\right]
\\
={}&
\tau(U,X)\E(C\mid U,Z,W,X)
\\
={}&
 c(U,X)\tau(U,X),
\end{aligned}
\]
where we use the fact that
\[Y\{A(r)\}=Y(0)+A(r)\{Y(1)-Y(0)\}.\]
Consequently,
\begin{equation}
\begin{aligned}
&\E\{\rho_0(R,X)Y\mid U,Z,X\}
\\
={}&
\E\!\left[
\E\{\rho_0(R,X)Y\mid U,Z,W,X\}
\,\middle|\,U,Z,X
\right]
\\
={}&
\E\{c(U,X)\tau(U,X)\mid U,Z,X\}
\\
={}&
 c(U,X)\tau(U,X).
\end{aligned}
\label{eq:app-latent-outcome-contrast-uzx}
\end{equation}

Next take \(V=A\mu_0(W,X)\) in \eqref{eq:iv-weighting-identity}. Then,
\[
\begin{aligned}
&\E\{\rho_0(R,X)A\mu_0(W,X)\mid U,Z,W,X\}
\\
={}&
\E\{A\mu_0(W,X)\mid R=1,U,Z,W,X\}
-
\E\{A\mu_0(W,X)\mid R=0,U,Z,W,X\}
\\
={}&
\mu_0(W,X)
\left[
\E(A\mid R=1,U,Z,W,X)
-
\E(A\mid R=0,U,Z,W,X)
\right]
\\
={}&
\mu_0(W,X)
\left[
\E\{A(1)\mid R=1,U,Z,W,X\}
-
\E\{A(0)\mid R=0,U,Z,W,X\}
\right]
\\
={}&
\mu_0(W,X)
\left[
\E\{A(1)\mid U,Z,W,X\}
-
\E\{A(0)\mid U,Z,W,X\}
\right]
\\
={}&
\mu_0(W,X)\E\{A(1)-A(0)\mid U,Z,W,X\}
\\
={}&
\mu_0(W,X)\E(C\mid U,Z,W,X)
\\
={}&
 c(U,X)\mu_0(W,X).
\end{aligned}
\]
Conditioning on \((U,Z,X)\) and applying \(Z\indep W\mid U,X\) gives
\begin{equation}
\begin{aligned}
&\E\{\rho_0(R,X)A\mu_0(W,X)\mid U,Z,X\}
\\
={}&
\E\!\left[
\E\{\rho_0(R,X)A\mu_0(W,X)\mid U,Z,W,X\}
\,\middle|\,U,Z,X
\right]
\\
={}&
\E\{c(U,X)\mu_0(W,X)\mid U,Z,X\}
\\
={}&
 c(U,X)\E\{\mu_0(W,X)\mid U,Z,X\}
\\
={}&
 c(U,X)\E\{\mu_0(W,X)\mid U,X\}
\\
={}&
 c(U,X)\tau(U,X).
\end{aligned}
\label{eq:app-latent-primal-contrast-clean}
\end{equation}
Comparing \eqref{eq:app-latent-outcome-contrast-uzx} and
\eqref{eq:app-latent-primal-contrast-clean}, and then conditioning on
\((Z,X)\) gives
\[
\begin{aligned}
(T\mu_0)(Z,X)
&=
\E\{\rho_0(R,X)A\mu_0(W,X)\mid Z,X\}
\\
&=
\E\!\left[
\E\{\rho_0(R,X)A\mu_0(W,X)\mid U,Z,X\}
\,\middle|\,Z,X
\right]
\\
&=
\E\{c(U,X)\tau(U,X)\mid Z,X\}
\\
&=
\E\!\left[
\E\{\rho_0(R,X)Y\mid U,Z,X\}
\,\middle|\,Z,X
\right]
\\
&=
\E\{\rho_0(R,X)Y\mid Z,X\}=b(Z,X).
\end{aligned}
\]
Thus \(T\mu_0=b\), or equivalently
\[
\E\!\left[
\rho_0(R,X)\{Y-A\mu_0(W,X)\}
\mid Z,X
\right]=0
\quad\text{almost surely}.
\]

Finally take \(V=Aq_0(Z,X)\) in \eqref{eq:iv-weighting-identity}. Then,
\[
\begin{aligned}
&\E\{\rho_0(R,X)Aq_0(Z,X)\mid U,Z,W,X\}
\\
={}&
\E\{Aq_0(Z,X)\mid R=1,U,Z,W,X\}
-
\E\{Aq_0(Z,X)\mid R=0,U,Z,W,X\}
\\
={}&
q_0(Z,X)
\left[
\E(A\mid R=1,U,Z,W,X)
-
\E(A\mid R=0,U,Z,W,X)
\right]
\\
={}&
q_0(Z,X)
\left[
\E\{A(1)\mid U,Z,W,X\}
-
\E\{A(0)\mid U,Z,W,X\}
\right]
\\
={}&
q_0(Z,X)\E\{A(1)-A(0)\mid U,Z,W,X\}
\\
={}&
q_0(Z,X)\E(C\mid U,Z,W,X)
\\
={}&
 c(U,X)q_0(Z,X).
\end{aligned}
\]
Conditioning on \((U,W,X)\) and applying \(Z\indep W\mid U,X\) gives
\[
\begin{aligned}
&\E\{\rho_0(R,X)Aq_0(Z,X)\mid U,W,X\}
\\
={}&
\E\!\left[
\E\{\rho_0(R,X)Aq_0(Z,X)\mid U,Z,W,X\}
\,\middle|\,U,W,X
\right]
\\
={}&
 c(U,X)\E\{q_0(Z,X)\mid U,W,X\}
\\
={}&
 c(U,X)\E\{q_0(Z,X)\mid U,X\}
\\
={}&
 c(U,X)\frac{1}{c(U,X)}=1.
\end{aligned}
\]
A final application of iterated expectation gives
\[
\begin{aligned}
(T^\ast q_0)(W,X)
&=
\E\{\rho_0(R,X)Aq_0(Z,X)\mid W,X\}
\\
&=
\E\!\left[
\E\{\rho_0(R,X)Aq_0(Z,X)\mid U,W,X\}
\,\middle|\,W,X
\right]
\\
&=
\E(1\mid W,X)=1.
\end{aligned}
\]
Therefore \(T^\ast q_0=1\), equivalently
\[
\E\!\left[
\rho_0(R,X)Aq_0(Z,X)
\mid W,X
\right]=1
\quad\text{almost surely}.
\]
This proves the observed bridge equations \eqref{eq:primal-observed-bridge} and \eqref{eq:adjoint-observed-bridge}. 
\end{proof}

\subsection{Proof of Theorem~\ref{thm:identification}}

\begin{proof}[Proof of Theorem~\ref{thm:identification}.]
By Proposition~\ref{thm:observed-bridges}, the sets
\(\mathcal S(T,b)\) and \(\mathcal S(T^\ast,1)\) are nonempty. Fix an arbitrary reference pair
\((\mu_0,q_0)\in\mathcal S(T,b)\times\mathcal S(T^\ast,1)\). 

Let \(\mu\in\mathcal S(T,b)\). Then
\(T\mu=T\mu_0=b\), so \(T(\mu-\mu_0)=0\). Using \(T^\ast q_0=1\), the adjoint identity, and the first latent bridge condition of \eqref{eq:latent-bridges} gives
\begin{equation}
\begin{aligned}
\E\{\mu(W,X)\}
&=
\E\{\mu_0(W,X)\}
+
\E\!\left[
\{\mu(W,X)-\mu_0(W,X)\}(T^\ast q_0)(W,X)
\right]
\\
&=
\E\{\mu_0(W,X)\}
+
\E\!\left[
q_0(Z,X)\{T\mu(Z,X)-T\mu_0(Z,X)\}
\right]
\\
&=
\E\{\mu_0(W,X)\} =
\E\!\left[
\E\{\mu_0(W,X)\mid U,X\}
\right]
\\
&=
\E\{\tau(U,X)\}=
\E\!\left[
\E\{Y(1)-Y(0)\mid U,X\}
\right]
\\
&=
\E\{Y(1)-Y(0)\}
=
\theta_0.
\end{aligned}
\label{eq:app-primal-solution-invariance}
\end{equation}
Hence the identifying functional \(\mu\mapsto\E\{\mu(W,X)\}\) is constant
on the entire solution set \(\mathcal S(T,b)\), even when the primal bridge itself is
not unique.

Next let \(q\in\mathcal S(T^\ast,1)\). Since \(b=T\mu_0\) and \(T^\ast q=1\),
\begin{equation}
\begin{aligned}
\E\{\rho_0(R,X)q(Z,X)Y\}
&=
\E\!\left[
q(Z,X)\E\{\rho_0(R,X)Y\mid Z,X\}
\right]
\\
&=
\E\{q(Z,X)b(Z,X)\}
\\
&=
\E\{q(Z,X)(T\mu_0)(Z,X)\}
\\
&=
\E\{\mu_0(W,X)(T^\ast q)(W,X)\}
\\
&=
\E\{\mu_0(W,X)\}
=
\theta_0.
\end{aligned}
\label{eq:app-adjoint-solution-invariance}
\end{equation}
Thus the weighted-outcome functional \(q\mapsto\E\{\rho_0(R,X)q(Z,X)Y\}\) is also constant over the full
adjoint solution set \(\mathcal S(T^\ast,1)\).

For arbitrary square-integrable \((\mu,q)\), we rewrite the expectation of
the doubly robust score in terms of the observed operators:
\begin{equation}
\begin{aligned}
\E\{\psi_{\mathrm{DR}}(O;\mu,q)\}
&=
\E\{\mu(W,X)\}
+
\E\!\left[
\rho_0(R,X)q(Z,X)\{Y-A\mu(W,X)\}
\right]
\\
&=
\E\{\mu(W,X)\}
+
\E\!\left[
q(Z,X)\E\{\rho_0(R,X)Y\mid Z,X\}
\right]
\\
&\quad-
\E\!\left[
q(Z,X)\E\{\rho_0(R,X)A\mu(W,X)\mid Z,X\}
\right]
\\
&=
\E\{\mu(W,X)\}
+
\E\{q(Z,X)b(Z,X)\}
-
\E\{q(Z,X)(T\mu)(Z,X)\}
\\
&=
\E\{\mu(W,X)\}
+
\E\!\left[
q(Z,X)\{b(Z,X)-(T\mu)(Z,X)\}
\right].
\end{aligned}
\label{eq:app-dr-operator-representation}
\end{equation}

If \(\mu\in\mathcal S(T,b)\), then \(T\mu=b\), and
\eqref{eq:app-primal-solution-invariance} together with \eqref{eq:app-dr-operator-representation}
yields
\[
\begin{aligned}
\E\{\psi_{\mathrm{DR}}(O;\mu,q)\}
&=
\E\{\mu(W,X)\}
+
\E\!\left[
q(Z,X)\{b(Z,X)-(T\mu)(Z,X)\}
\right]
\\
&=
\E\{\mu(W,X)\}
=
\theta_0.
\end{aligned}
\]

If instead \(q\in\mathcal S(T^\ast,1)\), then \(T^\ast q=1\). Using \eqref{eq:app-adjoint-solution-invariance} gives
\[
\begin{aligned}
\E\{\psi_{\mathrm{DR}}(O;\mu,q)\}
&=
\E\{\mu(W,X)\}
+
\E\{q(Z,X)b(Z,X)\}
-
\E\{q(Z,X)(T\mu)(Z,X)\}
\\
&=
\E\{\mu(W,X)\}
+
\E\{q(Z,X)b(Z,X)\}
-
\E\{\mu(W,X)(T^\ast q)(W,X)\}
\\
&=
\E\{q(Z,X)b(Z,X)\}
\\
&=
\E\{\rho_0(R,X)q(Z,X)Y\}
=
\theta_0.
\end{aligned}
\]
This
proves the stated doubly robust identification result.

It remains to establish the exact mixed-bias identity. Define
\[
\delta_\mu(W,X)=\mu(W,X)-\mu_0(W,X),
\quad
\delta_q(Z,X)=q(Z,X)-q_0(Z,X).
\]
Because \(b=T\mu_0\), \(T^\ast q_0=1\), and
\(\theta_0=\E\{\mu_0(W,X)\}\), the bias equals
\[
\begin{aligned}
\E\{\psi_{\mathrm{DR}}(O;\mu,q)\}-\theta_0
&=
\E\{\mu(W,X)-\mu_0(W,X)\}
+
\E\!\left[
q(Z,X)\{b(Z,X)-(T\mu)(Z,X)\}
\right]
\\
&=
\E\{\delta_\mu(W,X)\}
-
\E\{q(Z,X)(T\delta_\mu)(Z,X)\}
\\
&=
\E\{\delta_\mu(W,X)(T^\ast q_0)(W,X)\}
-
\E\{q(Z,X)(T\delta_\mu)(Z,X)\}
\\
&=
\E\{q_0(Z,X)(T\delta_\mu)(Z,X)\}
-
\E\{q(Z,X)(T\delta_\mu)(Z,X)\}
\\
&=
-\E\{\delta_q(Z,X)(T\delta_\mu)(Z,X)\}
\\
&=
-\E\!\left[
\{q(Z,X)-q_0(Z,X)\}
\{T\mu(Z,X)-T\mu_0(Z,X)\}
\right].
\end{aligned}
\]
Equivalently, applying the adjoint identity once more gives the symmetric
representation
\[
\begin{aligned}
\E\{\psi_{\mathrm{DR}}(O;\mu,q)\}-\theta_0
&=
-\E\{\delta_q(Z,X)(T\delta_\mu)(Z,X)\}
\\
&=
-\E\{\delta_\mu(W,X)(T^\ast\delta_q)(W,X)\}.
\end{aligned}
\]
This proves \eqref{eq:mixed-bias}. 
\end{proof}

\section{Efficient Influence Function}

	\begin{proof}[Proof of Theorem~\ref{thm:eif}]
We first establish two intermediate results that are common to both propensity models.

		\paragraph{Tangent-space characterization.}
Let \(L_2^0(P_0) = \{f:\E\{f(O)^2\}<\infty,\ \E\{f(O)\}=0\}\) be the Hilbert space of square-integrable, mean-zero functions under \(P_0\).
Also let \(\mathcal T_{\Mcal}\) and
\(\mathcal T_{\Mcal_{\pi_0}}\) denote the tangent spaces of
\(\Mcal\) and \(\Mcal_{\pi_0}\), respectively, at \(P_0\). 

We first establish the dense range property of the nuisance derivative.
At the true law, write the primal conditional moment as
		\[
		m(\mu)(Z,X)
		=
		\E\!\left[
		\rho_0(R,X)\{Y-A\mu(W,X)\}
		\mid Z,X
		\right].
		\]
For an arbitrary direction \(g\in L_2(P_{W,X})\), the map
\(\mu\mapsto m(\mu)\) is affine, and therefore
		\begin{align*}
			\left.
			\frac{\partial}{\partial t}
			m(\mu_0+t g)(Z,X)
			\right|_{t=0}
			&=
			\left.
			\frac{\partial}{\partial t}
			\E\!\left[
			\rho_0(R,X)
			\{Y-A[\mu_0(W,X)+t g(W,X)]\}
			\mid Z,X
			\right]
			\right|_{t=0}
			\\
			&=
			-\E\!\left[
			\rho_0(R,X)A g(W,X)
			\mid Z,X
			\right]
			\\
			&=
			-(Tg)(Z,X).
		\end{align*}
Thus the derivative of the conditional moment with respect to the bridge
nuisance is the linear operator \(-T\). Moreover, by conditional
Jensen's inequality and the uniform IV overlap condition in
Assumption~\ref{ass:iv},
		\[
		\|Tg\|_2^2
		\leq
		\E\{\rho_0(R,X)^2A^2g(W,X)^2\}
		\leq
		K\|g\|_2^2
		\]
for some finite constant \(K\), so \(T\) is bounded.

By the second condition in Assumption~\ref{ass:completeness},
\(\Ker(T^\ast)=\{0\}\). The Hilbert-space identity therefore gives
		\[
		\overline{\Range(T)}
		=
		\Ker(T^\ast)^\perp
		=
		L_2(P_{Z,X}).
		\]
Hence the nuisance derivative \(-T\) has dense range. Under the standard
regularity conditions for nonparametric conditional moment models, the local
overidentification characterization of \citet{chen2018overidentification}
implies that the primal bridge restriction is locally just identified and
therefore imposes no additional first-order restriction on the observed-data
law. 
Consequently, the unknown-propensity model \(\Mcal\) is locally
nonparametric, with tangent space \(\mathcal T_{\Mcal}=L_2^0(P_0)\).
In the known-propensity model \(\Mcal_{\pi_0}\), the primal bridge restriction
likewise imposes no additional first-order restriction; the tangent space is
therefore restricted only by the requirement that the instrument propensity
remain fixed at \(\pi_0\).

		\paragraph{Pathwise derivative.} We next derive the pathwise derivative.
Let \(D_0=Y-A\mu_0(W,X)\)
and
		\(H_0=q_0(Z,X)D_0\). Let
\(\{P_t:t\in(-\epsilon,\epsilon)\}\) denote an arbitrary regular
parametric submodel through \(P_0\) of the model under consideration,
with score \(S_0(O)\), and let \(\E_t\) denote expectation under \(P_t\). Take a dominated
representation with common dominating measure \(\Lambda\), and write
\(p_t\) for the relevant marginal or conditional density or mass
function under \(P_t\).
In the unknown-propensity model, write
		\[
		\pi_t(X)=P_t(R=1\mid X),
		\quad
		\rho_t(R,X)
		=
		\frac{R}{\pi_t(X)}
		-
		\frac{1-R}{1-\pi_t(X)}.
		\]
In the known-propensity model, \(\pi_t=\pi_0\) and hence
\(\rho_t=\rho_0\) along the submodel. 

Holding the bridge fixed at \(\mu_0\), define the first-order
perturbation of the primal conditional moment induced by the law in
score direction \(S_0\) as
		\begin{equation}
			\dot m_{S_0}(Z,X)
			=
			\left.
			\frac{\partial}{\partial t}
			\E_t\!\left[
			\rho_t(R,X)D_0
			\mid Z,X
			\right]
			\right|_{t=0}.
			\label{eq:moment-perturbation-definition}
		\end{equation}

To determine how this perturbation affects the target, consider an
arbitrary direction \(g\in L_2(P_{W,X})\). The perturbation
\(\mu_0+t g\) is used only to measure the first-order effect of changing
the bridge in direction \(g\); it is not assumed to be an actual bridge
under \(P_t\), and \(g\) is not assumed to be the derivative of any
path of bridge solutions.
First consider the resulting first-order change in the target.
Differentiating at \(t=0\) gives
		\begin{align*}
			&
			\left.
			\frac{\partial}{\partial t}
			\E_t\{\mu_0(W,X)+t g(W,X)\}
			\right|_{t=0}
			\\
			={}&
			\left.
			\frac{\partial}{\partial t}
			\E_t\{\mu_0(W,X)\}
			\right|_{t=0}
			+
			\left.
			\frac{\partial}{\partial t}
			\left[
			t\E_t\{g(W,X)\}
			\right]
			\right|_{t=0}
			\\
			={}&
			\E\{\mu_0(W,X)S_0(O)\}
			+
			\E\{g(W,X)\}
			\\
			={}&
			\E\!\left[
			\{\mu_0(W,X)-\theta_0\}S_0(O)
			\right]
			+
			\E\{g(W,X)\},
		\end{align*}
where the last equality uses \(\E\{S_0(O)\}=0\). The first term is
the direct effect of perturbing the observed-data law, whereas the
second is the contribution from changing the bridge in direction \(g\).

The law and bridge directions cannot vary independently if the primal
bridge restriction is to remain satisfied to first order. To determine
their relation, consider
		\[
		\E_t\!\left[
		\rho_t(R,X)\{D_0-tA g(W,X)\}
		\mid Z,X
		\right].
		\]
Differentiating at \(t=0\) yields
		\begin{align*}
			&
			\left.
			\frac{\partial}{\partial t}
			\E_t\!\left[
			\rho_t(R,X)\{D_0-tA g(W,X)\}
			\mid Z,X
			\right]
			\right|_{t=0}
			\\
			={}&
			\left.
			\frac{\partial}{\partial t}
			\E_t\{\rho_t(R,X)D_0\mid Z,X\}
			\right|_{t=0}
			-
			\left.
			\frac{\partial}{\partial t}
			\left[
			t\E_t\{\rho_t(R,X)A g(W,X)\mid Z,X\}
			\right]
			\right|_{t=0}
			\\
			={}&
			\dot m_{S_0}(Z,X)
			-
			\E\{\rho_0(R,X)A g(W,X)\mid Z,X\}
			\\
			={}&
			\dot m_{S_0}(Z,X)
			-
			(Tg)(Z,X).
		\end{align*}
Hence, whenever an exact bridge direction can cancel the law-induced
moment perturbation, preservation of the primal bridge restriction to
first order requires
		\begin{equation}
			(Tg)(Z,X)
			=
			\dot m_{S_0}(Z,X).
			\label{eq:linearized-bridge-restriction}
		\end{equation}

The contribution of \(g\) to the target can therefore be determined
without recovering \(g\) itself. Since \(T^\ast q_0=1\),
		\begin{align*}
			\E\{g(W,X)\}
			&=
			\left\langle 1,g\right\rangle_{L_2(P_{W,X})} \\
			&=
			\left\langle T^\ast q_0,g\right\rangle_{L_2(P_{W,X})} \\
			&=
			\left\langle q_0,Tg\right\rangle_{L_2(P_{Z,X})} \\
			&=
			\E\!\left[
			q_0(Z,X)(Tg)(Z,X)
			\right].
		\end{align*}
Therefore, if \(\dot m_{S_0}\in\Range(T)\),
\eqref{eq:linearized-bridge-restriction} implies
		\begin{equation}
			\E\{g(W,X)\}
			=
			\E\!\left[
			q_0(Z,X)\dot m_{S_0}(Z,X)
			\right].
			\label{eq:bridge-direction-target-effect}
		\end{equation}

Completeness gives only dense range, rather than
\(\Range(T)=L_2(P_{Z,X})\), so an exact direction \(g\) satisfying
\(Tg=\dot m_{S_0}\) need not exist. However, by the second completeness
condition, \(\overline{\Range(T)} = L_2(P_{Z,X})\).
Hence, for every square-integrable conditional moment perturbation
\(\dot m_{S_0}\), there exists a sequence
\(g_j\in L_2(P_{W,X})\) such that \(Tg_j\rightarrow\dot m_{S_0}\)
in \(L_2(P_{Z,X})\). Therefore,
		\begin{align*}
			&
			\left|
			\E\{g_j(W,X)\}
			-
			\E\!\left[
			q_0(Z,X)\dot m_{S_0}(Z,X)
			\right]
			\right|
			\\
			={}&
			\left|
			\E\!\left[
			q_0(Z,X)(Tg_j)(Z,X)
			\right]
			-
			\E\!\left[
			q_0(Z,X)\dot m_{S_0}(Z,X)
			\right]
			\right|
			\\
			={}&
			\left|
			\E\!\left[
			q_0(Z,X)
			\{(Tg_j)(Z,X)-\dot m_{S_0}(Z,X)\}
			\right]
			\right|
			\\
			\leq{}&
			\|q_0\|_{L_2(P_{Z,X})}
			\,
			\|Tg_j-\dot m_{S_0}\|_{L_2(P_{Z,X})}
			\rightarrow0.
		\end{align*}
Thus, although an exact \(g\) need not exist, the target contribution
of any approximating sequence \(g_j\) converges to the uniquely
determined limit \(\E[q_0(Z,X)\dot m_{S_0}(Z,X)]\), so the bridge
contribution to the target is well defined.

Combining this bridge contribution with the direct effect of perturbing
the observed-data law yields the pathwise derivative
		\begin{equation}
			\dot\theta_0(S_0)
			=
			\E\!\left[
			\{\mu_0(W,X)-\theta_0\}S_0(O)
			\right]
			+
			\E\!\left[
			q_0(Z,X)\dot m_{S_0}(Z,X)
			\right].
			\label{eq:riesz-pathwise-derivative}
		\end{equation}
This is the Riesz representer characterization of a regular linear
functional in a conditional moment model
\citep{ai2003efficient,severini2012efficiency}. The argument uses only
the first-order effect of a bridge perturbation \(g\) on the
conditional moment, and therefore does not require
differentiability of the bridge solution as the observed-data law
varies.

We now evaluate \(\dot m_{S_0}\) in
\eqref{eq:moment-perturbation-definition} separately in the unknown- and
known-propensity models.

		\paragraph{(i) Unknown \(\pi_0\).}
Take the submodel \(\{P_t:t\in(-\epsilon,\epsilon)\}\) to lie in \(\Mcal\). Since the
ambient observed-data law is otherwise unrestricted and the primal
bridge restriction imposes no additional first-order restriction, the
tangent space is \(\mathcal T_{\Mcal} = L_2^0(P_0)\).

To evaluate the variation of the propensity, first define the conditional
score of \((Y,A,Z,W,R)\) given \(X\) as
		\begin{align*}
			S_0^X(O)
			&=
			\left.
			\frac{\partial}{\partial t}
			\log p_t(Y,A,Z,W,R\mid X)
			\right|_{t=0}
			\\
			&=
			\left.
			\frac{\partial}{\partial t}
			\left\{
			\log p_t(O)-\log p_t(X)
			\right\}
			\right|_{t=0}
			\\
			&=
			S_0(O)
			-
			\frac{1}{p_0(X)}
			\left.
			\frac{\partial}{\partial t}
			p_t(X)
			\right|_{t=0}
			\\
			&=
			S_0(O)-\E\{S_0(O)\mid X\}.
		\end{align*}
Since \(\pi_t(x)=\E_t(R\mid X=x)\), differentiating at \(t=0\)
gives
		\begin{equation}
		\begin{aligned}
			\dot\pi_0(x)
			&=
			\left.
			\frac{\partial}{\partial t}
			\sum_{r\in\{0,1\}}
			\int
			r\,p_t(y,a,z,w,r\mid x)
			\,d\Lambda(y,a,z,w)
			\right|_{t=0}
			\\
			&=
			\sum_{r\in\{0,1\}}
			\int
			r\,p_0(y,a,z,w,r\mid x)S_0^X(o)
			\,d\Lambda(y,a,z,w)
			\\
			&=
			\E\{R S_0^X(O)\mid X=x\}
			\\
			&=
			\E\!\left[
			R\{S_0(O)-\E(S_0(O)\mid X)\}
			\mid X=x
			\right]
			\\
			&=
			\E\{R S_0(O)\mid X=x\}
			-
			\pi_0(x)\E\{S_0(O)\mid X=x\}
			\\
			&=
			\E\!\left[
			\{R-\pi_0(X)\}S_0(O)
			\mid X=x
			\right].
		\end{aligned}
		\label{eq:propensity-score-identity}
		\end{equation}
Differentiating \(\rho_t\) gives
		\begin{align*}
			\dot\rho_0(R,X)
			&=
			\left.
			\frac{\partial}{\partial t}
			\left\{
			\frac{R}{\pi_t(X)}
			-
			\frac{1-R}{1-\pi_t(X)}
			\right\}
			\right|_{t=0} \\
			&=
			-
			\left\{
			\frac{R}{\pi_0(X)^2}
			+
			\frac{1-R}{\{1-\pi_0(X)\}^2}
			\right\}
			\dot\pi_0(X).
		\end{align*}

Similarly, the conditional score of \((Y,A,R,W)\) given \((Z,X)\) is \(S_0^{Z,X}(O)=S_0(O)-\E\{S_0(O)\mid Z,X\}\).
Differentiating the conditional moment in
\eqref{eq:moment-perturbation-definition} gives
	\begin{equation}
	\begin{aligned}
		\dot m_{S_0}(Z,X)
		&=
		\left.
		\frac{\partial}{\partial t}
		\int
		\rho_t(r,X)D_0
		p_t(y,a,r,w\mid Z,X)
		\,d\Lambda(y,a,r,w)
		\right|_{t=0}
		\\
		&=
		\int
		\dot\rho_0(r,X)D_0
		p_0(y,a,r,w\mid Z,X)
		\,d\Lambda(y,a,r,w)
		\\
		&\quad+
		\int
		\rho_0(r,X)D_0
		\left.
		\frac{\partial}{\partial t}
		p_t(y,a,r,w\mid Z,X)
		\right|_{t=0}
		\,d\Lambda(y,a,r,w)
		\\
		&=
		\int
		\dot\rho_0(r,X)D_0
		p_0(y,a,r,w\mid Z,X)
		\,d\Lambda(y,a,r,w)
		\\
		&\quad+
		\int
		\rho_0(r,X)D_0
		p_0(y,a,r,w\mid Z,X)
		S_0^{Z,X}(o)
		\,d\Lambda(y,a,r,w)
		\\
		&=
		\E\{\dot\rho_0(R,X)D_0\mid Z,X\}
		+
		\E\{\rho_0(R,X)D_0S_0^{Z,X}(O)\mid Z,X\}
		\\
		&=
		\E\{\dot\rho_0(R,X)D_0\mid Z,X\}
		+
		\E\!\left[
		\rho_0(R,X)D_0
		\left\{
		S_0(O)-\E\{S_0(O)\mid Z,X\}
		\right\}
		\middle| Z,X
		\right]
		\\
		&=
		\E\{\dot\rho_0(R,X)D_0\mid Z,X\}
		+
		\E\{\rho_0(R,X)D_0S_0(O)\mid Z,X\}
		\\
		&\quad-
		\E\{\rho_0(R,X)D_0\mid Z,X\}
		\E\{S_0(O)\mid Z,X\}
		\\
		&=
		\E\{\rho_0(R,X)D_0S_0(O)\mid Z,X\}
		+
		\E\{\dot\rho_0(R,X)D_0\mid Z,X\}.
	\end{aligned}
	\label{eq:unknown-moment-perturbation}
	\end{equation}
Substituting~\eqref{eq:unknown-moment-perturbation} into
\eqref{eq:riesz-pathwise-derivative} gives
		\[
		\E\{q_0(Z,X)\dot m_{S_0}(Z,X)\}
		=
		\E\{\rho_0(R,X)q_0(Z,X)D_0S_0(O)\}
		+
		\E\{\dot\rho_0(R,X)q_0(Z,X)D_0\}.
		\]

To evaluate the second term, define
\(\nu_r^\star(X)=\E(H_0\mid R=r,X)\) for \(r\in\{0,1\}\). By the
primal bridge equation,
\begin{align*}
	0
	&=
	\E\!\left[
	q_0(Z,X)
	\E\{\rho_0(R,X)D_0\mid Z,X\}
	\mid X
	\right] \\
	&=
	\E\{\rho_0(R,X)q_0(Z,X)D_0\mid X\} \\
	&=
	\E\!\left[
	\left\{
	\frac{R}{\pi_0(X)}
	-
	\frac{1-R}{1-\pi_0(X)}
	\right\}
	H_0
	\middle| X
	\right] \\
	&=
	\frac{\E(RH_0\mid X)}{\pi_0(X)}
	-
	\frac{\E\{(1-R)H_0\mid X\}}{1-\pi_0(X)} \\
	&=
	\frac{
		\E\!\left[
		R\,\E(H_0\mid R,X)
		\mid X
		\right]
	}{\pi_0(X)}
	-
	\frac{
		\E\!\left[
		(1-R)\E(H_0\mid R,X)
		\mid X
		\right]
	}{1-\pi_0(X)} \\
	&=
	\frac{
		\pi_0(X)\nu_1^\star(X)
	}{\pi_0(X)}
	-
	\frac{
		\{1-\pi_0(X)\}\nu_0^\star(X)
	}{1-\pi_0(X)} \\
	&=
	\nu_1^\star(X)-\nu_0^\star(X).
\end{align*}
Thus \(\nu_1^\star(X)=\nu_0^\star(X)\). Moreover,
		\[
		\omega_0(X)
		=
		\E(H_0\mid X)
		=
		\pi_0(X)\nu_1^\star(X)
		+
		\{1-\pi_0(X)\}\nu_0^\star(X),
		\]
so
		\begin{equation}
			\nu_1^\star(X)
			=
			\nu_0^\star(X)
			=
			\omega_0(X)
			\quad\text{almost surely}.
			\label{eq:nu-star-equality-proof}
		\end{equation}
Substituting \(\dot\rho_0\), conditioning on \(X\), and using
\eqref{eq:nu-star-equality-proof} gives
\begin{align*}
	&
	\E\{\dot\rho_0(R,X)q_0(Z,X)D_0\} \\
	={}&
	-\E\!\left[
	q_0(Z,X)D_0
	\left\{
	\frac{R}{\pi_0(X)^2}
	+
	\frac{1-R}{\{1-\pi_0(X)\}^2}
	\right\}
	\dot\pi_0(X)
	\right] \\
	={}&
	-\E\!\left[
	\dot\pi_0(X)
	\E\!\left[
	q_0(Z,X)D_0
	\left\{
	\frac{R}{\pi_0(X)^2}
	+
	\frac{1-R}{\{1-\pi_0(X)\}^2}
	\right\}
	\middle| X
	\right]
	\right] \\
	={}&
	-\E\!\left[
	\dot\pi_0(X)
	\left\{
	\frac{\E\{R q_0(Z,X)D_0\mid X\}}{\pi_0(X)^2}
	+
	\frac{\E\{(1-R)q_0(Z,X)D_0\mid X\}}
	{\{1-\pi_0(X)\}^2}
	\right\}
	\right] \\
	={}&
	-\E\!\left[
	\dot\pi_0(X)
	\left\{
	\frac{\nu_1^\star(X)}{\pi_0(X)}
	+
	\frac{\nu_0^\star(X)}{1-\pi_0(X)}
	\right\}
	\right] \\
	={}&
	-\E\!\left[
	\frac{\omega_0(X)}
	{\pi_0(X)\{1-\pi_0(X)\}}
	\dot\pi_0(X)
	\right] \\
	={}&
	-\E\!\left[
	\frac{\omega_0(X)}
	{\pi_0(X)\{1-\pi_0(X)\}}
	\E\!\left[
	\{R-\pi_0(X)\}S_0(O)
	\mid X
	\right]
	\right] \\
	={}&
	-\E\!\left[
	\frac{\{R-\pi_0(X)\}\omega_0(X)}
	{\pi_0(X)\{1-\pi_0(X)\}}
	S_0(O)
	\right] \\
	={}&
	-\E\{\rho_0(R,X)\omega_0(X)S_0(O)\}.
\end{align*}
Combining this with \eqref{eq:riesz-pathwise-derivative} yields
		\begin{align*}
			\dot\theta_0(S_0)
			&=
			\E\!\Bigl[
			\underbrace{
				\Bigl\{
				\mu_0(W,X)-\theta_0
				+
				\rho_0(R,X)
				\bigl[q_0(Z,X)D_0-\omega_0(X)\bigr]
				\Bigr\}
			}_{\phi_0(O)}
			S_0(O)
			\Bigr].
		\end{align*}
Since \(\mathcal T_{\Mcal}=L_2^0(P_0)\), this gradient \(\phi_0(O)\) is already the
canonical gradient in the unknown-propensity model.

		\paragraph{(ii) Known \(\pi_0\).}
Now take the submodel \(\{P_t:t\in(-\epsilon,\epsilon)\}\) to lie in
\(\Mcal_{\pi_0}\). Since \(\pi_t=\pi_0\) along the submodel,
\(\dot\pi_0(X)=0\). Equation~\eqref{eq:propensity-score-identity}
therefore implies that every score in the known-propensity model
satisfies
		\begin{equation}
			\E\!\left[
			\{R-\pi_0(X)\}S_0(O)
			\mid X
			\right]
			=
			0.
			\label{eq:known-propensity-score-restriction}
		\end{equation}
Conversely, the dense range argument above implies that the primal
bridge restriction imposes no further first-order restriction relative
to the fixed-propensity ambient model. Hence
		\[
		\mathcal T_{\Mcal_{\pi_0}}
		=
		\left\{
		S\in L_2^0(P_0):
		\E\!\left[
		\{R-\pi_0(X)\}S(O)
		\mid X
		\right]
		=0
		\right\}.
		\]

In this model, \(\rho_t=\rho_0\) along the submodel. As above, the
conditional score of \((Y,A,R,W)\) given \((Z,X)\) is \(S_0^{Z,X}(O)=S_0(O)-\E\{S_0(O)\mid Z,X\}\).
Differentiating the conditional moment in
\eqref{eq:moment-perturbation-definition} gives
		\begin{equation}
		\begin{aligned}
			\dot m_{S_0}(Z,X)
			&=
			\left.
			\frac{\partial}{\partial t}
			\int
			\rho_0(r,X)D_0
			p_t(y,a,r,w\mid Z,X)
			\,d\Lambda(y,a,r,w)
			\right|_{t=0}
			\\
			&=
			\int
			\rho_0(r,X)D_0
			p_0(y,a,r,w\mid Z,X)
			S_0^{Z,X}(o)
			\,d\Lambda(y,a,r,w)
			\\
			&=
			\E\{\rho_0(R,X)D_0S_0^{Z,X}(O)\mid Z,X\}
			\\
			&=
			\E\{\rho_0(R,X)D_0S_0(O)\mid Z,X\}
			\\
			&\quad-
			\E\{\rho_0(R,X)D_0\mid Z,X\}
			\E\{S_0(O)\mid Z,X\}
			\\
			&=
			\E\{\rho_0(R,X)D_0S_0(O)\mid Z,X\},
		\end{aligned}
		\label{eq:known-moment-perturbation}
		\end{equation}
where the last equality follows from the primal bridge equation.
Substituting~\eqref{eq:known-moment-perturbation} into the second term
of \eqref{eq:riesz-pathwise-derivative} gives
		\begin{align*}
			\E\{q_0(Z,X)\dot m_{S_0}(Z,X)\}
			&=
			\E\!\left[
			q_0(Z,X)
			\E\{\rho_0(R,X)D_0S_0(O)\mid Z,X\}
			\right] \\
			&=
			\E\!\left[
			\E\{q_0(Z,X)\rho_0(R,X)D_0S_0(O)\mid Z,X\}
			\right] \\
			&=
			\E\{\rho_0(R,X)q_0(Z,X)D_0S_0(O)\}.
		\end{align*}
Hence
		\begin{align*}
			\dot\theta_0(S_0)
			&=
			\E\!\left[
			\{\mu_0(W,X)-\theta_0\}S_0(O)
			\right]
			+
			\E\{\rho_0(R,X)q_0(Z,X)D_0S_0(O)\} \\
			&=
			\E\!\left[
			\left\{
			\mu_0(W,X)-\theta_0
			+
			\rho_0(R,X)q_0(Z,X)D_0
			\right\}S_0(O)
			\right].
		\end{align*}
Thus
		\begin{equation}
			\widetilde\phi_0(O)
			=
			\mu_0(W,X)-\theta_0
			+
			\rho_0(R,X)q_0(Z,X)D_0
			\label{eq:known-raw-gradient}
		\end{equation}
is a gradient in the known-propensity model.

The orthogonal complement of the tangent space is
		\[
		\mathcal T_{\Mcal_{\pi_0}}^{\perp}
		=
		\left\{
		\{R-\pi_0(X)\}a(X):a\in L_2(P_X)
		\right\}.
		\]
Write, for some \(c_0\in L_2(P_X)\),
		\[
		\Pi_{\mathcal T_{\Mcal_{\pi_0}}^{\perp}}\widetilde\phi_0(O)
		=
		\{R-\pi_0(X)\}c_0(X).
		\]
The projection residual is orthogonal to every element
\(\{R-\pi_0(X)\}a(X)\) of
\(\mathcal T_{\Mcal_{\pi_0}}^{\perp}\). Thus, for every \(a\in L_2(P_X)\),
		\begin{align*}
			0
			&=
			\E\!\left[
			\left\{
			\widetilde\phi_0(O)-\{R-\pi_0(X)\}c_0(X)
			\right\}
			\{R-\pi_0(X)\}a(X)
			\right] \\
			&=
			\E\!\left[
			a(X)
			\left\{
			\E\!\left[
			\{R-\pi_0(X)\}\widetilde\phi_0(O)
			\mid X
			\right]
			-
			c_0(X)
			\E\!\left[
			\{R-\pi_0(X)\}^2
			\mid X
			\right]
			\right\}
			\right] \\
			&=
			\E\!\left[
			a(X)
			\left\{
			\E\!\left[
			\{R-\pi_0(X)\}\widetilde\phi_0(O)
			\mid X
			\right]
			-
			c_0(X)\pi_0(X)\{1-\pi_0(X)\}
			\right\}
			\right].
		\end{align*}
Therefore,
		\begin{equation}
			c_0(X)
			=
			\frac{
				\E\!\left[
				\{R-\pi_0(X)\}\widetilde\phi_0(O)
				\mid X
				\right]
			}{
				\pi_0(X)\{1-\pi_0(X)\}
			}.
			\label{eq:known-proj-coef}
		\end{equation}

To evaluate the numerator, expand \(\widetilde\phi_0\):
\begin{align*}
	&
	\E\!\left[
	\{R-\pi_0(X)\}\widetilde\phi_0(O)
	\mid X
	\right] \\
	={}&
	\E\!\left[
	\{R-\pi_0(X)\}\{\mu_0(W,X)-\theta_0\}
	\mid X
	\right]
	+
	\E\!\left[
	\{R-\pi_0(X)\}\rho_0(R,X)H_0
	\mid X
	\right] \\
	={}&
	\E\!\left[
	\{\mu_0(W,X)-\theta_0\}
	\E\{R-\pi_0(X)\mid W,X\}
	\mid X
	\right]
	+
	\E\!\left[
	\{R-\pi_0(X)\}\rho_0(R,X)H_0
	\mid X
	\right] \\
	={}&
	\E\!\left[
	\{R-\pi_0(X)\}\rho_0(R,X)H_0
	\mid X
	\right].
\end{align*}
The first term vanishes because at the true law
\(R\indep W\mid X\), and therefore
\(\E\{R-\pi_0(X)\mid W,X\}=0\). For the second term, conditioning on
\(R\) and using \eqref{eq:nu-star-equality-proof} gives
		\begin{align*}
			&\E\!\left[
			\{R-\pi_0(X)\}\rho_0(R,X)H_0
			\mid X
			\right] \\
			={}&
			\pi_0(X)\frac{1-\pi_0(X)}{\pi_0(X)}
			\E(H_0\mid R=1,X) \\
			&+
			\{1-\pi_0(X)\}
			\frac{\pi_0(X)}{1-\pi_0(X)}
			\E(H_0\mid R=0,X) \\
			={}&
			\{1-\pi_0(X)\}\nu_1^\star(X)
			+
			\pi_0(X)\nu_0^\star(X) \\
			={}&
			\omega_0(X).
		\end{align*}
It follows from \eqref{eq:known-proj-coef} that
		\[
		c_0(X)
		=
		\frac{\omega_0(X)}{\pi_0(X)\{1-\pi_0(X)\}}.
		\]
Consequently,
		\begin{align*}
			\Pi_{\mathcal T_{\Mcal_{\pi_0}}}\widetilde\phi_0(O)
			&=
			\widetilde\phi_0(O)
			-
			\Pi_{\mathcal T_{\Mcal_{\pi_0}}^{\perp}}\widetilde\phi_0(O) \\
			&=
			\widetilde\phi_0(O)
			-
			\frac{R-\pi_0(X)}{\pi_0(X)\{1-\pi_0(X)\}}
			\omega_0(X) \\
			&=
			\mu_0(W,X)-\theta_0
			+
			\rho_0(R,X)q_0(Z,X)D_0
			-
			\rho_0(R,X)\omega_0(X) \\
			&=
			\phi_0(O).
		\end{align*}
Therefore \(\phi_0\) is the canonical gradient in \(\Mcal_{\pi_0}\).

It remains to verify that \(\phi_0\in L_2^0(P_0)\).
By the finite second moment condition stated in Theorem~\ref{thm:eif} and the uniform IV overlap condition in
Assumption~\ref{ass:iv}, \(\rho_0(R,X)\) is uniformly bounded and hence
		\[
		\E\!\left[
		\{\rho_0(R,X)q_0(Z,X)D_0\}^2
		\right]
		\leq
		\|\rho_0\|_\infty^2
		\E\!\left[
		q_0(Z,X)^2D_0^2
		\right]
		<\infty.
		\]
Also, by conditional Jensen's inequality, \(\E\{\omega_0(X)^2\} = \E[ \{\E(H_0\mid X)\}^2 ] \leq \E[ \E(H_0^2\mid X) ] = \E(H_0^2) <\infty\),
and thus
		\[
		\E\!\left[
		\{\rho_0(R,X)\omega_0(X)\}^2
		\right]
		\leq
		\|\rho_0\|_\infty^2
		\E\{\omega_0(X)^2\}
		<\infty.
		\]
Together with
\(\mu_0\in L_2(P_{W,X})\), these relations imply that \(\phi_0\in L_2(P_0)\).

Finally, by the definition of \(\phi_0\),
		\begin{align*}
			\E\{\phi_0(O)\}
			&=
			\E\{\mu_0(W,X)-\theta_0\}
			+
			\E\{\rho_0(R,X)q_0(Z,X)D_0\}
			-
			\E\{\rho_0(R,X)\omega_0(X)\} \\
			&=
			0
			+
			\E\!\left[
			\E\{q_0(Z,X)\rho_0(R,X)D_0\mid Z,X\}
			\right]
			-
			\E\!\left[
			\E\{\omega_0(X)\rho_0(R,X)\mid X\}
			\right] \\
			&=
			\E\!\left[
			q_0(Z,X)
			\E\{\rho_0(R,X)D_0\mid Z,X\}
			\right]
			-
			\E\!\left[
			\omega_0(X)
			\E\{\rho_0(R,X)\mid X\}
			\right] =
			0.
		\end{align*}
Hence \(\phi_0\in L_2^0(P_0)\). Since \(\phi_0\) is the canonical
gradient in both the known- and unknown-propensity models, the
semiparametric efficiency bound in either model is \(\E\{\phi_0(O)^2\}\).
	\end{proof}

\section{Bias Decomposition}
\subsection{Proof of Proposition~\ref{prop:nested-robustness}}

\begin{proof}[Proof of Proposition~\ref{prop:nested-robustness}.]
For generic \((\mu,q)\), write
\[
H_{\mu,q}
=
q(Z,X)\{Y-A\mu(W,X)\},
\quad
\nu_r[\mu,q](X)
=
\E\{H_{\mu,q}\mid R=r,X\}.
\]
Also define the doubly robust score
\[
\psi_{\mathrm{DR}}(O;\mu,q)
=
\mu(W,X)+\rho_0(R,X)H_{\mu,q}.
\]

Conditional on \(X\),
\[
\begin{aligned}
\E\{\Psi(O;\mu,q,\pi,\nu)\mid X\}
&=
\E\{\mu(W,X)\mid X\}
+
\frac{1}{\pi(X)}
\E\!\left[
R\{H_{\mu,q}-\nu_1(X)\}
\mid X
\right] \\
&\quad-
\frac{1}{1-\pi(X)}
\E\!\left[
(1-R)\{H_{\mu,q}-\nu_0(X)\}
\mid X
\right]
+
\nu_1(X)-\nu_0(X) \\
&=
\E\{\mu(W,X)\mid X\}
+
\frac{\pi_0(X)}{\pi(X)}
\{\nu_1[\mu,q](X)-\nu_1(X)\} \\
&\quad-
\frac{1-\pi_0(X)}{1-\pi(X)}
\{\nu_0[\mu,q](X)-\nu_0(X)\}
+
\nu_1(X)-\nu_0(X).
\end{aligned}
\]
Similarly,
\[
\begin{aligned}
\E\{\psi_{\mathrm{DR}}(O;\mu,q)\mid X\}
&=
\E\{\mu(W,X)\mid X\}
+
\E\{\rho_0(R,X)H_{\mu,q}\mid X\} \\
&=
\E\{\mu(W,X)\mid X\}
+
\nu_1[\mu,q](X)
-
\nu_0[\mu,q](X).
\end{aligned}
\]
Therefore,
\[
\begin{aligned}
&
\E\{\Psi(O;\mu,q,\pi,\nu)\mid X\}
-
\E\{\psi_{\mathrm{DR}}(O;\mu,q)\mid X\} \\
={}&
\left\{
\frac{\pi_0(X)}{\pi(X)}-1
\right\}
\nu_1[\mu,q](X)
+
\left\{
1-\frac{\pi_0(X)}{\pi(X)}
\right\}
\nu_1(X) \\
&+
\left\{
1-\frac{1-\pi_0(X)}{1-\pi(X)}
\right\}
\nu_0[\mu,q](X)
+
\left\{
\frac{1-\pi_0(X)}{1-\pi(X)}-1
\right\}
\nu_0(X) \\
={}&
\frac{\pi_0(X)-\pi(X)}{\pi(X)}
\{\nu_1[\mu,q](X)-\nu_1(X)\} \\
&+
\frac{\pi_0(X)-\pi(X)}{1-\pi(X)}
\{\nu_0[\mu,q](X)-\nu_0(X)\}.
\end{aligned}
\]
Taking expectations gives
\begin{equation}
\begin{aligned}
&\E\{\Psi(O;\mu,q,\pi,\nu)\}
-
\E\{\psi_{\mathrm{DR}}(O;\mu,q)\} \\
={}&
\E\left[
\{\pi_0(X)-\pi(X)\}
\left\{
\frac{\nu_1[\mu,q](X)-\nu_1(X)}{\pi(X)}
+
\frac{\nu_0[\mu,q](X)-\nu_0(X)}
     {1-\pi(X)}
\right\}
\right].
\end{aligned}
\label{eq:aipw-bias-proof}
\end{equation}

Next, by the bridge equations \(T\mu_0=b\) and \(T^\ast q_0=1\),
\begin{equation}
\begin{aligned}
\E\{\psi_{\mathrm{br}}(O;\mu,q)\}-\theta_0
&=
\E\{\mu(W,X)-\mu_0(W,X)\}
\\
&\quad+
\E\!\left[
q(Z,X)\{b(Z,X)-T\mu(Z,X)\}
\right]
\\
&=
\E\!\left[
q_0(Z,X)\{T\mu(Z,X)-T\mu_0(Z,X)\}
\right]
\\
&\quad-
\E\!\left[
q(Z,X)\{T\mu(Z,X)-T\mu_0(Z,X)\}
\right]
\\
&=
-\E\!\left[
\{q(Z,X)-q_0(Z,X)\}
\{T\mu(Z,X)-T\mu_0(Z,X)\}
\right].
\end{aligned}
\label{eq:bridge-bias-proof}
\end{equation}

Combining \eqref{eq:aipw-bias-proof} and \eqref{eq:bridge-bias-proof} yields the exact bias identity
\begin{equation}
\begin{aligned}
&
\E\{\Psi(O;\mu,q,\pi,\nu)\}-\theta_0
\\
={}&
\bigl(
\E\{\Psi(O;\mu,q,\pi,\nu)\}
-
\E\{\psi_{\mathrm{DR}}(O;\mu,q)\}
\bigr)
+
\bigl(
\E\{\psi_{\mathrm{DR}}(O;\mu,q)\}
-
\theta_0
\bigr)
\\
={}&
\E\left[
\{\pi_0(X)-\pi(X)\}
\left\{
\frac{\nu_1[\mu,q](X)-\nu_1(X)}{\pi(X)}
+
\frac{\nu_0[\mu,q](X)-\nu_0(X)}
     {1-\pi(X)}
\right\}
\right]
\\
&-
\E\left[
\{q(Z,X)-q_0(Z,X)\}
\{T\mu(Z,X)-T\mu_0(Z,X)\}
\right].
\end{aligned}
\label{eq:nested-bias-proof}
\end{equation}

The first term in \eqref{eq:nested-bias-proof} vanishes if either
\(\pi=\pi_0\) or \(\nu_r=\nu_r[\mu,q]\)
for \(r\in\{0,1\}\).
The second term vanishes if \(T\mu=T\mu_0\). It also vanishes if
\(T^\ast q=T^\ast q_0\), since
\[
\begin{aligned}
\E\left[
(q-q_0)\{T\mu-T\mu_0\}
\right] 
=
\E\left[
\{T^\ast(q-q_0)\}
(\mu-\mu_0)
\right]
=
0.
\end{aligned}
\]
The claimed nested robustness property follows.
\end{proof}

\subsection{Proof of Proposition~\ref{prop:orthogonal-residualized-moments}}

\begin{proof}[Proof of Proposition~\ref{prop:orthogonal-residualized-moments}.]
Assumption~\ref{ass:iv} implies
\(\E\{R-\pi_0(X)\mid W,Z,X\}=0\), and hence the same conditional mean is
zero given either \((Z,X)\) or \((W,X)\). Since \(R\) is binary,
\(\E[\{R-\pi_0(X)\}^2\mid W,X] =\pi_0(X)\{1-\pi_0(X)\}\). We also use \(R-\pi_0(X)=\pi_0(X)\{1-\pi_0(X)\}\rho_0(R,X)\).

\paragraph{(1) Primal problem.}
By the definition of \(\xi_{\mu,0}\),
\[
\begin{aligned}
\E\{\xi_{\mu,0}(\mu)\mid Z,X\}
={}&
\E\!\left[
\{R-\pi_0(X)\}\{Y-A\mu(W,X)\}
\mid Z,X
\right] \\
&-
\gamma_0(Z,X)\E\{R-\pi_0(X)\mid Z,X\} \\
&+
\E\!\left[
\{R-\pi_0(X)\}\alpha_0(W,Z,X)\mu(W,X)
\mid Z,X
\right] \\
={}&
\E\!\left[
\{R-\pi_0(X)\}\{Y-A\mu(W,X)\}
\mid Z,X
\right] \\
&+
\E\!\left[
\alpha_0(W,Z,X)\mu(W,X)
\E\{R-\pi_0(X)\mid W,Z,X\}
\mid Z,X
\right] \\
={}&
\E\!\left[
\{R-\pi_0(X)\}\{Y-A\mu(W,X)\}
\mid Z,X
\right] \\
={}&
\pi_0(X)\{1-\pi_0(X)\}
\E\!\left[
\rho_0(R,X)\{Y-A\mu(W,X)\}
\mid Z,X
\right] \\
={}&
\pi_0(X)\{1-\pi_0(X)\}\{b-T\mu\}(Z,X).
\end{aligned}
\]

We next establish the primal nuisance perturbation identity. Because \(R-\pi(X)=R-\pi_0(X)-\Delta_\pi(X)\),
\(\gamma=\gamma_0+\Delta_\gamma\), and
\(\alpha=\alpha_0+\Delta_\alpha\), direct expansion gives
\[
\begin{aligned}
&\xi_\mu(\mu;\pi,\gamma,\alpha)-\xi_{\mu,0}(\mu) \\
={}&
-\{R-\pi_0(X)\}\Delta_\gamma(Z,X)
+\{R-\pi_0(X)\}\Delta_\alpha(W,Z,X)\mu(W,X) \\
&-
\Delta_\pi(X)
\Bigl[
Y-\gamma_0(Z,X)
-\{A-\alpha_0(W,Z,X)\}\mu(W,X)
\Bigr] \\
&+
\Delta_\pi(X)\Delta_\gamma(Z,X)
-\Delta_\pi(X)\Delta_\alpha(W,Z,X)\mu(W,X).
\end{aligned}
\]
The conditional mean of the first term is zero because
\(\Delta_\gamma(Z,X)\) is measurable with respect to \((Z,X)\). The
remaining first-order terms vanish as follows:
\[
\begin{aligned}
&\E\!\left[
\{R-\pi_0(X)\}\Delta_\alpha(W,Z,X)\mu(W,X)
\mid Z,X
\right] \\
={}&
\E\!\left[
\E\{\{R-\pi_0(X)\}\Delta_\alpha(W,Z,X)\mu(W,X)\mid W,Z,X\}
\mid Z,X
\right] 
\\
={}&
\E\!\left[
\Delta_\alpha(W,Z,X)\mu(W,X)
\E\{R-\pi_0(X)\mid W,Z,X\}
\mid Z,X
\right] \\
={}&0,
\end{aligned}
\]
and
\[
\begin{aligned}
&\E\!\left[
Y-\gamma_0(Z,X)
-\{A-\alpha_0(W,Z,X)\}\mu(W,X)
\mid Z,X
\right] \\
={}&
\E\{Y-\gamma_0(Z,X)\mid Z,X\} -
\E\!\left[
\{A-\alpha_0(W,Z,X)\}\mu(W,X)
\mid Z,X
\right] \\
={}&
-
\E\!\left[
\E\!\left[
\{A-\alpha_0(W,Z,X)\}\mu(W,X)
\mid W,Z,X
\right]
\mid Z,X
\right] \\
={}&
-
\E\!\left[
\mu(W,X)
\E\{A-\alpha_0(W,Z,X)\mid W,Z,X\}
\mid Z,X
\right] \\
={}&0.
\end{aligned}
\]
Where we use definitions
\(\gamma_0(Z,X)=\E(Y\mid Z,X)\) and
\(\alpha_0(W,Z,X)=\E(A\mid W,Z,X)\). Taking conditional expectations in
the expansion therefore yields
\[
\begin{aligned}
&\E\{\xi_\mu(\mu;\pi,\gamma,\alpha)
      -\xi_{\mu,0}(\mu)\mid Z,X\} \\
={}&
\Delta_\pi(X)\Delta_\gamma(Z,X)
-\Delta_\pi(X)
 \E\{\Delta_\alpha(W,Z,X)\mu(W,X)\mid Z,X\} \\
={}&
\Delta_\pi(X)
\left[
\Delta_\gamma(Z,X)
-
\E\{\Delta_\alpha(W,Z,X)\mu(W,X)\mid Z,X\}
\right],
\end{aligned}
\]
which proves \eqref{eq:primal-orthogonality-identity}.

To verify Neyman orthogonality, fix a square-integrable candidate
\(\mu\) and admissible directions \(h_\pi\), \(h_\gamma\), and
\(h_\alpha\). Consider the paths
\(\pi_t=\pi_0+t h_\pi\), \(\gamma_t=\gamma_0+t h_\gamma\), and
\(\alpha_t=\alpha_0+t h_\alpha\). Along these paths, the residual admits
the exact expansion
\[
\begin{aligned}
\xi_\mu(\mu;\pi_t,\gamma_t,\alpha_t) 
&=
\xi_{\mu,0}(\mu)
+t\Biggl[
-h_\pi(X)
\Bigl\{
Y-\gamma_0(Z,X)
-\{A-\alpha_0(W,Z,X)\}\mu(W,X)
\Bigr\} \\
&\qquad
-\{R-\pi_0(X)\}h_\gamma(Z,X)
+\{R-\pi_0(X)\}h_\alpha(W,Z,X)\mu(W,X)
\Biggr] \\
&\quad+
 t^2 h_\pi(X)
\Bigl[
 h_\gamma(Z,X)-h_\alpha(W,Z,X)\mu(W,X)
\Bigr].
\end{aligned}
\]
Then,
\(
\left.
\frac{\mathrm d}{\mathrm dt}
\E\{\xi_\mu(\mu;\pi_t,\gamma_t,\alpha_t)\mid Z,X\}
\right|_{t=0} 
=0
\),
which establishes Neyman orthogonality of the primal conditional moment
with respect to nuisances \((\pi,\gamma,\alpha)\).

\paragraph{(2) Adjoint problem.}
By the definition of \(\xi_{q,0}\),
\[
\begin{aligned}
\E\{\xi_{q,0}(q)\mid W,X\}
={}&
\E\!\left[
\{R-\pi_0(X)\}^2
\mid W,X
\right]
-
\E\!\left[
\{R-\pi_0(X)\}Aq(Z,X)
\mid W,X
\right] \\
&+
\E\!\left[
\{R-\pi_0(X)\}\alpha_0(W,Z,X)q(Z,X)
\mid W,X
\right] \\
={}&
\pi_0(X)\{1-\pi_0(X)\}
-
\E\!\left[
\{R-\pi_0(X)\}Aq(Z,X)
\mid W,X
\right] \\
&+
\E\!\left[
\alpha_0(W,Z,X)q(Z,X)
\E\{R-\pi_0(X)\mid W,Z,X\}
\mid W,X
\right] \\
={}&
\pi_0(X)\{1-\pi_0(X)\} -
\pi_0(X)\{1-\pi_0(X)\}
\E\!\left[
\rho_0(R,X)Aq(Z,X)
\mid W,X
\right] \\
={}&
\pi_0(X)\{1-\pi_0(X)\}\{1-T^\ast q\}(W,X).
\end{aligned}
\]
Since uniform IV overlap implies \(\pi_0(X)\{1-\pi_0(X)\}>0\) almost surely, the
residualized adjoint restriction is zero if and only if \(T^\ast q=1\).

We next establish the adjoint nuisance perturbation identity. Because
\(R-\pi(X)=R-\pi_0(X)-\Delta_\pi(X)\) and
\(\alpha=\alpha_0+\Delta_\alpha\), direct expansion gives
\[
\begin{aligned}
&\xi_q(q;\pi,\alpha)-\xi_{q,0}(q) \\
={}&
-\{R-\pi_0(X)\}\Delta_\pi(X)
+\{R-\pi_0(X)\}\Delta_\alpha(W,Z,X)q(Z,X) \\
&-
\Delta_\pi(X)
\Bigl[
R-\pi_0(X)
-\{A-\alpha_0(W,Z,X)\}q(Z,X)
\Bigr] \\
&+
\{\Delta_\pi(X)\}^2
-\Delta_\pi(X)\Delta_\alpha(W,Z,X)q(Z,X).
\end{aligned}
\]
The conditional mean of the first term is zero because
\(\Delta_\pi(X)\) is measurable with respect to \((W,X)\). The remaining
first-order terms vanish as follows:
\[
\begin{aligned}
&\E\!\left[
\{R-\pi_0(X)\}\Delta_\alpha(W,Z,X)q(Z,X)
\mid W,X
\right] \\
={}&
\E\!\left[
\E\{\{R-\pi_0(X)\}\Delta_\alpha(W,Z,X)q(Z,X)\mid W,Z,X\}
\mid W,X
\right] \\
={}&
\E\!\left[
\Delta_\alpha(W,Z,X)q(Z,X)
\E\{R-\pi_0(X)\mid W,Z,X\}
\mid W,X
\right] \\
={}&0,
\end{aligned}
\]
and
\[
\begin{aligned}
&\E\!\left[
R-\pi_0(X)
-\{A-\alpha_0(W,Z,X)\}q(Z,X)
\mid W,X
\right] \\
={}&
\E\{R-\pi_0(X)\mid W,X\} -
\E\!\left[
\{A-\alpha_0(W,Z,X)\}q(Z,X)
\mid W,X
\right] \\
={}&
-
\E\!\left[
\E\!\left[
\{A-\alpha_0(W,Z,X)\}q(Z,X)
\mid W,Z,X
\right]
\mid W,X
\right] \\
={}&-
\E\!\left[
q(Z,X)
\E\{A-\alpha_0(W,Z,X)\mid W,Z,X\}
\mid W,X
\right] \\
={}&0.
\end{aligned}
\]
Where we use
\(\E\{R-\pi_0(X)\mid W,X\}=0\) and
\(\alpha_0(W,Z,X)=\E(A\mid W,Z,X)\). Taking conditional expectations in
the expansion therefore gives
\[
\begin{aligned}
&\E\{\xi_q(q;\pi,\alpha)-\xi_{q,0}(q)\mid W,X\} \\
={}&
-\Delta_\pi(X)\E\{R-\pi_0(X)\mid W,X\} \\
&+
\E\!\left[
\{R-\pi_0(X)\}\Delta_\alpha(W,Z,X)q(Z,X)
\mid W,X
\right] \\
&-
\Delta_\pi(X)
\E\!\left[
R-\pi_0(X)
-\{A-\alpha_0(W,Z,X)\}q(Z,X)
\mid W,X
\right] \\
&+
\{\Delta_\pi(X)\}^2
-\Delta_\pi(X)
\E\{\Delta_\alpha(W,Z,X)q(Z,X)\mid W,X\} \\
={}&
\{\Delta_\pi(X)\}^2
-\Delta_\pi(X)
\E\{\Delta_\alpha(W,Z,X)q(Z,X)\mid W,X\} \\
={}&
\Delta_\pi(X)
\left[
\Delta_\pi(X)
-
\E\{\Delta_\alpha(W,Z,X)q(Z,X)\mid W,X\}
\right],
\end{aligned}
\]
which proves \eqref{eq:adjoint-orthogonality-identity}.

To verify Neyman orthogonality, fix a square-integrable candidate \(q\)
and admissible directions \(h_\pi\) and \(h_\alpha\). Consider the paths
\(\pi_t=\pi_0+t h_\pi\) and
\(\alpha_t=\alpha_0+t h_\alpha\). Along these paths, the residual admits
the exact expansion
\[
\begin{aligned}
\xi_q(q;\pi_t,\alpha_t)
&=
\xi_{q,0}(q)
+t\Biggl[
-h_\pi(X)
\Bigl\{
R-\pi_0(X)
-\{A-\alpha_0(W,Z,X)\}q(Z,X)
\Bigr\} \\
&\qquad
-\{R-\pi_0(X)\}h_\pi(X)
+\{R-\pi_0(X)\}h_\alpha(W,Z,X)q(Z,X)
\Biggr] \\
&\quad+
 t^2 h_\pi(X)
\Bigl[
 h_\pi(X)-h_\alpha(W,Z,X)q(Z,X)
\Bigr].
\end{aligned}
\]
Then, \(\left. \frac{\mathrm d}{\mathrm dt} \E\{\xi_q(q;\pi_t,\alpha_t)\mid W,X\} \right|_{t=0} =0\),
which establishes Neyman orthogonality of the adjoint conditional moment
with respect to nuisances \((\pi,\alpha)\).
\end{proof}

\input{large_sample_appendix}

	\section{Additional Details on Experiments}
    \label{sec:app-experiments}

	\subsection{Estimators and inference}
	\label{sec:app-estimators}

	\subsubsection{AIPW Wald estimator}

Let \(m_r^Y(x)=\E(Y\mid R=r,X=x)\) and
\(m_r^A(x)=\E(A\mid R=r,X=x)\), for \(r\in\{0,1\}\). Let
\(\widehat\pi_\ell\) denote the fold-specific instrument propensity estimate,
with \(\widehat\pi_\ell=\pi_0\) when the propensity is known. Following
\citet{tan2006regression}, the cross-fitted AIPW scores are
    \begin{equation}
	\begin{aligned}
		\widehat\psi^Y_{i,\ell}
		={}&
		\widehat m^Y_{1,\ell}(X_i)-\widehat m^Y_{0,\ell}(X_i)
		+\frac{R_i\{Y_i-\widehat m^Y_{1,\ell}(X_i)\}}
		{\widehat\pi_\ell(X_i)}
		-\frac{(1-R_i)\{Y_i-\widehat m^Y_{0,\ell}(X_i)\}}
		{1-\widehat\pi_\ell(X_i)},\\
		\widehat\psi^A_{i,\ell}
		={}&
		\widehat m^A_{1,\ell}(X_i)-\widehat m^A_{0,\ell}(X_i)
		+\frac{R_i\{A_i-\widehat m^A_{1,\ell}(X_i)\}}
		{\widehat\pi_\ell(X_i)}
		-\frac{(1-R_i)\{A_i-\widehat m^A_{0,\ell}(X_i)\}}
		{1-\widehat\pi_\ell(X_i)}.
	\end{aligned}
    \label{eq:app-wald-score}
    \end{equation}
The outcome and first-stage regressions are fitted by
feed-forward networks. The corresponding reduced-form and first-stage
estimates are
	\[
	\widehat\Delta_Y
	=
	\frac{1}{n}\sum_{\ell=1}^{L}\sum_{i\in\mathcal I_\ell}
	\widehat\psi^Y_{i,\ell},
	\quad
	\widehat\Delta_A
	=
	\frac{1}{n}\sum_{\ell=1}^{L}\sum_{i\in\mathcal I_\ell}
	\widehat\psi^A_{i,\ell},
	\]
and the AIPW Wald estimator is
	\[
	\widehat\theta_{\mathrm{Wald}}
	=
	\frac{\widehat\Delta_Y}{\widehat\Delta_A}.
	\]

For inference, we use the
delta-method contributions
	\[
	\widehat\varphi_i
	=
	\frac{
		\widehat \psi_i^Y
		-\widehat\theta_{\mathrm{Wald}}\widehat\psi_i^A
	}{
		\widehat\Delta_A
	},
	\]
and apply the variance estimator in \eqref{eq:variance-estimator}, replacing
\(\widehat\psi_{i,\ell}-\widehat\theta\) by \(\widehat\varphi_{i,\ell}\).

\subsubsection{Inverse compliance-score-weighted estimator}

We implement the ICSW estimator of
\citet{aronow2013beyond}. Let \(\widetilde X\) denote the observed predictors
used to model compliance; in our experiments, these comprise \((X,Z,W)\).
Under monotonicity, the conditional compliance score is
\[
s(\widetilde x)
=
\E(A\mid R=1,\widetilde X=\widetilde x)
-
\E(A\mid R=0,\widetilde X=\widetilde x).
\]
We estimate the two conditional treatment probabilities by probit regression
and take their difference. Under one-sided noncompliance, the second term is
zero, and the procedure reduces to fitting a single probit model among
observations with \(R=1\).

To limit the influence of small estimated compliance scores,  let
\(\widehat Q_{n^{-0.275}}\) denote the empirical \(n^{-0.275}\)-quantile of
\(\{\widehat s(\widetilde X_i)\}_{i=1}^n\), and define
\[
\widetilde s_i
=
\max\left\{
\widehat s(\widetilde X_i),
\widehat Q_{n^{-0.275}}
\right\}.
\]
Let \(\widehat\pi_i\) denote the instrument propensity used in estimation,
with \(\widehat\pi_i=\pi_0(\tilde X_i)\) when it is known. The inverse-probability
and final ICSW weights are
\[
d_i
=
\frac{R_i}{\widehat\pi_i}
+
\frac{1-R_i}{1-\widehat\pi_i},
\quad
\bar d
=
\frac{1}{n}\sum_{i=1}^n d_i,
\quad
w_i
=
\frac{d_i/\bar d}{\widetilde s_i}.
\]

We implement the ICSW estimator by weighted two-stage least squares. The first stage
regresses \(A\) on \(R\) and \(\widetilde X\), and the second stage regresses
\(Y\) on the fitted treatment and \(\widetilde X\), using weights
\(w_i\) in both stages. The resulting treatment coefficient is the ICSW
estimate.

Following \citet{aronow2013beyond}, we construct confidence intervals by
bootstrapping the entire estimation procedure. In each bootstrap sample, we
refit the compliance model, recompute the truncation threshold, and refit the
weighted instrumental variable regression. The reported \(95\%\) confidence
interval is formed from the \(2.5\%\) and \(97.5\%\) empirical quantiles of
the bootstrap estimates. When applicable, the bootstrap resampling respects the strata and
clusters of the study design.

	\subsubsection{Naive AIPW estimator}

As a benchmark, we also consider the conventional AIPW estimator
\citep{bang2005doubly}, which assumes that treatment \(A\) is conditionally
unconfounded given \(X\). This estimator uses neither the instrument \(R\)
nor the proxies \((Z,W)\) and is therefore generally biased under the
data generating processes considered here. Let
\(e(x)=P(A=1\mid X=x)\) and \(m_a(x)=\E(Y\mid A=a,X=x)\) for \(a\in\{0,1\}\).
For \(i\in\mathcal I_\ell\), the cross-fitted AIPW score is
\begin{equation}
\widehat\psi^{\mathrm{AIPW}}_{i,\ell}
=
\widehat m_{1,\ell}(X_i)-\widehat m_{0,\ell}(X_i)
+
\frac{A_i\{Y_i-\widehat m_{1,\ell}(X_i)\}}
     {\widehat e_\ell(X_i)}
-
\frac{(1-A_i)\{Y_i-\widehat m_{0,\ell}(X_i)\}}
     {1-\widehat e_\ell(X_i)}.
\label{eq:app-naive-aipw}
\end{equation}
The treatment propensity and outcome regressions are estimated outside
\(\mathcal I_\ell\). The AIPW estimator then averages the score contributions,
\[
\widehat\theta_{\mathrm{AIPW}}
=
\frac{1}{n}
\sum_{\ell=1}^{L}
\sum_{i\in\mathcal I_\ell}
\widehat\psi^{\mathrm{AIPW}}_{i,\ell},
\]
and inference uses the variance estimator in
\eqref{eq:variance-estimator}.

	\subsubsection{Proximal causal learning estimator}

We implement the PCL estimator of \citet{ghassami2022minimax} and
\citet{cui2024semiparametric}, by treating recorded contact \(A\) as the
exposure, with the same \((Z,W,X)\) as Proximal IV but without using the
instrumental variable. Outcome and treatment confounding bridges are
fitted by minimax learning in a Gaussian radial basis kernel. The
estimator is the mean of the corresponding doubly robust proximal score,
and inference still uses \eqref{eq:variance-estimator}.

	\subsection{Kernels and regularization}
	\label{sec:app-kernel}

	\subsubsection{Gaussian kernel}

In the simulations the bridge and critic spaces are generated by a
Gaussian radial basis kernel
	\begin{equation*}
		k(v,v')
		=
		\exp\!\left\{-\frac{\|v-v'\|_{2}^{2}}{2h^{2}}\right\},
	\end{equation*}
with a common bandwidth \(h\) for the continuous arguments of each
space. The bandwidth is set by a median heuristic on the training block:
\(h=c_h\,\widetilde h\), where \(\widetilde h\) is the median pairwise
Euclidean distance and \(c_h\) is a multiplier selected as described in
Appendix~\ref{sec:app-tuning}. Critic spaces use the same kernel family as
the corresponding bridge, evaluated at the conditioning variables.

	\subsubsection{Mixed-data kernel}
For the D017 dataset with mixed-type covariates, the bridge and critic spaces use a generalized product kernel adapted to the mixed
measurement scale of \((Z,X)\) and \((W,X)\). For arguments
\(v=(v^{c},v^{d},v^{p})\) collecting continuous, unordered discrete and
binary proxy coordinates,
	\begin{equation*}
		k(v,v')
		=
		\prod_{j\in c}
		\exp\!\left\{-\frac{(v^{c}_j-v'^{c}_j)^{2}}{2h_j^{2}}\right\}
		\times
		\prod_{j\in d}
		k_{\mathrm{AA}}(v^{d}_j,v'^{d}_j;\lambda_{\mathrm{cat}})
		\times
		\prod_{j\in p}
		k_{\mathrm{AA}}(v^{p}_j,v'^{p}_j;\lambda_{\mathrm{prox}}),
	\end{equation*}
where the continuous coordinates are age and household size, the unordered
coordinates are city, race, party, sex and the missingness indicators, and
the proxy coordinate is the binary lagged-turnout variable entering that
space. The discrete factors use a diagonal-normalized Aitchison--Aitken
kernel \citep{aitchison1976multivariate}, equivalently the unordered kernel of
\citet{racine2004nonparametric},
	\begin{equation*}
		k_{\mathrm{AA}}(u,u';\lambda)
		=
		\begin{cases}
			1, & u=u',\\
			\lambda, & u\neq u',
		\end{cases}
		\quad \lambda\in[0,1],
	\end{equation*}
so that matching levels contribute one and mismatching levels contribute
the categorical correlation \(\lambda\). The choice \(\lambda=0\) recovers
exact matching; larger \(\lambda\) smooths across levels. For continuous coordinate \(j\), we set
\(h_j=c_h\widetilde h_j\), where \(\widetilde h_j\) is its median pairwise
absolute difference on the training sample.

	\subsection{Hyperparameter selection}
	\label{sec:app-tuning}

	\subsubsection{Bridge hyperparameters}

Hyperparameters for the bridge learners are chosen by a held-out criterion
for the residual conditional moment restriction. Let
\(D_\mu(O;\mu)=\rho_0(R,S)\{Y-A\mu(W,X)\}\) be the primal residual and let
\(k'_\mu\) be the primal critic kernel on \((Z,X)\). The population criterion
is the quadratic kernel-moment risk
    \begin{equation*}
	   \mathcal{J}_\mu(\mu)
	   =
	   \E\bigl[D_\mu(O;\mu)\,k'_\mu\bigl((Z,X),(Z',X')\bigr)\,D_\mu(O';\mu)\bigr],
    \end{equation*}
where \((O',Z',X')\) is an independent copy of \((O,Z,X)\). The adjoint
risk \(\mathcal{J}_q(q)\) is defined in the same way from the residual
\(D_q(O;q)=\rho_0(R,S)\{R-Aq(Z,X)\}\) and the critic kernel \(k'_q\) on
\((W,X)\). Each functional is the squared kernel norm of the corresponding
conditional moment violation and vanishes at a solution of the bridge
equation. It is the kernel maximum moment restriction of
\citet{mastouri2021proximal} and the population analogue of the profiled
minimax loss of \citet{dikkala2020minimax}.

Let
\(\{O_i\}_{i=1}^{n_\mathrm{val}}\) denote a validation sample independent
of the sample used to fit a candidate bridge. We estimate
\(\mathcal J_\mu(\mu)\) by the V-statistic 
    \begin{equation*}
    \widehat{\mathcal V}_\mu(\mu)
    =
    \frac{1}{n_{\mathrm{val}}^2}
    \sum_{i=1}^{n_{\mathrm{val}}}
    \sum_{j=1}^{n_{\mathrm{val}}}
    D_\mu(O_i;\mu)\,
    k'_\mu\bigl((Z_i,X_i),(Z_j,X_j)\bigr)\,
    D_\mu(O_j;\mu),
    \end{equation*}
and define \(\widehat{\mathcal V}_q(q)\) analogously. We select the
hyperparameters that minimize the corresponding validation criterion.

	\subsubsection{Regression nuisances}

We estimate all regression nuisance functions using feed-forward neural
networks. These include the pseudo-outcome regressions, the instrument propensity and auxiliary regressions in the proximal IV estimator when \(\pi_0\) is unknown; the outcome and
first-stage regressions in \eqref{eq:app-wald-score}; and the treatment propensity and outcome regressions in \eqref{eq:app-naive-aipw}. For each regression, the network architecture and weight decay parameter are selected from a discrete grid by minimizing held-out squared-error loss for continuous responses or the held-out Brier score for binary responses. The candidate architectures have one or two hidden layers.

    \subsection{Robustness to nuisance misspecification}
    \label{sec:app-robustness}

\subsubsection{Experiment Design}
We assess the finite-sample behavior implied by
Proposition~\ref{prop:nested-robustness} under the simulation setup in Section~\ref{sec:main-sim-setup}, using \(n=2{,}000\) and \(100\) Monte Carlo replications. We consider all \(2^4=16\) combinations of correct and misspecified models for the outcome bridge \(\mu\), the adjoint bridge \(q\), the instrument propensity \(\pi\), and the pseudo-outcome regressions \(\nu=(\nu_0,\nu_1)\). To represent correct specification, we use Gaussian RKHS minimax learners for the bridges and feed-forward neural networks for the propensity and pseudo-outcome regressions; deliberate misspecification restricts the corresponding nuisance function to a constant. To assess the two robustness margins separately, we fit the bridge equations using the known propensity \(\pi_0\) throughout and vary the working propensity \(\pi\) only in the adjoint, doubly robust, and efficient estimating functions. 

\subsubsection{Result}

Figure~\ref{fig:sim-robust} reports the Monte Carlo distributions under the alternative nuisance specifications. The primal estimator remains centered near the population ATE when \(\mu\) is flexibly specified; the adjoint estimator does so when \(q\) and \(\pi\) are flexibly specified; the doubly robust estimator does so when \(\pi\) and at least one bridge are flexibly specified; and the efficient estimator does so when at least one bridge and either \(\pi\) or \(\nu\) are flexibly specified. In particular, when \(q\) and \(\nu\) are flexible but \(\mu\) and \(\pi\) are restricted to constants, only the efficient estimator remains near the ATE, which showcases the additional robustness supplied by the pseudo-outcome regressions. When both bridges are restricted to constants, none of the estimators is centered near the ATE. Overall, these results provide finite-sample support for the nested robustness structure under nuisance misspecification.

    \begin{figure}
        \centering
        \includegraphics[width=1\linewidth]{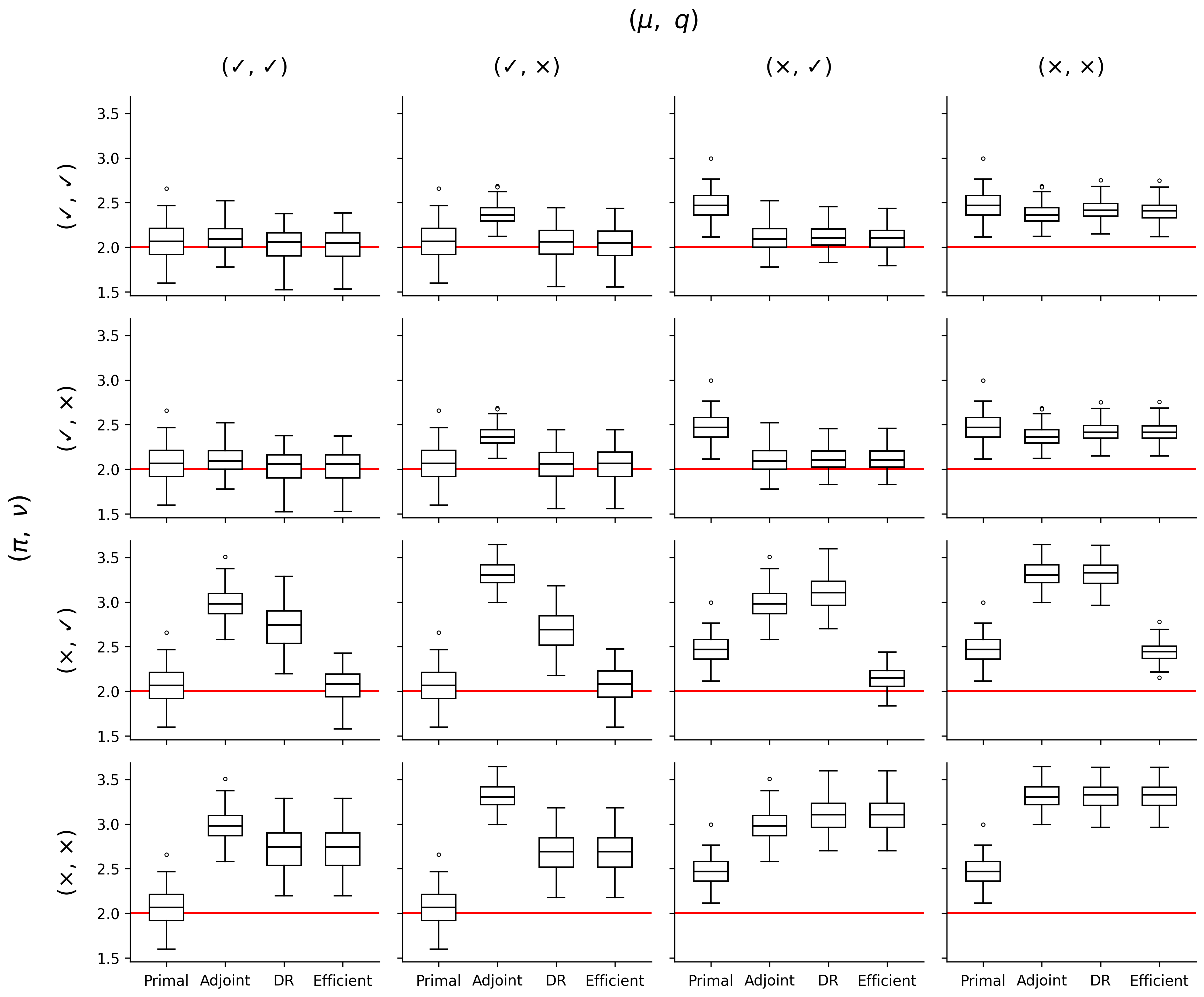}
        \caption{Robustness to nuisance misspecification. Columns index the specifications of \((\mu,q)\), and rows index specifications of \((\pi,\nu)\), with \(\checkmark\) and \(\times\) denoting correct and misspecified nuisance models, respectively. Each panel reports the four estimators over \(100\) Monte Carlo replications with \(n=2{,}000\); the red line denotes the true ATE.}
        \label{fig:sim-robust}
    \end{figure}

%% file: large_sample_appendix.tex
\section{Large Sample Theory}
\label{app:large-sample}

	\subsection{Technical setup and assumptions}
	\label{subsec:appendix-local-process}

Let \(\mathcal I_1,\ldots,\mathcal I_L\) denote the outer evaluation folds,
and, for each outer fold \(\ell\), let \(\mathcal T_\ell\) and \(\mathcal J_\ell\)
denote the bridge-training and generated-outcome regression samples,
respectively. Let \(n_\ell=|\mathcal I_\ell|\) and let
\(\mathscr F_{\mathcal T_\ell}\) denote the sigma-field generated by
\(\mathcal T_\ell\). For each \(\ell\), define \(\widehat H_\ell = \widehat q_\ell(Z,X) \{Y-A\widehat\mu_\ell(W,X)\}\),
and, for \(r\in\{0,1\}\), \(\bar\nu_{r,\ell}(x) = \E\!\left\{ \widehat H_\ell \mid R=r,X=x,\mathscr F_{\mathcal T_\ell} \right\}\).
The conditioning on \(\mathscr F_{\mathcal T_\ell}\) treats the bridge
estimators, and hence the generated outcome \(\widehat H_\ell\), as fixed.
Because \(\mathcal J_\ell\) is disjoint from \(\mathcal T_\ell\), observations
used to regress \(\widehat H_\ell\) on \((R,X)\) are independent of \(\mathscr F_{\mathcal T_\ell}\).

Let the bridge-training set
\(\mathcal T_\ell\) have size \(m\) and be partitioned into disjoint inner folds
\(\mathcal T_{\ell,1},\ldots,\mathcal T_{\ell,D}\), so that
\(\mathcal T_\ell=\cup_{d=1}^D\mathcal T_{\ell,d}\). For each \(d\), define
\(\mathcal T_{\ell,-d}=\mathcal T_\ell\setminus\mathcal T_{\ell,d}\). To simplify notation, for a fixed \(\ell\) we write \(\mathcal T_d=\mathcal T_{\ell,d}\),
\(\mathcal T_{-d}=\mathcal T_{\ell,-d}\), \(m_d=|\mathcal T_d|\), \(\omega_d=m_d/m\), and
\(\mathbb P_df=m_d^{-1}\sum_{i\in\mathcal T_d}f(O_i)\). For each \(d\), let
\(\mathscr F_{\mathcal T_{-d}}\) denote the sigma-field generated by the observations used to train the
preliminary nuisance functions evaluated on \(\mathcal T_d\). Conditional on
\(\mathscr F_{\mathcal T_{-d}}\), these fitted functions are fixed and the observations in
\(\mathcal T_d\) are i.i.d. from \(P_0\). The numbers of inner and outer folds are fixed, and we consider a fixed outer fold \(\ell\) and suppress its index until the proof of Theorem~\ref{thm:final-asymptotics}.

For a function class \(\mathcal F\), let
\(\operatorname{star}(\mathcal F)=\{tf:f\in\mathcal F,0\le t\le1\}\) and define its
localized Rademacher complexity at radius \(r\) and sample size \(s\) by
	\[
	\mathcal R_s(r,\mathcal F)
	=
	\E\left[
	\sup_{f\in\operatorname{star}(\mathcal F):\|f\|_2\le r}
	\left|\frac1s\sum_{i=1}^s\epsilon_i f(O_i)\right|
	\right],
	\]
where \(\epsilon_1,\ldots,\epsilon_s\) are independent Rademacher variables \citep{bartlett2002rademacher}. A
critical radius is any positive \(r\) that upper-bounds the localized Rademacher complexity, i.e., \(r_s(\mathcal F) = \inf\left\{ r>0: \mathcal R_s(r,\mathcal F)\le r^2 \right\}\). 

We first collect basic boundedness and overlap conditions used throughout the
cross-fitted empirical process arguments and the asymptotic analysis.

	\begin{assumption}[Nuisance regularities]
		\label{ass:large-sample-regularity}
The following conditions hold.

		\begin{enumerate}
			
			\item \(\|Y\|_\infty<\infty\), and \(\min\left\{ \|\widehat q_\ell\|_\infty+\|\mu_0\|_\infty,\, \|q_0\|_\infty+\|\widehat\mu_\ell\|_\infty \right\} = O_p(1)\).
			
			\item
			\(\underline\pi \le \widehat\pi_\ell(X) \le 1-\underline\pi\) almost surely.
			
			\item \(\|\widehat\pi_d\|_\infty + \|\widehat\gamma_d\|_\infty + \|\widehat\alpha_d\|_\infty = O(1)\).
			
		\end{enumerate}
	\end{assumption}
These conditions ensure uniformly controlled score envelopes and stable
inverse-propensity weights. Condition (iii) of
Assumption~\ref{ass:large-sample-regularity} is natural under our
implementation: since \(Y\) is assumed bounded, we clip
\(\widehat\gamma_d\) to a fixed bounded range, while the sigmoid output layers
for neural networks \(\widehat\pi_d\) and \(\widehat\alpha_d\) already ensure uniformly bounded
predictions.

We next impose a compatibility condition linking bridge errors to their induced
critic functions. 	Fix constants \(B_\mu,B_q<\infty\) and define \(\mathcal H_{\mu,B_\mu}=\{\mu\in\mathcal H_\mu:\|\mu\|_{\mathcal H_\mu}^2\le B_\mu\}\) and \(\mathcal H_{q,B_q}=\{q\in\mathcal H_q:\|q\|_{\mathcal H_q}^2\le B_q\}\). Let \(\sigma_{R,0}^2(X)=\operatorname{Var}(R\mid X)=\pi_0(X)\{1-\pi_0(X)\}\).

	\begin{assumption}[Bridge--critic compatibility]
		\label{ass:ls-classes}
The true bridges satisfy \(\mu_0\in\mathcal H_{\mu,B_\mu}\) and \(q_0\in\mathcal H_{q,B_q}\). There exist fixed constants \(L_\mu,L_q<\infty\) such that, for every \(\mu\in\mathcal H_\mu\), \(\sigma_{R,0}^2T(\mu-\mu_0)\in\mathcal G_\mu\) and \(\|\sigma_{R,0}^2T(\mu-\mu_0)\|_{\mathcal G_\mu}\le L_\mu\|\mu-\mu_0\|_{\mathcal H_\mu}\), and, for every \(q\in\mathcal H_q\), \(\sigma_{R,0}^2T^\ast(q-q_0)\in\mathcal G_q\) and \(\|\sigma_{R,0}^2T^\ast(q-q_0)\|_{\mathcal G_q}\le L_q\|q-q_0\|_{\mathcal H_q}\). For the known-propensity criteria, the same conditions hold with \(\sigma_{R,0}^2T\) and \(\sigma_{R,0}^2T^\ast\) replaced by \(T\) and \(T^\ast\), respectively.
	\end{assumption}

Assumption~\ref{ass:ls-classes} is the quantitative closedness condition used in the
regularized minimax analysis of \citet{dikkala2020minimax}. It ensures
that the conditional projection of each bridge error is an admissible critic and
controls the critic RKHS penalty evaluated at that projection. Since
\(\mathcal G_\mu\) and \(\mathcal G_q\) are RKHSs, symmetry and star-convexity around
zero are automatic.

Controlling the nuisance perturbation terms additionally requires an
interpolation inequality between the critic RKHS norm, its \(L_2(P_0)\) norm,
and its supremum norm. We obtain this inequality from the following spectral
regularity condition.

	\begin{assumption}[Critic kernel spectral regularity]
		\label{ass:ls-critic-spectral}
Let \((\sigma_{\mu,j},\beta_{\mu,j})_{j\ge1}\) denote the eigenvalue--eigenfunction pairs of the kernel integral operator associated with \(\mathcal G_\mu\) on \(L_2(P_{Z,X})\), with eigenfunctions orthonormal in \(L_2(P_{Z,X})\). There exists an exponent \(\vartheta_\mu\in(0,1)\) such that \(\sup_{j\ge1}j^{1/\vartheta_\mu}\sigma_{\mu,j}<\infty\) and \(\sup_{j\ge1}\|\beta_{\mu,j}\|_\infty<\infty\).
Similarly, let \((\sigma_{q,j},\beta_{q,j})_{j\ge1}\) denote the eigenvalue--eigenfunction pairs of the kernel integral operator associated with \(\mathcal G_q\) on \(L_2(P_{W,X})\), with eigenfunctions orthonormal in \(L_2(P_{W,X})\). There exists an exponent \(\vartheta_q\in(0,1)\) such that \(\sup_{j\ge1}j^{1/\vartheta_q}\sigma_{q,j}<\infty\) and \(\sup_{j\ge1}\|\beta_{q,j}\|_\infty<\infty\).
	\end{assumption}

The polynomial eigenvalue decay and uniformly bounded eigenfunctions in
Assumption~\ref{ass:ls-critic-spectral} are the same type of spectral
regularity conditions imposed in the kernel analyses of the R-learner \citep{nie2021quasi} and the proximal
P-learner \citep{sverdrup2023proximal}. Under Assumption~\ref{ass:ls-critic-spectral}, Lemma~5.1 of
\citet{mendelson2010regularization}, applied separately to the two
critic RKHSs, yields
	\begin{equation}
		\begin{aligned}
			\|g\|_\infty
			&\lesssim
			\|g\|_{\mathcal G_\mu}^{\vartheta_\mu}
			\|g\|_2^{1-\vartheta_\mu},
			&& g\in\mathcal G_\mu,\\
			\|g\|_\infty
			&\lesssim
			\|g\|_{\mathcal G_q}^{\vartheta_q}
			\|g\|_2^{1-\vartheta_q},
			&& g\in\mathcal G_q.
		\end{aligned}
		\label{eq:critic-interpolation}
	\end{equation}

Similar to \citet{dikkala2020minimax}, we next assume uniformly bounded
envelopes on the RKHS balls over which the complexity bounds are applied.

	\begin{assumption}[Uniform boundedness]
		\label{ass:ls-sampling}
The classes \(\mathcal H_{\mu,B_\mu}\), \(\mathcal H_{q,B_q}\),
\(\mathcal G_{\mu,2L_\mu^2B_\mu}\), and \(\mathcal G_{q,2L_q^2B_q}\)
are classes of \(b\)-uniformly bounded functions for some fixed
\(b<\infty\). Note that a function class \(\mathcal F\) is \(b\)-uniformly bounded if \(\sup_{f\in\mathcal F}\|f\|_\infty\le b\).
	\end{assumption}

While identification requires only weak compliance positivity, i.e.,
\(c(U,X)>0\) almost surely, for the kernel-based asymptotic theory,
Assumptions~\ref{ass:ls-classes} and~\ref{ass:ls-sampling} additionally imply
\(\|q_0\|_\infty\leq b\). By the latent compliance bridge restriction,
		\[
		\frac{1}{c(U,X)}
		=
		\E\{q_0(Z,X)\mid U,X\}
		\leq
		\E\{|q_0(Z,X)|\mid U,X\}
		\leq b,
		\]
and hence \(c(U,X)\geq b^{-1}\) almost surely. Thus, identification allows
\(c(U,X)\) to approach zero, whereas the asymptotic analysis relies on the stronger uniform
compliance overlap condition through boundedness of the compliance bridge.

We now specify the function classes whose localized complexities determine the
stochastic component of the bridge estimation rates. Consider the critic balls
\(\mathcal G_{\mu,2L_\mu^2B_\mu}\) and
\(\mathcal G_{q,2L_q^2B_q}\), and define the tied bridge--critic classes
	\[
	\begin{aligned}
		\mathcal F_{\mu,B_\mu}
		&=
		\operatorname{star}\Big\{
		O\mapsto h(W,X)[\sigma_{R,0}^2Th](Z,X):
		h\in\mathcal H_{\mu,B_\mu}
		\Big\},\\
		\mathcal F_{q,B_q}
		&=
		\operatorname{star}\Big\{
		O\mapsto h(Z,X)[\sigma_{R,0}^2T^\ast h](W,X):
		h\in\mathcal H_{q,B_q}
		\Big\}.
	\end{aligned}
	\]
For the known-propensity criteria, the corresponding tied classes are defined
after removing the factor \(\sigma_{R,0}^2\).

The next assumption summarizes the complexity of these classes through
deterministic upper bounds on their critical radii.

	\begin{assumption}[Critical radii]
		\label{ass:ls-complexity}
The sequences \(\delta^\circ_{\mu,n}\) and \(\delta^\circ_{q,n}\) converge to zero. For every sample size proportional to \(n\), \(\delta^\circ_{\mu,n}\) upper-bounds the critical radii of \(\mathcal G_{\mu,2L_\mu^2B_\mu}\) and \(\mathcal F_{\mu,B_\mu}\), while \(\delta^\circ_{q,n}\) upper-bounds the critical radii of \(\mathcal G_{q,2L_q^2B_q}\) and \(\mathcal F_{q,B_q}\). 
	\end{assumption}

Finally, for later conversion of projected bridge rates into ordinary
\(L_2(P_0)\) rates, we impose a high-probability localization condition on the
estimated bridge errors (this is not required for proving
Theorem~\ref{thm:bridge-rates}).

	\begin{assumption}[Bridge-error localization]
		\label{ass:bridge-error-localization}
There exist fixed constants \(C_\mu,C_q<\infty\) such that
		\[
		\Pr\left[
		\max_{1\le\ell\le L}
		\|\widehat\mu_\ell-\mu_0\|_{\mathcal H_\mu}^2
		\le C_\mu,\quad
		\max_{1\le\ell\le L}
		\|\widehat q_\ell-q_0\|_{\mathcal H_q}^2
		\le C_q
		\right]
		\longrightarrow1.
		\]
	\end{assumption}

\subsection{Projected bridge-risk bounds}
	\label{subsec:appendix-projected-rates}

We first introduce two localized empirical process results used in the proof.

	\begin{lemma}[Local empirical-norm comparison]
		\label{lem:local-empirical-norm}
Let \(\mathcal F\) be a measurable, star-shaped, uniformly bounded function
class, let \(\mathbb P_s\) be the empirical measure of an i.i.d. sample of size
\(s\) from \(P_0\), and let \(r_s\) upper-bound the critical radius of \(\mathcal F\).
There exist fixed constants \(c_0,c_1>0\) such that, for every
\(\zeta\in(0,1)\), if
		\[
		\bar r_s(\zeta)
		=
		r_s+c_0\sqrt{\frac{\log(c_1/\zeta)}{s}},
		\]
then, with probability at least \(1-\zeta\),
		\[
		\left|
		\mathbb P_sf^2-P_0f^2
		\right|
		\le
		\frac12\|f\|_2^2+
		\frac12\bar r_s(\zeta)^2
		\quad
		\text{for every }f\in\mathcal F.
		\]
	\end{lemma}

Lemma~\ref{lem:local-empirical-norm} is the local empirical-norm comparison
of \citet{wainwright2019high}.

	\begin{lemma}[Localized Lipschitz-loss deviation]
		\label{lem:localized-lipschitz-deviation}
Let \(\mathcal F\) be a measurable, uniformly bounded function class, fix
\(f_0\in\mathcal F\), and let \(r_s\) upper-bound the critical radius of
\(\operatorname{star}(\mathcal F-f_0)\). Suppose that
\(f\mapsto\ell_f(O)\) is Lipschitz in its function-valued argument with a fixed
finite Lipschitz constant. There exist fixed constants \(c_0,c_1,K>0\) such
that, for every \(\zeta\in(0,1)\), if
\(\bar r_s(\zeta)=r_s+c_0\sqrt{\log(c_1/\zeta)/s}\), then, with probability at
least \(1-\zeta\),
		\[
		\left|
		(\mathbb P_s-P_0)
		\{\ell_f-\ell_{f_0}\}
		\right|
		\le
		K\left\{
		\bar r_s(\zeta)\|f-f_0\|_2+
		\bar r_s(\zeta)^2
		\right\}
		\quad
		\text{for every }f\in\mathcal F.
		\]
	\end{lemma}

Lemma~\ref{lem:localized-lipschitz-deviation} is the localized
Lipschitz-loss deviation result of \citet{foster2023orthogonal}.

We use these two lemmas to convert the critical radius bounds in
Assumption~\ref{ass:ls-complexity} into high-probability complexity radii that
hold uniformly over the fixed collection of folds. By increasing the constants
if necessary, take \(c_0\) and \(c_1\) to be common to
Lemmas~\ref{lem:local-empirical-norm} and~\ref{lem:localized-lipschitz-deviation}.
By the fixed-fold construction, there exists a deterministic sequence
\(m_n\asymp n\) such that every relevant inner-fold sample size satisfies
\(m_d\ge m_n\). Set \(\zeta_n=n^{-2}\) and define
	\begin{align*}
		\delta_{\mu,n}
		&=
		\delta^\circ_{\mu,n}
		+
		c_0\sqrt{\frac{\log(c_1/\zeta_n)}{m_n}},
		\\
		\delta_{q,n}
		&=
		\delta^\circ_{q,n}
		+
		c_0\sqrt{\frac{\log(c_1/\zeta_n)}{m_n}}.
	\end{align*}
Since \(\zeta_n=n^{-2}\) and \(m_n\asymp n\), the concentration terms in \(\delta_{\mu,n}\) and \(\delta_{q,n}\) are \(O\{\sqrt{\log n/n}\}\), and hence \(\delta_{\mu,n}\to0\) and \(\delta_{q,n}\to0\).

For estimated propensity scores, the bridge-rate benchmark must additionally
account for the second-order perturbation induced by the preliminary nuisance
estimation. For the unknown-propensity criteria, denote \(\chi_{\mu,n} = \delta_{\mu,n} + [ \varepsilon_{\pi,n} (\varepsilon_{\gamma,n}+\varepsilon_{\alpha,n}) ]^{1/(1+\vartheta_\mu)}\)
and \(\chi_{q,n} = \delta_{q,n} + [ \varepsilon_{\pi,n} (\varepsilon_{\pi,n}+\varepsilon_{\alpha,n}) ]^{1/(1+\vartheta_q)}\).
For the known-propensity criteria, the corresponding projected-rate
benchmarks are \(\delta_{\mu,n}\) and \(\delta_{q,n}\).

We then impose restrictions on the penalty coefficients.

	\begin{assumption}[Penalty lower bounds]
		\label{ass:ls-penalty-lower}
The critic penalties satisfy
\(\kappa_{\mu,n}\ge\chi_{\mu,n}^2/(2L_\mu^2B_\mu)\) and
\(\kappa_{q,n}\ge\chi_{q,n}^2/(2L_q^2B_q)\), while the bridge penalties satisfy
\(\lambda_{\mu,n}\ge C_{\lambda,\mu}L_\mu^2\kappa_{\mu,n}\) and
\(\lambda_{q,n}\ge C_{\lambda,q}L_q^2\kappa_{q,n}\), where
\(C_{\lambda,\mu},C_{\lambda,q}>0\) are sufficiently large fixed constants.
	\end{assumption}

	\begin{assumption}[Penalty calibration]
		\label{ass:ls-penalty-calibration}
In addition to Assumption~\ref{ass:ls-penalty-lower}, \(\lambda_{\mu,n}\|\mu_0\|_{\mathcal H_\mu}^2\lesssim\chi_{\mu,n}^2\) and \(\lambda_{q,n}\|q_0\|_{\mathcal H_q}^2\lesssim\chi_{q,n}^2\). 
	\end{assumption}

These restrictions play similar roles as the regularization conditions in
Theorem~1 of \citet{dikkala2020minimax}. The critic penalty absorbs the discrepancy between empirical and population norms, while the bridge penalty absorbs the
RKHS norm of the tied critic and the quadratic terms produced by rescaling bridge
errors outside the benchmark ball. The upper calibration ensures that the
resulting regularization bias term remains of the same order as the target bridge rate.

Now we can prove Theorem~\ref{thm:bridge-rates}.

	\begin{proof}[Proof of Theorem~\ref{thm:bridge-rates}]
Fix an outer fold and suppress its index. Let the bridge-training sample
\(\mathcal T\) have size \(m=|\mathcal T|\), and partition it into inner folds
\(\mathcal T_1,\ldots,\mathcal T_D\). Recall that
\(m_d=|\mathcal T_d|\), \(\omega_d=m_d/m\), and write
\(\mathbb P_df=m_d^{-1}\sum_{i\in\mathcal T_d}f(O_i)\). For inner fold
\(d\), write
\(\xi_{\mu,d}(\mu) =\xi_\mu(\mu;\widehat\pi_d,\widehat\gamma_d,\widehat\alpha_d)\) and
\(\xi_{q,d}(q) =\xi_q(q;\widehat\pi_d,\widehat\alpha_d)\), where
		\begin{align*}
			\xi_\mu(\mu;\pi,\gamma,\alpha)
			&=
			\{R-\pi(X)\}
			\{Y-\gamma(Z,X)-[A-\alpha(W,Z,X)]\mu(W,X)\},\\
			\xi_q(q;\pi,\alpha)
			&=
			\{R-\pi(X)\}
			\{R-\pi(X)-[A-\alpha(W,Z,X)]q(Z,X)\}.
		\end{align*}
Also write \(\Delta_{\pi,d}=\widehat\pi_d-\pi_0\),
\(\Delta_{\gamma,d}=\widehat\gamma_d-\gamma_0\), and \(\Delta_{\alpha,d}=\widehat\alpha_d-\alpha_0\)
for the corresponding preliminary nuisance errors. For brevity, let
\(b_{\mu,n} =\varepsilon_{\pi,n} (\varepsilon_{\gamma,n}+\varepsilon_{\alpha,n})\) and
\(b_{q,n} =\varepsilon_{\pi,n} (\varepsilon_{\pi,n}+\varepsilon_{\alpha,n})\), and recall that
\(\chi_{\mu,n} =\delta_{\mu,n}+b_{\mu,n}^{1/(1+\vartheta_\mu)}\) and
\(\chi_{q,n} =\delta_{q,n}+b_{q,n}^{1/(1+\vartheta_q)}\). Hence \(b_{\mu,n}\le\chi_{\mu,n}^{1+\vartheta_\mu}\), \(b_{q,n}\le\chi_{q,n}^{1+\vartheta_q}\),
\(\delta_{\mu,n}\le\chi_{\mu,n}\), and \(\delta_{q,n}\le\chi_{q,n}\).

For the primal problem, define
		\begin{align*}
			\mathcal L_{\mu,n}(\mu,g)
			&=
			\sum_{d=1}^D\omega_d
			\left[
			\mathbb P_d\{\xi_{\mu,d}(\mu)g\}
			-\frac12\mathbb P_dg^2
			\right]
			-\frac{\kappa_{\mu,n}}2\|g\|_{\mathcal G_\mu}^2.
		\end{align*}
Then \(\widehat\mu\) minimizes
\(\sup_{g\in\mathcal G_\mu}\mathcal L_{\mu,n}(\mu,g) +\lambda_{\mu,n}\|\mu\|_{\mathcal H_\mu}^2/2\).  Also define the centered primal game by
		\begin{align*}
			\mathcal C_{\mu,n}(g)
			&=
			\sum_{d=1}^D\omega_d
			\mathbb P_d
			\left[
			\{\xi_{\mu,d}(\widehat\mu)-\xi_{\mu,d}(\mu_0)\}g
			\right]
			-
			\sum_{d=1}^D\omega_d\mathbb P_dg^2
			-
			\kappa_{\mu,n}\|g\|_{\mathcal G_\mu}^2.
		\end{align*}

The proof then follows the same three-step analysis as in \citet{dikkala2020minimax} and
\citet{ghassami2022minimax}. First, we upper-bound
\(\sup_{g\in\mathcal G_\mu}\mathcal C_{\mu,n}(g)\) using optimality of
\(\widehat\mu\). Second, when \(\|g_\mu\|_2>\chi_{\mu,n}\), we lower-bound the
same quantity by evaluating it at a scaled version of the tied critic \(g_\mu\).
Third, we combine the two bounds. Before proving these steps, we establish the
empirical process inequalities used throughout.

		\paragraph{Preliminaries: empirical process bounds.}
Fix an inner fold \(d\) and condition on \(\mathscr F_{\mathcal T_{-d}}\).
Conditional on \(\mathscr F_{\mathcal T_{-d}}\), the preliminary functions are
fixed and the observations in \(\mathcal T_d\) are i.i.d. from \(P_0\).
By Assumption~\ref{ass:large-sample-regularity}, the relevant residual
multipliers have fixed finite envelopes, uniformly in \(d\). Since \(m_d\ge m_n\),
		\begin{align*}
			\delta^\circ_{\mu,n}
			+c_0\sqrt{\frac{\log(c_1/\zeta_n)}{m_d}}
			&\le \delta_{\mu,n},
			\\
			\delta^\circ_{q,n}
			+c_0\sqrt{\frac{\log(c_1/\zeta_n)}{m_d}}
			&\le \delta_{q,n}.
		\end{align*}
By Assumption~\ref{ass:ls-complexity}, \(\delta^\circ_{\mu,n}\) upper-bounds the
critical radius of \(\mathcal G_{\mu,2L_\mu^2B_\mu}\). Hence,
Lemma~\ref{lem:local-empirical-norm}, applied conditionally on
\(\mathscr F_{\mathcal T_{-d}}\) with failure probability \(\zeta_n\), implies that
with conditional probability at least \(1-\zeta_n\),
		\[
		\left|
		\mathbb P_dg^2-P_0g^2
		\right|
		\le
		\frac12\|g\|_2^2+
		\frac12\delta_{\mu,n}^2
		\]
simultaneously for all
\(g\in\mathcal G_{\mu,2L_\mu^2B_\mu}\). We next extend this inequality to the
full critic space. For \(g\in\mathcal G_\mu\) satisfying
\(\|g\|_{\mathcal G_\mu}^2>2L_\mu^2B_\mu\), let
\(s_g=\sqrt{2L_\mu^2B_\mu}/\|g\|_{\mathcal G_\mu}\) and
\(\widetilde g=s_gg\). Then \(s_g\in(0,1)\) and
\(\widetilde g\in\mathcal G_{\mu,2L_\mu^2B_\mu}\). Therefore,
		\[
		s_g^2
		\left|
		\mathbb P_dg^2-P_0g^2
		\right|
		=
		\left|
		\mathbb P_d\widetilde g^2-P_0\widetilde g^2
		\right|
		\le
		\frac12s_g^2\|g\|_2^2+
		\frac12\delta_{\mu,n}^2,
		\]
which gives
		\[
		\left|
		\mathbb P_dg^2-P_0g^2
		\right|
		\le
		\frac12\|g\|_2^2+
		\frac{\delta_{\mu,n}^2}
		{4L_\mu^2B_\mu}
		\|g\|_{\mathcal G_\mu}^2.
		\]
Combining the two inequalities gives, for every \(g\in\mathcal G_\mu\),
		\begin{equation}
		\begin{aligned}
			\left|
			\mathbb P_dg^2-P_0g^2
			\right|
			&\le
			\frac12\|g\|_2^2
			+
			\frac{\delta_{\mu,n}^2}
			{4L_\mu^2B_\mu}
			\|g\|_{\mathcal G_\mu}^2
			+
			\frac12\delta_{\mu,n}^2.
		\end{aligned}
		\label{eq:primal-critic-norm}
		\end{equation}

For the residual process, take
\(\ell_g(O)=\xi_{\mu,d}(\mu_0)g(Z,X)\) and \(g_0=0\).
Conditional on \(\mathscr F_{\mathcal T_{-d}}\), this loss is Lipschitz in \(g\),
and Assumption~\ref{ass:large-sample-regularity} gives a fixed finite
Lipschitz constant. Lemma~\ref{lem:localized-lipschitz-deviation}
therefore gives, with conditional probability at least \(1-\zeta_n\),
		\[
		\left|
		(\mathbb P_d-P_0)
		\{\xi_{\mu,d}(\mu_0)g\}
		\right|
		\lesssim
		\delta_{\mu,n}\|g\|_2
		+
		\delta_{\mu,n}^2
		\]
uniformly over \(g\in\mathcal G_{\mu,2L_\mu^2B_\mu}\).
For \(g\) outside the ball,
		\[
		s_g
		\left|
		(\mathbb P_d-P_0)
		\{\xi_{\mu,d}(\mu_0)g\}
		\right|
		=
		\left|
		(\mathbb P_d-P_0)
		\{\xi_{\mu,d}(\mu_0)\widetilde g\}
		\right|
		\lesssim
		\delta_{\mu,n}s_g\|g\|_2
		+
		\delta_{\mu,n}^2.
		\]
Combining the two regimes gives, for every \(g\in\mathcal G_\mu\),
		\begin{equation}
		\begin{aligned}
			\left|
			(\mathbb P_d-P_0)
			\{\xi_{\mu,d}(\mu_0)g\}
			\right|
			\lesssim
			\delta_{\mu,n}\|g\|_2
			+
			\delta_{\mu,n}^2
			+
			\frac{\delta_{\mu,n}^2}
			{\sqrt{2L_\mu^2B_\mu}}
			\|g\|_{\mathcal G_\mu}.
		\end{aligned}
		\label{eq:primal-residual-process}
		\end{equation}

For \(h\in\mathcal H_\mu\), put \(g_h=\sigma_{R,0}^2Th\). If
\(\|h\|_{\mathcal H_\mu}^2\le B_\mu\), Assumption~\ref{ass:ls-sampling} gives a
fixed bound on \(\|h\|_\infty\), while
Assumption~\ref{ass:large-sample-regularity} gives a fixed bound on the
nuisance multiplier. Hence, by the Rademacher contraction principle, the complexity of the multiplier class
		\[
		O\mapsto
		\{R-\widehat\pi_d(X)\}
		\{A-\widehat\alpha_d(W,Z,X)\}
		h(W,X)g_h(Z,X)
		\]
is controlled by the tied class \(\mathcal F_{\mu,B_\mu}\) up to a fixed
constant. Lemma~\ref{lem:localized-lipschitz-deviation} thus gives, with
conditional probability at least \(1-\zeta_n\),
		\[
		\left|
		(\mathbb P_d-P_0)
		\left[
		\{R-\widehat\pi_d(X)\}
		\{A-\widehat\alpha_d(W,Z,X)\}
		h(W,X)g_h(Z,X)
		\right]
		\right|
		\lesssim
		\delta_{\mu,n}\|g_h\|_2
		+
		\delta_{\mu,n}^2.
		\]
If \(\|h\|_{\mathcal H_\mu}^2>B_\mu\), let \(s_h=\sqrt{B_\mu}/\|h\|_{\mathcal H_\mu}\),
\(\widetilde h=s_hh\), and
\(g_{\widetilde h}=s_hg_h\). Applying the preceding bound to
\(\widetilde h\) gives
		\begin{align*}
			s_h^2
			\left|
			(\mathbb P_d-P_0)
			\left[
			(R-\widehat\pi_d)
			(A-\widehat\alpha_d)hg_h
			\right]
			\right|
			&\lesssim
			\delta_{\mu,n}s_h\|g_h\|_2
			+
			\delta_{\mu,n}^2.
		\end{align*}
Combining the two regimes yields
		\begin{equation}
		\begin{aligned}
			&
			\left|
			(\mathbb P_d-P_0)
			\left[
			\{R-\widehat\pi_d(X)\}
			\{A-\widehat\alpha_d(W,Z,X)\}
			h(W,X)g_h(Z,X)
			\right]
			\right|
			\\
			\lesssim{}&
			\delta_{\mu,n}\|g_h\|_2
			\max\left\{
			1,
			\frac{\|h\|_{\mathcal H_\mu}}{\sqrt{B_\mu}}
			\right\}
			+
			\delta_{\mu,n}^2
			\max\left\{
			1,
			\frac{\|h\|_{\mathcal H_\mu}^2}{B_\mu}
			\right\}.
		\end{aligned}
		\label{eq:primal-tied-process}
		\end{equation}

The identical argument applied to
\(\mathcal G_{q,2L_q^2B_q}\) and \(\mathcal F_{q,B_q}\) gives the adjoint
counterparts of the preceding inequalities. In particular, with the same
conditional probability statements, for every \(g\in\mathcal G_q\),
		\begin{align}
			\left|
			\mathbb P_dg^2-P_0g^2
			\right|
			&\le
			\frac12\|g\|_2^2
			+
			\frac{\delta_{q,n}^2}
			{4L_q^2B_q}
			\|g\|_{\mathcal G_q}^2
			+
			\frac12\delta_{q,n}^2,
			\label{eq:adjoint-critic-norm}\\
			\left|
			(\mathbb P_d-P_0)
			\{\xi_{q,d}(q_0)g\}
			\right|
			&\lesssim
			\delta_{q,n}\|g\|_2
			+
			\delta_{q,n}^2
			+
			\frac{\delta_{q,n}^2}
			{\sqrt{2L_q^2B_q}}
			\|g\|_{\mathcal G_q}.
			\label{eq:adjoint-residual-process}
		\end{align}

For \(h\in\mathcal H_q\) and \(g_h=\sigma_{R,0}^2T^\ast h\),
		\begin{equation}
		\begin{aligned}
			&\left|
			(\mathbb P_d-P_0)
			\left[
			\{R-\widehat\pi_d(X)\}
			\{A-\widehat\alpha_d(W,Z,X)\}
			h(Z,X)g_h(W,X)
			\right]
			\right|
			\\
			\lesssim{}&
			\delta_{q,n}\|g_h\|_2
			\max\left\{
			1,
			\frac{\|h\|_{\mathcal H_q}}{\sqrt{B_q}}
			\right\}
			+
			\delta_{q,n}^2
			\max\left\{
			1,
			\frac{\|h\|_{\mathcal H_q}^2}{B_q}
			\right\}.
		\end{aligned}
		\label{eq:adjoint-tied-process}
		\end{equation}

For each inner fold, let
\(\mathcal E_{\ell,d,j}\) denote the event that the \(j\)-th empirical process
bound holds, for \(j=1,\ldots,6\), and define
		\(\mathcal E_{\ell,d} = \bigcap_{j=1}^6 \mathcal E_{\ell,d,j}\), which is the event that all six bounds hold simultaneously. Conditional on the corresponding training sample, a union bound gives
		\[
		P\!\left(
		\mathcal E_{\ell,d}^c
		\mid
		\mathscr F_{\mathcal T_{\ell,-d}}
		\right)
		\le
		\sum_{j=1}^6
		P\!\left(
		\mathcal E_{\ell,d,j}^c
		\mid
		\mathscr F_{\mathcal T_{\ell,-d}}
		\right)
		\le
		6\zeta_n,
		\]
and hence \(P\!\left( \mathcal E_{\ell,d} \mid \mathscr F_{\mathcal T_{\ell,-d}} \right) \ge 1-6\zeta_n\).
Taking expectations over the conditioning yields
		\[
		P(\mathcal E_{\ell,d}^c)
		=
		\E\!\left[
		P\!\left(
		\mathcal E_{\ell,d}^c
		\mid
		\mathscr F_{\mathcal T_{\ell,-d}}
		\right)
		\right]
		\le
		6\zeta_n.
		\]
Therefore, applying another union bound over the fixed collections of \(L\)
outer folds and \(D\) inner folds,
		\[
		P\!\left(
		\bigcup_{\ell=1}^L\bigcup_{d=1}^D
		\mathcal E_{\ell,d}^c
		\right)
		\le
		\sum_{\ell=1}^L\sum_{d=1}^D
		P(\mathcal E_{\ell,d}^c)
		\le
		6LD\zeta_n,
		\]
so all six empirical process bounds hold simultaneously over all relevant
folds with probability at least \(1-6LD\zeta_n\). Since \(\zeta_n=n^{-2}\) and
\(L\) and \(D\) are fixed, \(6LD\zeta_n\to0\).
Consequently, there exists an event \(\mathcal E_n\) with
\(P(\mathcal E_n)\to1\) on which all six empirical process bounds hold
simultaneously over \(\ell,d\) with the inflated critical radii \(\delta_{\mu,n}\) and
\(\delta_{q,n}\). We work on \(\mathcal E_n\) throughout the remainder of the
proof.

		\paragraph{Step 1: upper bound for the centered primal game.}
At the true primal bridge, \eqref{eq:primal-orthogonality-identity}, \eqref{eq:critic-interpolation},
conditional Jensen's inequality, and Cauchy--Schwarz give,
uniformly in \(d\) and \(g\in\mathcal G_\mu\),
		\begin{equation}
		\begin{aligned}
			\left|P_0\{\xi_{\mu,d}(\mu_0)g\}\right|
			&=
			\left|P_0\left[
			\E\{\xi_{\mu,d}(\mu_0)\mid Z,X\}g(Z,X)
			\right]\right|
			\\
			&=
			\left|P_0\left[
			\Delta_{\pi,d}(X)
			\left\{
			\Delta_{\gamma,d}(Z,X)
			-
			\E[\Delta_{\alpha,d}(W,Z,X)\mu_0(W,X)\mid Z,X]
			\right\}
			g(Z,X)
			\right]\right|
			\\
			&\le
			\|\Delta_{\pi,d}\|_2
			\left\{
			\|\Delta_{\gamma,d}\|_2
			+
			\|\E(\Delta_{\alpha,d}\mu_0\mid Z,X)\|_2
			\right\}
			\|g\|_\infty
			\\
			&\le
			\|\Delta_{\pi,d}\|_2
			\left\{
			\|\Delta_{\gamma,d}\|_2
			+
			\|\mu_0\|_\infty\|\Delta_{\alpha,d}\|_2
			\right\}
			\|g\|_{\mathcal G_\mu}^{\vartheta_\mu}
			\|g\|_2^{1-\vartheta_\mu}
			\\
			&\lesssim
			b_{\mu,n}
			\|g\|_{\mathcal G_\mu}^{\vartheta_\mu}
			\|g\|_2^{1-\vartheta_\mu}.
		\end{aligned}
		\label{eq:primal-truth-perturbation}
		\end{equation}
Combining \eqref{eq:primal-truth-perturbation} with
\eqref{eq:primal-residual-process} gives
		\begin{align*}
			\left|\sum_{d=1}^D\omega_d
			\mathbb P_d\{\xi_{\mu,d}(\mu_0)g\}\right|
			&\le
			\sum_{d=1}^D\omega_d
			\left|P_0\{\xi_{\mu,d}(\mu_0)g\}\right|
			+
			\sum_{d=1}^D\omega_d
			\left|(\mathbb P_d-P_0)\{\xi_{\mu,d}(\mu_0)g\}\right|\\
			&\lesssim
			b_{\mu,n}
			\|g\|_{\mathcal G_\mu}^{\vartheta_\mu}
			\|g\|_2^{1-\vartheta_\mu}
			+
			\delta_{\mu,n}\|g\|_2
			+
			\delta_{\mu,n}^2+
			\frac{\delta_{\mu,n}^2}{\sqrt{2L_\mu^2B_\mu}}
			\|g\|_{\mathcal G_\mu}.
		\end{align*}

The nuisance product term can be absorbed into the quadratic criterion. Indeed,
by \(b_{\mu,n}\le\chi_{\mu,n}^{1+\vartheta_\mu}\), the critic penalty
lower bound in Assumption~\ref{ass:ls-penalty-lower}, the weighted arithmetic--geometric mean inequality, and Young's
inequality \(ab\le \varepsilon a^2+b^2/4\varepsilon\)
with \(\varepsilon=1/32\),
		\begin{align*}
			b_{\mu,n}
			\|g\|_{\mathcal G_\mu}^{\vartheta_\mu}
			\|g\|_2^{1-\vartheta_\mu}
			&\le
			\chi_{\mu,n}^{1+\vartheta_\mu}
			\|g\|_{\mathcal G_\mu}^{\vartheta_\mu}
			\|g\|_2^{1-\vartheta_\mu}\\
			&=
			\chi_{\mu,n}
			\left(
			\frac{\chi_{\mu,n}^2}{\kappa_{\mu,n}}
			\right)^{\vartheta_\mu/2}
			\left\{
			\sqrt{\kappa_{\mu,n}}\|g\|_{\mathcal G_\mu}
			\right\}^{\vartheta_\mu}
			\|g\|_2^{1-\vartheta_\mu}\\
			&\lesssim
			\chi_{\mu,n}
			\left\{
			\sqrt{\kappa_{\mu,n}}\|g\|_{\mathcal G_\mu}
			\right\}^{\vartheta_\mu}
			\|g\|_2^{1-\vartheta_\mu}\\
			&\le
			\chi_{\mu,n}
			\left\{
			\vartheta_\mu\sqrt{\kappa_{\mu,n}}\|g\|_{\mathcal G_\mu}
			+
			(1-\vartheta_\mu)\|g\|_2
			\right\}\\
			&\le
			\frac1{32}\|g\|_2^2
			+
			\frac{\kappa_{\mu,n}}{32}\|g\|_{\mathcal G_\mu}^2
			+
			K\chi_{\mu,n}^2.
		\end{align*}
The remaining terms are absorbed similarly. Using
\(\delta_{\mu,n}\le\chi_{\mu,n}\) and \(\kappa_{\mu,n}\ge \chi_{\mu,n}^2/(2L_\mu^2B_\mu)\),
		\begin{align*}
			&\delta_{\mu,n}\|g\|_2
			+
			\frac{\delta_{\mu,n}^2}{\sqrt{2L_\mu^2B_\mu}}
			\|g\|_{\mathcal G_\mu}
			+
			\delta_{\mu,n}^2
			\\
			={}&
			\delta_{\mu,n}\|g\|_2
			+
			\left\{
			\sqrt{\kappa_{\mu,n}}\|g\|_{\mathcal G_\mu}
			\right\}
			\left\{
			\frac{\delta_{\mu,n}^2}
			{\sqrt{2\kappa_{\mu,n}L_\mu^2B_\mu}}
			\right\}
			+
			\delta_{\mu,n}^2
			\\
			\le{}&
			\frac1{32}\|g\|_2^2
			+
			\frac{\kappa_{\mu,n}}{32}\|g\|_{\mathcal G_\mu}^2
			+
			K\frac{\delta_{\mu,n}^4}
			{\kappa_{\mu,n}L_\mu^2B_\mu}
			+
			K\delta_{\mu,n}^2
			\\
			\le{}&
			\frac1{32}\|g\|_2^2
			+
			\frac{\kappa_{\mu,n}}{32}\|g\|_{\mathcal G_\mu}^2
			+
			K\delta_{\mu,n}^2
			+
			K\frac{\delta_{\mu,n}^4}{\chi_{\mu,n}^2}
			\\
			\le{}&
			\frac1{32}\|g\|_2^2
			+
			\frac{\kappa_{\mu,n}}{32}\|g\|_{\mathcal G_\mu}^2
			+
			K\chi_{\mu,n}^2.
		\end{align*}
Consequently,
		\begin{equation}
		\begin{aligned}
			\left|\sum_{d=1}^D\omega_d
			\mathbb P_d\{\xi_{\mu,d}(\mu_0)g\}\right|
			&\le
			\frac1{16}\|g\|_2^2
			+
			\frac{\kappa_{\mu,n}}{16}\|g\|_{\mathcal G_\mu}^2
			+
			K\chi_{\mu,n}^2.
		\end{aligned}
		\label{eq:primal-linear-bound}
		\end{equation}

On the other hand, \eqref{eq:primal-critic-norm} gives
		\begin{equation}
		\begin{aligned}
			&
			-\frac12\sum_{d=1}^D\omega_d\mathbb P_dg^2
			-
			\frac{\kappa_{\mu,n}}2\|g\|_{\mathcal G_\mu}^2
			\\
			\le{}&
			-
			\frac12\sum_{d=1}^D\omega_d
			\left\{
			P_0g^2-
			\left|\mathbb P_dg^2-P_0g^2\right|
			\right\}
			-
			\frac{\kappa_{\mu,n}}2\|g\|_{\mathcal G_\mu}^2
			\\
			\le{}&
			-
			\frac12\sum_{d=1}^D\omega_d
			\left\{
			\|g\|_2^2
			-
			\frac12\|g\|_2^2
			-
			\frac{\delta_{\mu,n}^2}{4L_\mu^2B_\mu}
			\|g\|_{\mathcal G_\mu}^2
			-
			\frac12\delta_{\mu,n}^2
			\right\}
			-
			\frac{\kappa_{\mu,n}}2\|g\|_{\mathcal G_\mu}^2
			\\
			={}&
			-
			\frac14\|g\|_2^2
			-
			\left\{
			\frac{\kappa_{\mu,n}}2
			-
			\frac{\delta_{\mu,n}^2}{8L_\mu^2B_\mu}
			\right\}
			\|g\|_{\mathcal G_\mu}^2
			+
			\frac14\delta_{\mu,n}^2
			\\
			\le{}&
			-
			\frac14\|g\|_2^2
			-
			\frac{\kappa_{\mu,n}}4\|g\|_{\mathcal G_\mu}^2
			+
			K\chi_{\mu,n}^2.
		\end{aligned}
		\label{eq:primal-quadratic-bound}
		\end{equation}
The last inequality uses \(\delta_{\mu,n}\le\chi_{\mu,n}\) and the exact
critic penalty lower bound in Assumption~\ref{ass:ls-penalty-lower}, which give
		\[
		\frac{\delta_{\mu,n}^2}{8L_\mu^2B_\mu}
		\le
		\frac{\chi_{\mu,n}^2}{8L_\mu^2B_\mu}
		\le
		\frac{\kappa_{\mu,n}}4.
		\]

Combining \eqref{eq:primal-linear-bound} and
\eqref{eq:primal-quadratic-bound},
		\begin{equation}
		\begin{aligned}
			\sup_{g\in\mathcal G_\mu}
			\mathcal L_{\mu,n}(\mu_0,g)
			&=
			\sup_{g\in\mathcal G_\mu}
			\Bigg\{
			\sum_{d=1}^D\omega_d
			\mathbb P_d\{\xi_{\mu,d}(\mu_0)g\}
			-
			\frac12\sum_{d=1}^D\omega_d\mathbb P_dg^2
			-
			\frac{\kappa_{\mu,n}}2
			\|g\|_{\mathcal G_\mu}^2
			\Bigg\}
			\\
			&\le
			\sup_{g\in\mathcal G_\mu}
			\Bigg\{
			\left|
			\sum_{d=1}^D\omega_d
			\mathbb P_d\{\xi_{\mu,d}(\mu_0)g\}
			\right|
			-
			\frac12\sum_{d=1}^D\omega_d\mathbb P_dg^2
			-
			\frac{\kappa_{\mu,n}}2
			\|g\|_{\mathcal G_\mu}^2
			\Bigg\}
			\\
			&\le
			\sup_{g\in\mathcal G_\mu}
			\Bigg\{
			\frac1{16}\|g\|_2^2
			+
			\frac{\kappa_{\mu,n}}{16}
			\|g\|_{\mathcal G_\mu}^2
			+
			K\chi_{\mu,n}^2
			-
			\frac14\|g\|_2^2
			-
			\frac{\kappa_{\mu,n}}4
			\|g\|_{\mathcal G_\mu}^2
			+
			K\chi_{\mu,n}^2
			\Bigg\}
			\\
			&\le
			\sup_{g\in\mathcal G_\mu}
			\left\{
			-
			\frac3{16}\|g\|_2^2
			-
			\frac{3\kappa_{\mu,n}}{16}
			\|g\|_{\mathcal G_\mu}^2
			+
			K\chi_{\mu,n}^2
			\right\}
			\\
			&\le
			K\chi_{\mu,n}^2.
		\end{aligned}
		\label{eq:primal-truth-game}
		\end{equation}

We now use optimality. Since \(\widehat\mu\) minimizes the profiled penalized
game,
		\begin{align*}
			\sup_{g\in\mathcal G_\mu}\mathcal L_{\mu,n}(\widehat\mu,g)
			+
			\frac{\lambda_{\mu,n}}2\|\widehat\mu\|_{\mathcal H_\mu}^2
			&\le
			\sup_{g\in\mathcal G_\mu}\mathcal L_{\mu,n}(\mu_0,g)
			+
			\frac{\lambda_{\mu,n}}2\|\mu_0\|_{\mathcal H_\mu}^2.
		\end{align*}
Moreover, for every \(g\in\mathcal G_\mu\), symmetry of the RKHS gives
\(-g\in\mathcal G_\mu\), and direct expansion shows \(\mathcal C_{\mu,n}(g) = \mathcal L_{\mu,n}(\widehat\mu,g) + \mathcal L_{\mu,n}(\mu_0,-g)\).
Therefore,
		\begin{equation}
		\begin{aligned}
			\sup_{g\in\mathcal G_\mu}\mathcal C_{\mu,n}(g)
			&\le
			\sup_{g\in\mathcal G_\mu}\mathcal L_{\mu,n}(\widehat\mu,g)
			+
			\sup_{g\in\mathcal G_\mu}\mathcal L_{\mu,n}(\mu_0,-g)\\
			&=
			\sup_{g\in\mathcal G_\mu}\mathcal L_{\mu,n}(\widehat\mu,g)
			+
			\sup_{g\in\mathcal G_\mu}\mathcal L_{\mu,n}(\mu_0,g)\\
			&\le
			2\sup_{g\in\mathcal G_\mu}\mathcal L_{\mu,n}(\mu_0,g)
			+
			\frac{\lambda_{\mu,n}}2
			\left(
			\|\mu_0\|_{\mathcal H_\mu}^2
			-
			\|\widehat\mu\|_{\mathcal H_\mu}^2
			\right)\\
			&\lesssim
			\chi_{\mu,n}^2
			+
			\frac{\lambda_{\mu,n}}2
			\left(
			\|\mu_0\|_{\mathcal H_\mu}^2
			-
			\|\widehat\mu\|_{\mathcal H_\mu}^2
			\right).
		\end{aligned}
		\label{eq:primal-centered-upper}
		\end{equation}

		\paragraph{Step 2: lower bound from the scaled tied critic.}
Write
\(h_\mu=\widehat\mu-\mu_0\) and
\(g_\mu=\sigma_{R,0}^2Th_\mu\). If \(\|g_\mu\|_2\le\chi_{\mu,n}\), the desired projected-rate bound is
immediate. Suppose \(\|g_\mu\|_2>\chi_{\mu,n}\). Write
\(a_\mu=\chi_{\mu,n}/(2\|g_\mu\|_2)\), and for
\(h\in\mathcal H_\mu\) write
		\[
		M_{\mu,d}(h)=\{R-\widehat\pi_d(X)\}
		\{A-\widehat\alpha_d(W,Z,X)\}h(W,X).
		\]
By Assumption~\ref{ass:ls-classes},
\(g_\mu\in\mathcal G_\mu\), and symmetry and star-convexity imply that
\(g_\mu^\star=-a_\mu g_\mu\in\mathcal G_\mu\). Since \(\xi_{\mu,d}(\widehat\mu)-\xi_{\mu,d}(\mu_0)=-M_{\mu,d}(h_\mu)\),
		\begin{equation}
		\begin{aligned}
			\sup_{g\in\mathcal G_\mu}\mathcal C_{\mu,n}(g)
			&\ge
			\sum_{d=1}^D\omega_d\mathbb P_d
			\left[
			\{\xi_{\mu,d}(\widehat\mu)-\xi_{\mu,d}(\mu_0)\}
			g_\mu^\star
			\right]
			-
			\sum_{d=1}^D\omega_d\mathbb P_d(g_\mu^\star)^2
			-
			\kappa_{\mu,n}\|g_\mu^\star\|_{\mathcal G_\mu}^2
			\\
			&=
			a_\mu\sum_{d=1}^D\omega_d
			\mathbb P_d\{M_{\mu,d}(h_\mu)g_\mu\}
			-
			a_\mu^2\sum_{d=1}^D\omega_d\mathbb P_dg_\mu^2
			-
			\kappa_{\mu,n}a_\mu^2\|g_\mu\|_{\mathcal G_\mu}^2
			\\
			&=
			a_\mu\sum_{d=1}^D\omega_d
			P_0\{M_{\mu,d}(h_\mu)g_\mu\}
			+
			a_\mu\sum_{d=1}^D\omega_d
			(\mathbb P_d-P_0)\{M_{\mu,d}(h_\mu)g_\mu\}
			\\
			&\quad-
			a_\mu^2\sum_{d=1}^D\omega_d\mathbb P_dg_\mu^2
			-
			\kappa_{\mu,n}a_\mu^2\|g_\mu\|_{\mathcal G_\mu}^2.
		\end{aligned}
		\label{eq:primal-tied-evaluation}
		\end{equation}

We first control the population term. Let \(e_0=R-\pi_0(X)=\sigma_{R,0}^2(X)\rho_0(R,X)\). Since
\(R\perp(W,Z)\mid X\), \(\E(e_0\mid W,Z,X)=0\), while
\(\alpha_0(W,Z,X)=\E(A\mid W,Z,X)\) gives
\(\E(A-\alpha_0\mid W,Z,X)=0\). Hence, for every \(h\in\mathcal H_\mu\),
		\begin{equation}
		\begin{aligned}
			\E\{M_{\mu,d}(h)\mid Z,X\}
			&=
			\E\left[
			\{e_0-\Delta_{\pi,d}\}
			\{A-\alpha_0-\Delta_{\alpha,d}\}h
			\mid Z,X
			\right]
			\\
			&=
			\E\{e_0(A-\alpha_0)h\mid Z,X\}
			-
			\E(e_0\Delta_{\alpha,d}h\mid Z,X)
			\\
			&\quad-
			\Delta_{\pi,d}\E\{(A-\alpha_0)h\mid Z,X\}
			+
			\Delta_{\pi,d}\E(\Delta_{\alpha,d}h\mid Z,X)
			\\
			&=
			\E(e_0Ah\mid Z,X)
			+
			\Delta_{\pi,d}\E(\Delta_{\alpha,d}h\mid Z,X)
			\\
			&=
			[\sigma_{R,0}^2Th](Z,X)
			+
			\Delta_{\pi,d}\E(\Delta_{\alpha,d}h\mid Z,X).
		\end{aligned}
		\label{eq:primal-operator-identity}
		\end{equation}
Here the omitted terms vanish by iterated expectation:
		\begin{align*}
			\E(e_0\Delta_{\alpha,d}h\mid Z,X)
			&=
			\E\!\left[
			\Delta_{\alpha,d}h\,
			\E(e_0\mid W,Z,X)
			\mid Z,X
			\right]
			=0,\\
			\E(e_0\alpha_0h\mid Z,X)
			&=
			\E\!\left[
			\alpha_0h\,
			\E(e_0\mid W,Z,X)
			\mid Z,X
			\right]
			=0,\\
			\E\{(A-\alpha_0)h\mid Z,X\}
			&=
			\E\!\left[
			h\,
			\E(A-\alpha_0\mid W,Z,X)
			\mid Z,X
			\right]
			=0.
		\end{align*}
By \eqref{eq:primal-operator-identity}, taking \(h=h_\mu\),
multiplying by \(g_\mu\), and averaging over the inner folds gives
		\begin{align*}
			\sum_{d=1}^D\omega_d
			P_0\{M_{\mu,d}(h_\mu)g_\mu\}
			&=
			\sum_{d=1}^D\omega_d
			P_0\!\left[
			\E\{M_{\mu,d}(h_\mu)\mid Z,X\}
			g_\mu
			\right]
			\\
			&=
			\sum_{d=1}^D\omega_d
			P_0\!\left[
			\left\{
			g_\mu
			+
			\Delta_{\pi,d}
			\E(\Delta_{\alpha,d}h_\mu\mid Z,X)
			\right\}
			g_\mu
			\right]
			\\
			&=
			\sum_{d=1}^D\omega_d
			P_0(g_\mu^2)
			+
			\sum_{d=1}^D\omega_d
			P_0\!\left[
			\Delta_{\pi,d}g_\mu
			\E(\Delta_{\alpha,d}h_\mu\mid Z,X)
			\right]
			\\
			&=
			\|g_\mu\|_2^2
			+
			\sum_{d=1}^D\omega_d
			P_0(
			\Delta_{\pi,d}\Delta_{\alpha,d}h_\mu g_\mu).
		\end{align*}
For a generic \(h\in\mathcal H_\mu\), with
\(g_h=\sigma_{R,0}^2Th\), by
Assumptions~\ref{ass:ls-sampling},
\ref{ass:ls-classes}, and
\ref{ass:nuisance-rates}, together with
\eqref{eq:critic-interpolation},
		\begin{equation}
		\begin{aligned}
			\left|P_0(\Delta_{\pi,d}\Delta_{\alpha,d}h g_h)\right|
			&\le
			\|\Delta_{\pi,d}\|_2
			\|\Delta_{\alpha,d}\|_2
			\|h\|_\infty
			\|g_h\|_\infty
			\\
			&\lesssim
			\varepsilon_{\pi,n}\varepsilon_{\alpha,n}
			\|h\|_{\mathcal H_\mu}
			\|g_h\|_{\mathcal G_\mu}^{\vartheta_\mu}
			\|g_h\|_2^{1-\vartheta_\mu}
			\\
			&\le
			L_\mu^{\vartheta_\mu}
			\varepsilon_{\pi,n}\varepsilon_{\alpha,n}
			\|h\|_{\mathcal H_\mu}^{1+\vartheta_\mu}
			\|g_h\|_2^{1-\vartheta_\mu}
			\\
			&\lesssim
			\varepsilon_{\pi,n}\varepsilon_{\alpha,n}
			\|h\|_{\mathcal H_\mu}^{1+\vartheta_\mu}
			\|g_h\|_2^{1-\vartheta_\mu}.
		\end{aligned}
		\label{eq:primal-slope-perturbation}
		\end{equation}
Using \eqref{eq:primal-slope-perturbation},
\(\varepsilon_{\pi,n}\varepsilon_{\alpha,n}\le b_{\mu,n} \le\chi_{\mu,n}^{1+\vartheta_\mu}\), and \(\|g_\mu\|_2>\chi_{\mu,n}\),
		\begin{equation}
		\begin{aligned}
			a_\mu\sum_{d=1}^D\omega_d
			P_0\{M_{\mu,d}(h_\mu)g_\mu\}
			&=
			\frac{\chi_{\mu,n}}{2\|g_\mu\|_2}
			\left\{
			\|g_\mu\|_2^2
			+
			\sum_{d=1}^D\omega_d
			P_0(\Delta_{\pi,d}\Delta_{\alpha,d}h_\mu g_\mu)
			\right\}
			\\
			&\ge
			\frac12\chi_{\mu,n}\|g_\mu\|_2
			-
			K\frac{\chi_{\mu,n}
				\varepsilon_{\pi,n}\varepsilon_{\alpha,n}}
			{\|g_\mu\|_2^{\vartheta_\mu}}
			\|h_\mu\|_{\mathcal H_\mu}^{1+\vartheta_\mu}
			\\
			&\ge
			\frac12\chi_{\mu,n}\|g_\mu\|_2
			-
			K\chi_{\mu,n}^2
			\max\left\{
			1,
			\frac{\|h_\mu\|_{\mathcal H_\mu}}{\sqrt{B_\mu}}
			\right\}^{1+\vartheta_\mu}
			\\
			&\ge
			\frac12\chi_{\mu,n}\|g_\mu\|_2
			-
			K\chi_{\mu,n}^2
			-
			K\frac{\chi_{\mu,n}^2}{B_\mu}
			\|h_\mu\|_{\mathcal H_\mu}^2.
		\end{aligned}
		\label{eq:primal-tied-population-bound}
		\end{equation}

Next, \eqref{eq:primal-tied-process} and
\(\delta_{\mu,n}\le\chi_{\mu,n}<\|g_\mu\|_2\) give
		\begin{equation}
		\begin{aligned}
			&
			a_\mu
			\left|\sum_{d=1}^D\omega_d
			(\mathbb P_d-P_0)\{M_{\mu,d}(h_\mu)g_\mu\}\right|
			\\
			\lesssim{}&
			\frac{\chi_{\mu,n}}{\|g_\mu\|_2}
			\Bigg[
			\delta_{\mu,n}\|g_\mu\|_2
			\max\left\{
			1,
			\frac{\|h_\mu\|_{\mathcal H_\mu}}{\sqrt{B_\mu}}
			\right\}
			+
			\delta_{\mu,n}^2
			\max\left\{
			1,
			\frac{\|h_\mu\|_{\mathcal H_\mu}^2}{B_\mu}
			\right\}
			\Bigg]
			\\
			\le{}&
			K\left(
			\chi_{\mu,n}\delta_{\mu,n}
			+
			\frac{\chi_{\mu,n}\delta_{\mu,n}^2}{\|g_\mu\|_2}
			\right)
			\max\left\{
			1,
			\frac{\|h_\mu\|_{\mathcal H_\mu}^2}{B_\mu}
			\right\}
			\\
			\lesssim{}&
			\chi_{\mu,n}^2
			\max\left\{
			1,
			\frac{\|h_\mu\|_{\mathcal H_\mu}^2}{B_\mu}
			\right\}
			\\
			\le{}&
			K\chi_{\mu,n}^2
			+
			K\frac{\chi_{\mu,n}^2}{B_\mu}
			\|h_\mu\|_{\mathcal H_\mu}^2.
		\end{aligned}
		\label{eq:primal-tied-fluctuation-bound}
		\end{equation}

For the empirical quadratic term, \eqref{eq:primal-critic-norm} gives
		\begin{align*}
			\sum_{d=1}^D\omega_d\mathbb P_dg_\mu^2
			&\le
			\frac32\|g_\mu\|_2^2
			+
			\frac{\delta_{\mu,n}^2}{4L_\mu^2B_\mu}
			\|g_\mu\|_{\mathcal G_\mu}^2
			+
			\frac12\delta_{\mu,n}^2.
		\end{align*}
By \(\|g_\mu\|_2>\chi_{\mu,n}\ge\delta_{\mu,n}\) and
\(\|g_\mu\|_{\mathcal G_\mu}\le L_\mu\|h_\mu\|_{\mathcal H_\mu}\) (Assumption~\ref{ass:ls-classes}), 
		\begin{equation}
		\begin{aligned}
			a_\mu^2
			\sum_{d=1}^D\omega_d\mathbb P_d g_\mu^2
			&=
			\frac{\chi_{\mu,n}^2}
			{4\|g_\mu\|_2^2}
			\sum_{d=1}^D\omega_d\mathbb P_d g_\mu^2
			\\
			&\le
			\frac38\chi_{\mu,n}^2
			+
			\frac{\chi_{\mu,n}^2\delta_{\mu,n}^2}
			{16L_\mu^2B_\mu\|g_\mu\|_2^2}
			\|g_\mu\|_{\mathcal G_\mu}^2
			+
			\frac{\chi_{\mu,n}^2\delta_{\mu,n}^2}
			{8\|g_\mu\|_2^2}\\
			&\le
			\frac38\chi_{\mu,n}^2
			+
			\frac{\chi_{\mu,n}^2}
			{16L_\mu^2B_\mu}
			\|g_\mu\|_{\mathcal G_\mu}^2
			+
			\frac18\chi_{\mu,n}^2\\
			&\le \frac12\chi_{\mu,n}^2
			+
			\frac{\chi_{\mu,n}^2}
			{16L_\mu^2B_\mu}
			L_\mu^2\|h_\mu\|_{\mathcal H_\mu}^2\\
			&=
			\frac12\chi_{\mu,n}^2
			+
			\frac{\chi_{\mu,n}^2}
			{16B_\mu}
			\|h_\mu\|_{\mathcal H_\mu}^2\\
			&\lesssim
			\chi_{\mu,n}^2
			+
			\frac{\chi_{\mu,n}^2}{B_\mu}
			\|h_\mu\|_{\mathcal H_\mu}^2.
		\end{aligned}
		\label{eq:primal-tied-quadratic-bound}
		\end{equation}

Finally, the critic penalty satisfies
		\begin{equation}
		\begin{aligned}
			\kappa_{\mu,n}a_\mu^2\|g_\mu\|_{\mathcal G_\mu}^2=\frac{\kappa_{\mu,n}\chi_{\mu,n}^2}{4\|g_\mu\|_2^2}
			\|g_\mu\|_{\mathcal G_\mu}^2
			&\le
			\frac14\kappa_{\mu,n}L_\mu^2
			\|h_\mu\|_{\mathcal H_\mu}^2.
		\end{aligned}
		\label{eq:primal-tied-penalty-bound}
		\end{equation}

Substituting
\eqref{eq:primal-tied-population-bound},
\eqref{eq:primal-tied-fluctuation-bound},
\eqref{eq:primal-tied-quadratic-bound}, and
\eqref{eq:primal-tied-penalty-bound}
into \eqref{eq:primal-tied-evaluation}, and using
\(\chi_{\mu,n}^2/B_\mu\le 2L_\mu^2\kappa_{\mu,n}\) gives, for some
fixed constant \(K<\infty\),
		\begin{align*}
			\sup_{g\in\mathcal G_\mu}\mathcal C_{\mu,n}(g)
			\ge{}&
			\frac12\chi_{\mu,n}\|g_\mu\|_2
			-
			K\chi_{\mu,n}^2
			-
			K L_\mu^2\kappa_{\mu,n}
			\|h_\mu\|_{\mathcal H_\mu}^2.
		\end{align*}
By Assumption~\ref{ass:ls-penalty-lower},
\(\lambda_{\mu,n}\ge C_{\lambda,\mu}L_\mu^2\kappa_{\mu,n}\), where
\(C_{\lambda,\mu}\) is sufficiently large. In particular, taking
\(C_{\lambda,\mu}\ge 8K\) yields
		\[
		K L_\mu^2\kappa_{\mu,n}
		\|h_\mu\|_{\mathcal H_\mu}^2
		\le
		\frac{\lambda_{\mu,n}}8
		\|h_\mu\|_{\mathcal H_\mu}^2.
		\]
Since
		\[
		\|h_\mu\|_{\mathcal H_\mu}^2
		\le
		2\|\widehat\mu\|_{\mathcal H_\mu}^2
		+
		2\|\mu_0\|_{\mathcal H_\mu}^2,
		\]
we obtain
		\begin{equation}
		\begin{aligned}
			\sup_{g\in\mathcal G_\mu}\mathcal C_{\mu,n}(g)
			\ge{}&
			\frac12\chi_{\mu,n}\|g_\mu\|_2
			-
			K\chi_{\mu,n}^2
			-
			\frac{\lambda_{\mu,n}}4
			\left(
			\|\widehat\mu\|_{\mathcal H_\mu}^2
			+
			\|\mu_0\|_{\mathcal H_\mu}^2
			\right).
		\end{aligned}
		\label{eq:primal-centered-lower}
		\end{equation}

		\paragraph{Step 3: combine the two bounds.}
Combining \eqref{eq:primal-centered-upper} and
\eqref{eq:primal-centered-lower} gives
		\begin{align*}
			\frac12\chi_{\mu,n}\|g_\mu\|_2
			&\lesssim
			\chi_{\mu,n}^2
			+
			\frac34\lambda_{\mu,n}\|\mu_0\|_{\mathcal H_\mu}^2
			-
			\frac14\lambda_{\mu,n}\|\widehat\mu\|_{\mathcal H_\mu}^2\\
			&\le
			K\chi_{\mu,n}^2
			+
			\frac34\lambda_{\mu,n}\|\mu_0\|_{\mathcal H_\mu}^2.
		\end{align*}
Thus, when \(\|g_\mu\|_2>\chi_{\mu,n}\),
		\[
		\|\sigma_{R,0}^2T(\widehat\mu-\mu_0)\|_2=\|g_\mu\|_2
		=
		O_p\!\left(
		\chi_{\mu,n}
		+
		\frac{\lambda_{\mu,n}\|\mu_0\|_{\mathcal H_\mu}^2}
		{\chi_{\mu,n}}
		\right).
		\]

When \(\|g_\mu\|_2\le\chi_{\mu,n}\), the same bound holds trivially. 
Assumption~\ref{ass:ls-penalty-calibration} further implies that the second term is
\(O(\chi_{\mu,n})\). By uniform IV overlap,
\(\|\sigma_{R,0}^2T(\widehat\mu-\mu_0)\|_2\) is equivalent to
\(\|T(\widehat\mu-\mu_0)\|_2\), and therefore
		\[
		\|T(\widehat\mu-\mu_0)\|_2
		=O_p(\chi_{\mu,n}).
		\]

		\paragraph{Adjoint problem.}
We now prove the corresponding bound for the adjoint bridge. 
For the adjoint problem, define
		\begin{align*}
			\mathcal L_{q,n}(q,g)
			&=
			\sum_{d=1}^D\omega_d
			\left[
			\mathbb P_d\{\xi_{q,d}(q)g\}
			-\frac12\mathbb P_dg^2
			\right]
			-\frac{\kappa_{q,n}}2\|g\|_{\mathcal G_q}^2.
		\end{align*}
Then \(\widehat q\) minimizes
\(\sup_{g\in\mathcal G_q}\mathcal L_{q,n}(q,g) +\lambda_{q,n}\|q\|_{\mathcal H_q}^2/2\). Also define the centered
adjoint game by
		\begin{align*}
			\mathcal C_{q,n}(g)
			&=
			\sum_{d=1}^D\omega_d
			\mathbb P_d
			\left[
			\{\xi_{q,d}(\widehat q)-\xi_{q,d}(q_0)\}g
			\right]
			-
			\sum_{d=1}^D\omega_d\mathbb P_dg^2
			-
			\kappa_{q,n}\|g\|_{\mathcal G_q}^2.
		\end{align*}

		\paragraph{Step 1: upper bound for the centered adjoint game.}
At the true adjoint bridge,
\eqref{eq:adjoint-orthogonality-identity},
\eqref{eq:critic-interpolation},
conditional Jensen's inequality, and Cauchy--Schwarz give,
uniformly in \(d\) and \(g\in\mathcal G_q\),
		\begin{equation}
		\begin{aligned}
			\left|P_0\{\xi_{q,d}(q_0)g\}\right|
			&=
			\left|P_0\left[
			\E\{\xi_{q,d}(q_0)\mid W,X\}g(W,X)
			\right]\right|
			\\
			&=
			\left|P_0\left[
			\Delta_{\pi,d}(X)
			\left\{
			\Delta_{\pi,d}(X)
			-
			\E[\Delta_{\alpha,d}(W,Z,X)q_0(Z,X)\mid W,X]
			\right\}
			g(W,X)
			\right]\right|
			\\
			&\le
			\|\Delta_{\pi,d}\|_2
			\left\{
			\|\Delta_{\pi,d}\|_2
			+
			\|q_0\|_\infty\|\Delta_{\alpha,d}\|_2
			\right\}
			\|g\|_{\mathcal G_q}^{\vartheta_q}
			\|g\|_2^{1-\vartheta_q}
			\\
			&\lesssim
			b_{q,n}
			\|g\|_{\mathcal G_q}^{\vartheta_q}
			\|g\|_2^{1-\vartheta_q}.
		\end{aligned}
		\label{eq:adjoint-truth-perturbation}
		\end{equation}
Combining \eqref{eq:adjoint-truth-perturbation} with
\eqref{eq:adjoint-residual-process}, and repeating the two absorption
arguments used in the primal problem gives
		\begin{equation}
		\begin{aligned}
			\left|\sum_{d=1}^D\omega_d
			\mathbb P_d\{\xi_{q,d}(q_0)g\}\right|
			&\le
			\frac1{16}\|g\|_2^2
			+
			\frac{\kappa_{q,n}}{16}\|g\|_{\mathcal G_q}^2
			+
			K\chi_{q,n}^2.
		\end{aligned}
		\label{eq:adjoint-linear-bound}
		\end{equation}

Similarly, \eqref{eq:adjoint-critic-norm} and the same quadratic
calculation as in the primal problem give
		\begin{equation}
		\begin{aligned}
			-\frac12\sum_{d=1}^D\omega_d\mathbb P_dg^2
			-
			\frac{\kappa_{q,n}}2\|g\|_{\mathcal G_q}^2
			\le
			-\frac14\|g\|_2^2
			-
			\frac{\kappa_{q,n}}4\|g\|_{\mathcal G_q}^2
			+
			K\chi_{q,n}^2.
		\end{aligned}
		\label{eq:adjoint-quadratic-bound}
		\end{equation}

Combining \eqref{eq:adjoint-linear-bound} and
\eqref{eq:adjoint-quadratic-bound},
		\begin{equation}
		\begin{aligned}
			\sup_{g\in\mathcal G_q}
			\mathcal L_{q,n}(q_0,g)
			\le
			\sup_{g\in\mathcal G_q}
			\left\{
			-\frac3{16}\|g\|_2^2
			-
			\frac{3\kappa_{q,n}}{16}
			\|g\|_{\mathcal G_q}^2
			+
			K\chi_{q,n}^2
			\right\}
			\le
			K\chi_{q,n}^2.
		\end{aligned}
		\label{eq:adjoint-truth-game}
		\end{equation}

We now use optimality. By the same argument as in the primal problem,
symmetry of the RKHS gives \(-g\in\mathcal G_q\), and \(\mathcal C_{q,n}(g) = \mathcal L_{q,n}(\widehat q,g) + \mathcal L_{q,n}(q_0,-g)\).
Hence, using \eqref{eq:adjoint-truth-game},
		\begin{equation}
		\begin{aligned}
			\sup_{g\in\mathcal G_q}\mathcal C_{q,n}(g)
			&\lesssim
			\chi_{q,n}^2
			+
			\frac{\lambda_{q,n}}2
			\left(
			\|q_0\|_{\mathcal H_q}^2
			-
			\|\widehat q\|_{\mathcal H_q}^2
			\right).
		\end{aligned}
		\label{eq:adjoint-centered-upper}
		\end{equation}

		\paragraph{Step 2: lower bound from the scaled tied critic.}
Write
\(h_q=\widehat q-q_0\) and
\(g_q=\sigma_{R,0}^2T^\ast h_q\). If
\(\|g_q\|_2\le\chi_{q,n}\), the desired projected-rate bound is
immediate. Suppose \(\|g_q\|_2>\chi_{q,n}\). Write
\(a_q=\chi_{q,n}/(2\|g_q\|_2)\), and for
\(h\in\mathcal H_q\) write
		\[
		M_{q,d}(h)
		=
		\{R-\widehat\pi_d(X)\}
		\{A-\widehat\alpha_d(W,Z,X)\}h(Z,X).
		\]
By Assumption~\ref{ass:ls-classes},
\(g_q\in\mathcal G_q\), and symmetry and star-convexity imply that
\(g_q^\star=-a_qg_q\in\mathcal G_q\). Since
\(\xi_{q,d}(\widehat q)-\xi_{q,d}(q_0)=-M_{q,d}(h_q)\), the same
evaluation as in the primal problem gives
		\begin{equation}
		\begin{aligned}
			\sup_{g\in\mathcal G_q}\mathcal C_{q,n}(g)
			&\ge
			a_q\sum_{d=1}^D\omega_d
			P_0\{M_{q,d}(h_q)g_q\}
			+
			a_q\sum_{d=1}^D\omega_d
			(\mathbb P_d-P_0)\{M_{q,d}(h_q)g_q\}
			\\
			&\quad-
			a_q^2\sum_{d=1}^D\omega_d\mathbb P_dg_q^2
			-
			\kappa_{q,n}a_q^2\|g_q\|_{\mathcal G_q}^2.
		\end{aligned}
		\label{eq:adjoint-tied-evaluation}
		\end{equation}

We first control the population term. Let
\(e_0=R-\pi_0(X)=\sigma_{R,0}^2(X)\rho_0(R,X)\). As in the primal
calculation,
\(\E(e_0\mid W,Z,X)=0\) and
\(\E(A-\alpha_0\mid W,Z,X)=0\). Hence, for every \(h\in\mathcal H_q\),
		\begin{equation}
		\begin{aligned}
			\E\{M_{q,d}(h)\mid W,X\}
			&=
			\E\left[
			\{e_0-\Delta_{\pi,d}\}
			\{A-\alpha_0-\Delta_{\alpha,d}\}h
			\mid W,X
			\right]
			\\
			&=
			\E(e_0Ah\mid W,X)
			+
			\Delta_{\pi,d}
			\E(\Delta_{\alpha,d}h\mid W,X)
			\\
			&=
			[\sigma_{R,0}^2T^\ast h](W,X)
			+
			\Delta_{\pi,d}
			\E(\Delta_{\alpha,d}h\mid W,X).
		\end{aligned}
		\label{eq:adjoint-operator-identity}
		\end{equation}
By \eqref{eq:adjoint-operator-identity}, taking \(h=h_q\), multiplying by
\(g_q\), and averaging over the inner folds gives
		\begin{align*}
			\sum_{d=1}^D\omega_d
			P_0\{M_{q,d}(h_q)g_q\}
			&=
			\|g_q\|_2^2
			+
			\sum_{d=1}^D\omega_d
			P_0(\Delta_{\pi,d}\Delta_{\alpha,d}h_qg_q).
		\end{align*}
For a generic \(h\in\mathcal H_q\), with
\(g_h=\sigma_{R,0}^2T^\ast h\), the same argument as in
\eqref{eq:primal-slope-perturbation}, using
Assumptions~\ref{ass:ls-sampling},
\ref{ass:ls-classes}, and
\ref{ass:nuisance-rates}, together with
\eqref{eq:critic-interpolation} gives
		\begin{align*}
			\left|P_0(\Delta_{\pi,d}\Delta_{\alpha,d}h g_h)\right|
			&\lesssim
			\varepsilon_{\pi,n}\varepsilon_{\alpha,n}
			\|h\|_{\mathcal H_q}^{1+\vartheta_q}
			\|g_h\|_2^{1-\vartheta_q}.
		\end{align*}
Using
\(\varepsilon_{\pi,n}\varepsilon_{\alpha,n} \le b_{q,n}\le\chi_{q,n}^{1+\vartheta_q}\) and
\(\|g_q\|_2>\chi_{q,n}\), we obtain
		\begin{equation}
		\begin{aligned}
			a_q\sum_{d=1}^D\omega_d
			P_0\{M_{q,d}(h_q)g_q\}
			\ge
			\frac12\chi_{q,n}\|g_q\|_2
			-
			K\chi_{q,n}^2
			-
			K\frac{\chi_{q,n}^2}{B_q}
			\|h_q\|_{\mathcal H_q}^2.
		\end{aligned}
		\label{eq:adjoint-tied-population-bound}
		\end{equation}

Next, \eqref{eq:adjoint-tied-process} and
\(\delta_{q,n}\le\chi_{q,n}<\|g_q\|_2\) give
		\begin{equation}
		\begin{aligned}
			a_q
			\left|\sum_{d=1}^D\omega_d
			(\mathbb P_d-P_0)\{M_{q,d}(h_q)g_q\}\right|
			\lesssim
			\chi_{q,n}^2
			+
			\frac{\chi_{q,n}^2}{B_q}
			\|h_q\|_{\mathcal H_q}^2.
		\end{aligned}
		\label{eq:adjoint-tied-fluctuation-bound}
		\end{equation}

Likewise, \eqref{eq:adjoint-critic-norm},
\(\|g_q\|_2>\chi_{q,n}\ge\delta_{q,n}\), and
\(\|g_q\|_{\mathcal G_q} \le L_q\|h_q\|_{\mathcal H_q}\) imply
		\begin{equation}
		\begin{aligned}
			a_q^2
			\sum_{d=1}^D\omega_d\mathbb P_dg_q^2
			&\lesssim
			\chi_{q,n}^2
			+
			\frac{\chi_{q,n}^2}{B_q}
			\|h_q\|_{\mathcal H_q}^2.
		\end{aligned}
		\label{eq:adjoint-tied-quadratic-bound}
		\end{equation}

Finally,
		\begin{equation}
			\kappa_{q,n}a_q^2\|g_q\|_{\mathcal G_q}^2=\kappa_{q,n}a_q^2\|g_q\|_{\mathcal G_q}^2
			\le
			\frac14\kappa_{q,n}L_q^2
			\|h_q\|_{\mathcal H_q}^2.
			\label{eq:adjoint-tied-penalty-bound}
		\end{equation}

Substituting
\eqref{eq:adjoint-tied-population-bound},
\eqref{eq:adjoint-tied-fluctuation-bound},
\eqref{eq:adjoint-tied-quadratic-bound}, and
\eqref{eq:adjoint-tied-penalty-bound}
into \eqref{eq:adjoint-tied-evaluation}, and using
Assumption~\ref{ass:ls-penalty-lower}, the implication
\(\chi_{q,n}^2/B_q\le2L_q^2\kappa_{q,n}\), and
\(\|h_q\|_{\mathcal H_q}^2 \le2\|\widehat q\|_{\mathcal H_q}^2 +2\|q_0\|_{\mathcal H_q}^2\) gives
		\begin{equation}
		\begin{aligned}
			\sup_{g\in\mathcal G_q}\mathcal C_{q,n}(g)
			&\ge
			\frac12\chi_{q,n}\|g_q\|_2
			-
			K\chi_{q,n}^2
			-
			\frac{\lambda_{q,n}}4
			\left(
			\|\widehat q\|_{\mathcal H_q}^2
			+
			\|q_0\|_{\mathcal H_q}^2
			\right).
		\end{aligned}
		\label{eq:adjoint-centered-lower}
		\end{equation}

		\paragraph{Step 3: combine the two bounds.}
Combining \eqref{eq:adjoint-centered-upper} and
\eqref{eq:adjoint-centered-lower} gives
		\begin{align*}
			\frac12\chi_{q,n}\|g_q\|_2
			&\lesssim
			\chi_{q,n}^2
			+
			\frac34\lambda_{q,n}\|q_0\|_{\mathcal H_q}^2
			-
			\frac14\lambda_{q,n}\|\widehat q\|_{\mathcal H_q}^2\\
			&\le
			K\chi_{q,n}^2
			+
			\frac34\lambda_{q,n}\|q_0\|_{\mathcal H_q}^2.
		\end{align*}
Thus, when \(\|g_q\|_2>\chi_{q,n}\),
		\[
		\|\sigma_{R,0}^2T^\ast(\widehat q-q_0)\|_2
		=
		\|g_q\|_2
		=
		O_p\!\left(
		\chi_{q,n}
		+
		\frac{\lambda_{q,n}\|q_0\|_{\mathcal H_q}^2}
		{\chi_{q,n}}
		\right).
		\]

When \(\|g_q\|_2\le\chi_{q,n}\), the same bound holds trivially.
Assumption~\ref{ass:ls-penalty-calibration} further implies that the second
term is \(O(\chi_{q,n})\). By uniform IV overlap,
\(\|\sigma_{R,0}^2T^\ast(\widehat q-q_0)\|_2\) is equivalent to
\(\|T^\ast(\widehat q-q_0)\|_2\), and therefore
		\[
		\|T^\ast(\widehat q-q_0)\|_2
		=
		O_p(\chi_{q,n}).
		\]
	\end{proof}

\subsection{Spectral characterization of ill-posedness}
\label{subsec:appendix-ill-posedness}

We use the spectral characterization of the local measure of
ill-posedness in Lemma~11 of \citet{dikkala2020minimax}. For \(\delta>0\),
define \(\mathcal H_{\mu,B_\mu}^{\mid\delta} = \{h\in\mathcal H_{\mu,B_\mu}:\|Th\|_2\le\delta\}\)
and \(\mathcal H_{q,B_q}^{\mid\delta} = \{h\in\mathcal H_{q,B_q}:\|T^\ast h\|_2\le\delta\}\),
and let \(\tau_\mu(\delta) = \sup_{h\in\mathcal H_{\mu,B_\mu}^{\mid\delta}}\|h\|_2\)
and \(\tau_q(\delta) = \sup_{h\in\mathcal H_{q,B_q}^{\mid\delta}}\|h\|_2\).

Let
\(\{\varrho_{\mu,j},\varphi_{\mu,j}\}_{j\ge1}\) be the nonincreasing
kernel eigenvalues and corresponding
\(L_2(P_{W,X})\)-orthonormal eigenfunctions of \(\mathcal H_\mu\). For
\(m\ge1\), define the \(m\times m\) matrix \(V_{\mu,m}\) by \([V_{\mu,m}]_{ij} = P_0\{(T\varphi_{\mu,i})(T\varphi_{\mu,j})\}\),
for \(i,j=1,\ldots,m\), and let \(\sigma_{\mu,m}=\lambda_{\min}(V_{\mu,m})\).
Define
\(\{\varrho_{q,j},\varphi_{q,j}\}_{j\ge1}\), \(V_{q,m}\), and
\(\sigma_{q,m}\) analogously, with \(T^\ast\) in place of \(T\).

\begin{proposition}[Spectral bounds for local ill-posedness]
\label{prop:spectral-ill-posedness}
Suppose that there exist fixed constants \(c_\mu,c_q<\infty\) such that,
for every \(m\ge1\),
\[
\begin{aligned}
\left|
P_0\left\{
(T\varphi_{\mu,i})(T\varphi_{\mu,j})
\right\}
\right|
&\le
c_\mu\sigma_{\mu,m},
\quad
1\le i\le m<j,
\\
\left|
P_0\left\{
(T^\ast\varphi_{q,i})(T^\ast\varphi_{q,j})
\right\}
\right|
&\le
c_q\sigma_{q,m},
\quad
1\le i\le m<j.
\end{aligned}
\]
Then
\[
\begin{aligned}
\tau_\mu(\delta)^2
&\le
\inf_{m\ge1}
\left\{
\frac{4\delta^2}{\sigma_{\mu,m}}
+
(4c_\mu^2+1)B_\mu\varrho_{\mu,m+1}
\right\},\\
\tau_q(\delta)^2
&\le
\inf_{m\ge1}
\left\{
\frac{4\delta^2}{\sigma_{q,m}}
+
(4c_q^2+1)B_q\varrho_{q,m+1}
\right\}.
\end{aligned}
\]
\end{proposition}

Consequently, on any event on which
\(\widehat\mu-\mu_0\in\mathcal H_{\mu,B_\mu}\) and
\(\widehat q-q_0\in\mathcal H_{q,B_q}\), we have \(\|\widehat\mu-\mu_0\|_2 \le \tau_\mu\!\left( \|T(\widehat\mu-\mu_0)\|_2 \right)\)
and \(\|\widehat q-q_0\|_2 \le \tau_q\!\left( \|T^\ast(\widehat q-q_0)\|_2 \right)\).
Combining these inequalities with the projected bridge rates
\(\|T(\widehat\mu-\mu_0)\|_2=O_p(\chi_{\mu,n})\) and
\(\|T^\ast(\widehat q-q_0)\|_2=O_p(\chi_{q,n})\) gives the ordinary
\(L_2\) bridge rates \(\|\widehat\mu-\mu_0\|_2 = O_p\{\tau_\mu(\chi_{\mu,n})\}\)
and \(\|\widehat q-q_0\|_2 = O_p\{\tau_q(\chi_{q,n})\}\).
By Proposition~\ref{prop:spectral-ill-posedness}, these rates further
satisfy \(\|\widehat\mu-\mu_0\|_2 = O_p\!\left( \inf_{m\ge1} \left\{ \chi_{\mu,n}^2/\sigma_{\mu,m} + B_\mu\varrho_{\mu,m+1} \right\}^{1/2} \right)\)
and \(\|\widehat q-q_0\|_2 = O_p\!\left( \inf_{m\ge1} \left\{ \chi_{q,n}^2/\sigma_{q,m} + B_q\varrho_{q,m+1} \right\}^{1/2} \right)\).

\subsection{Pseudo-outcome regression}
\label{subsec:pseudo-outcome-regression}

\begin{lemma}[Pseudo-outcome regression]
\label{lem:pseudo-outcome-regression}
Under Assumptions~\ref{ass:nuisance-rates} and \ref{ass:large-sample-regularity},
\(H_0\in L_2(P_0)\) and \(\nu_r^\star\in L_2(P_X)\) for
\(r\in\{0,1\}\). Moreover, for each fixed evaluation fold \(\ell\),
\(\|\widehat H_\ell-H_0\|_2=o_p(1)\) and
\(\|\widehat\nu_{r,\ell}-\nu_r^\star\|_2=o_p(1)\) for \(r\in\{0,1\}\).
\end{lemma}

\begin{proof}
The pairwise boundedness condition implies \(\min\left\{ \|\mu_0\|_\infty,\, \|q_0\|_\infty \right\} <\infty\).
Therefore,
\begin{align*}
\|H_0\|_2
&=
\left\|
q_0(Z,X)
\{Y-A\mu_0(W,X)\}
\right\|_2\\
&\le
\min\left\{
\|q_0\|_2
\|Y-A\mu_0\|_\infty,\,
\|q_0\|_\infty
\|Y-A\mu_0\|_2
\right\}\\
&\le
\min\left\{
\|q_0\|_2
\left(
\|Y\|_\infty+\|\mu_0\|_\infty
\right),\,
\|q_0\|_\infty
\left(
\|Y\|_2+\|\mu_0\|_2
\right)
\right\}\\
&\le
\min\left\{
\|q_0\|_2
\left(
\|Y\|_\infty+\|\mu_0\|_\infty
\right),\,
\|q_0\|_\infty
\left(
\|Y\|_\infty+\|\mu_0\|_2
\right)
\right\}\\
&<
\infty.
\end{align*}
Thus, \(H_0\in L_2(P_0)\).

By uniform IV overlap and conditional Jensen's inequality,
\begin{align*}
\|\nu_1^\star\|_2^2
&=
\E\!\left[
\E(H_0\mid R=1,X)^2
\right]\\
&\le
\frac{1}{\underline\pi}
\E\!\left[
\pi_0(X)
\E(H_0\mid R=1,X)^2
\right]\\
&\le
\frac{1}{\underline\pi}
\E\!\left[
\pi_0(X)
\E(H_0^2\mid R=1,X)
\right]\\
&=
\frac{1}{\underline\pi}
\E\!\left[
\E(RH_0^2\mid X)
\right]\\
&=
\frac{1}{\underline\pi}
\E(RH_0^2)\le
\frac{1}{\underline\pi}
\E(H_0^2)
=
\frac{1}{\underline\pi}
\|H_0\|_2^2
<\infty.
\end{align*}
The same calculation gives
\(\nu_0^\star\in L_2(P_X)\). Hence
\(\nu_r^\star\in L_2(P_X)\) for \(r\in\{0,1\}\).

Adding and subtracting
\(\widehat q_\ell(Z,X)\{Y-A\mu_0(W,X)\}\) gives
\begin{align*}
\widehat H_\ell-H_0
&=
\widehat q_\ell\{Y-A\widehat\mu_\ell\}
-
q_0\{Y-A\mu_0\}\\
&=
\widehat q_\ell\{Y-A\widehat\mu_\ell\}
-
\widehat q_\ell\{Y-A\mu_0\}+
\widehat q_\ell\{Y-A\mu_0\}
-
q_0\{Y-A\mu_0\}\\
&=
(\widehat q_\ell-q_0)\{Y-A\mu_0\}
-
A\widehat q_\ell(\widehat\mu_\ell-\mu_0).
\end{align*}
Therefore,
\begin{align*}
\|\widehat H_\ell-H_0\|_2
&\le
\left(
\|Y\|_\infty+\|\mu_0\|_\infty
\right)
\|\widehat q_\ell-q_0\|_2+
\|\widehat q_\ell\|_\infty
\|\widehat\mu_\ell-\mu_0\|_2\\
&\le
\left(
\|Y\|_\infty+\|\mu_0\|_\infty+\|\widehat q_\ell\|_\infty
\right)
\left\{
\|\widehat q_\ell-q_0\|_2
+
\|\widehat\mu_\ell-\mu_0\|_2
\right\}.
\end{align*}

Alternatively, adding and subtracting
\(q_0(Z,X)\{Y-A\widehat\mu_\ell(W,X)\}\) gives
\begin{align*}
\widehat H_\ell-H_0
&=
\widehat q_\ell\{Y-A\widehat\mu_\ell\}
-
q_0\{Y-A\widehat\mu_\ell\}+
q_0\{Y-A\widehat\mu_\ell\}
-
q_0\{Y-A\mu_0\}\\
&=
(\widehat q_\ell-q_0)
\{Y-A\widehat\mu_\ell\}
-
Aq_0(\widehat\mu_\ell-\mu_0),
\end{align*}
and hence
\begin{align*}
\|\widehat H_\ell-H_0\|_2&=
\left\|
\widehat q_\ell\{Y-A\widehat\mu_\ell\}
-
q_0\{Y-A\mu_0\}
\right\|_2\\
&\le
\left(
\|Y\|_\infty+\|\widehat\mu_\ell\|_\infty
\right)
\|\widehat q_\ell-q_0\|_2+
\|q_0\|_\infty
\|\widehat\mu_\ell-\mu_0\|_2\\
&\le
\left(
\|Y\|_\infty+\|\widehat\mu_\ell\|_\infty+\|q_0\|_\infty
\right)
\left\{
\|\widehat q_\ell-q_0\|_2
+
\|\widehat\mu_\ell-\mu_0\|_2
\right\}.
\end{align*}

Combining the above two bounds,
\begin{align*}
\|\widehat H_\ell-H_0\|_2
&\le
\left[
\|Y\|_\infty
+
\min\left\{
\|\widehat q_\ell\|_\infty+\|\mu_0\|_\infty,\,
\|q_0\|_\infty+\|\widehat\mu_\ell\|_\infty
\right\}
\right]\\
&\quad\times
\left\{
\|\widehat q_\ell-q_0\|_2
+
\|\widehat\mu_\ell-\mu_0\|_2
\right\}=
o_p(1).
\end{align*}

Let \(\mathscr F_{\mathcal T_\ell}\) denote the
sigma-field generated by the bridge-training sample
\(\mathcal T_\ell\). Because the observations in
\(\mathcal J_\ell\) are independent of the training data \(\mathcal T_\ell\), \(\E(H_0\mid R=1,X, \mathscr F_{\mathcal T_\ell}) = \E(H_0\mid R=1,X) = \nu_1^\star(X)\).
Then, \(\bar\nu_{1,\ell}(X)-\nu_1^\star(X) = \E(\widehat H_\ell-H_0 \mid R=1,X,\mathscr F_{\mathcal T_\ell})\).
Conditional Jensen's inequality and uniform IV overlap yield
\begin{align*}
\|\bar\nu_{1,\ell}-\nu_1^\star\|_2^2
&=
\E\!\left[
\left\{
\E(\widehat H_\ell-H_0
\mid R=1,X,\mathscr F_{\mathcal T_\ell})
\right\}^2
\Bigm|
\mathscr F_{\mathcal T_\ell}
\right]\\
&\le
\frac{1}{\underline\pi}
\E\!\left[
\pi_0(X)
\left\{
\E(\widehat H_\ell-H_0
\mid R=1,X,\mathscr F_{\mathcal T_\ell})
\right\}^2
\Bigm|
\mathscr F_{\mathcal T_\ell}
\right]\\
&\le
\frac{1}{\underline\pi}
\E\!\left[
\pi_0(X)
\E\!\left\{
(\widehat H_\ell-H_0)^2
\mid R=1,X,\mathscr F_{\mathcal T_\ell}
\right\}
\Bigm|
\mathscr F_{\mathcal T_\ell}
\right]\\
&=
\frac{1}{\underline\pi}
\E\!\left[
P(R=1\mid X,\mathscr F_{\mathcal T_\ell})
\E\!\left\{
(\widehat H_\ell-H_0)^2
\mid R=1,X,\mathscr F_{\mathcal T_\ell}
\right\}
\Bigm|
\mathscr F_{\mathcal T_\ell}
\right]\\
&=
\frac{1}{\underline\pi}
\E\!\left[
\E\!\left\{
R(\widehat H_\ell-H_0)^2
\mid X,\mathscr F_{\mathcal T_\ell}
\right\}
\Bigm|
\mathscr F_{\mathcal T_\ell}
\right]\\
&=
\frac{1}{\underline\pi}
\E\!\left[
R(\widehat H_\ell-H_0)^2
\mid\mathscr F_{\mathcal T_\ell}
\right]\\
&\le
\frac{1}{\underline\pi}
\E\!\left[
(\widehat H_\ell-H_0)^2
\mid\mathscr F_{\mathcal T_\ell}
\right]=
\frac{1}{\underline\pi}
\|\widehat H_\ell-H_0\|_2^2.
\end{align*}
The result for \(r=0\) follows analogously.
Consequently,
\begin{align*}
\|\widehat\nu_{r,\ell}-\nu_r^\star\|_2
&\le
\|\widehat\nu_{r,\ell}-\bar\nu_{r,\ell}\|_2
+
\|\bar\nu_{r,\ell}-\nu_r^\star\|_2\\
&\le
\|\widehat\nu_{r,\ell}-\bar\nu_{r,\ell}\|_2
+
\underline\pi^{-1/2}
\|\widehat H_\ell-H_0\|_2\\
&=
O_p(\varepsilon_{\nu,r,n})+o_p(1)\\
&=
o_p(1), \quad r\in\{0,1\}.
\end{align*}
\end{proof}

\subsection{Asymptotic normality of final estimator}
\label{subsec:final-estimator}

\begin{proof}[Proof of Theorem~\ref{thm:final-asymptotics}.]
Let \(\mathbb P_\ell\) denote the empirical measure on
\(\mathcal I_\ell\). Consider the following decomposition:
\begin{equation}
\sqrt n(\widehat\theta-\theta_0)
=
\underbrace{
\frac{1}{\sqrt n}
\sum_{i=1}^n
\phi_0(O_i)
}_{\displaystyle T_1:=}
+
\underbrace{
\sqrt n
\sum_{\ell=1}^L
\frac{n_\ell}{n}
(\mathbb P_\ell-P_0)
(\widehat\psi_\ell-\psi_0)
}_{\displaystyle T_2:=}
+
\underbrace{
\sqrt n
\sum_{\ell=1}^L
\frac{n_\ell}{n}
(P_0\widehat\psi_\ell-\theta_0)
}_{\displaystyle T_3:=}.
\label{eq:final-estimator-decomposition}
\end{equation}

\paragraph{Analysis of \(T_1\).}
Lemma~\ref{lem:pseudo-outcome-regression} gives
\(H_0\in L_2(P_0)\) and
\(\nu_r^\star\in L_2(P_X)\) for \(r\in\{0,1\}\). Hence, 
\begin{align*}
\|\psi_0\|_2
&=
\left\|
\mu_0
+\nu_1^\star
-\nu_0^\star
+\frac{R}{\pi_0(X)}(H_0-\nu_1^\star)
-\frac{1-R}{1-\pi_0(X)}(H_0-\nu_0^\star)
\right\|_2\\
&\le
\|\mu_0\|_2
+\|\nu_1^\star\|_2
+\|\nu_0^\star\|_2
+\left\|
\frac{R}{\pi_0(X)}(H_0-\nu_1^\star)
\right\|_2+
\left\|
\frac{1-R}{1-\pi_0(X)}(H_0-\nu_0^\star)
\right\|_2\\
&\le
\|\mu_0\|_2
+\|\nu_1^\star\|_2
+\|\nu_0^\star\|_2
+\underline\pi^{-1}
\|R(H_0-\nu_1^\star)\|_2+
\underline\pi^{-1}
\|(1-R)(H_0-\nu_0^\star)\|_2\\
&\le
\|\mu_0\|_2
+\|\nu_1^\star\|_2
+\|\nu_0^\star\|_2
+\underline\pi^{-1}
\|H_0-\nu_1^\star\|_2+
\underline\pi^{-1}
\|H_0-\nu_0^\star\|_2\\
&\le
\|\mu_0\|_2
+\|\nu_1^\star\|_2
+\|\nu_0^\star\|_2+
\underline\pi^{-1}
\left\{
\|H_0\|_2+\|\nu_1^\star\|_2
\right\}+
\underline\pi^{-1}
\left\{
\|H_0\|_2+\|\nu_0^\star\|_2
\right\}
<\infty.
\end{align*}

By identification, \(P_0\psi_0=\theta_0\). Since
\(\psi_0\in L_2(P_0)\), the Cauchy--Schwarz inequality gives
\begin{align*}
|\theta_0|
&=
|P_0\psi_0|
\le
P_0|\psi_0|
\le
\{P_0(\psi_0^2)\}^{1/2}
=
\|\psi_0\|_2
<\infty.
\end{align*}
Therefore,
\[P_0(\phi_0^2)=
P_0\{(\psi_0-\theta_0)^2\}
=
P_0(\psi_0^2)
-2\theta_0P_0\psi_0
+\theta_0^2
=
P_0(\psi_0^2)-\theta_0^2
<\infty.
\]
It follows by CLT that
\[
T_1
=
\frac{1}{\sqrt n}\sum_{i=1}^n\phi_0(O_i)
\xrightarrow{\cal D}
N(0,V_\mathrm{eff}).
\]

\paragraph{Analysis of \(T_2\).}
Fix an arbitrary fold \(\ell\). By definitions
\begin{equation}
\begin{aligned}
\widehat\psi_\ell-\psi_0
&=
\widehat\mu_\ell-\mu_0+
\left(
\frac{R}{\widehat\pi_\ell}
-
\frac{1-R}{1-\widehat\pi_\ell}
\right)
(\widehat H_\ell-H_0)+
\left(
1-\frac{R}{\widehat\pi_\ell}
\right)
(\widehat\nu_{1,\ell}-\nu_1^\star)\\
&\quad+
\left(
\frac{1-R}{1-\widehat\pi_\ell}-1
\right)
(\widehat\nu_{0,\ell}-\nu_0^\star)+
R
\left(
\frac{1}{\widehat\pi_\ell}
-
\frac{1}{\pi_0}
\right)
(H_0-\nu_1^\star)\\
&\quad-
(1-R)
\left(
\frac{1}{1-\widehat\pi_\ell}
-
\frac{1}{1-\pi_0}
\right)
(H_0-\nu_0^\star).
\end{aligned}
\label{eq:final-score-difference-decomposition}
\end{equation}
IV propensity overlap and clipping imply
\[
\left|
\frac{1}{\widehat\pi_\ell}
-
\frac{1}{\pi_0}
\right|
\lesssim
|\widehat\pi_\ell-\pi_0|,
\quad
\left|
\frac{1}{1-\widehat\pi_\ell}
-
\frac{1}{1-\pi_0}
\right|
\lesssim
|\widehat\pi_\ell-\pi_0|,
\]
and all remaining inverse-propensity multipliers are bounded.

For each \(r\in\{0,1\}\) and any fixed \(M<\infty\),
\begin{align*}
P_0\!\left[
(\widehat\pi_\ell-\pi_0)^2
(H_0-\nu_r^\star)^2
\right]
&=
P_0\!\left[
(\widehat\pi_\ell-\pi_0)^2
(H_0-\nu_r^\star)^2
I\left\{(H_0-\nu_r^\star)^2\le M\right\}
\right]\\
&\quad+
P_0\!\left[
(\widehat\pi_\ell-\pi_0)^2
(H_0-\nu_r^\star)^2
I\left\{(H_0-\nu_r^\star)^2>M\right\}
\right]\\
&\le
M P_0\!\left[
(\widehat\pi_\ell-\pi_0)^2
I\left\{(H_0-\nu_r^\star)^2\le M\right\}
\right]\\
&\quad+
P_0\!\left[
(H_0-\nu_r^\star)^2
I\left\{(H_0-\nu_r^\star)^2>M\right\}
\right]\\
&\le
M\|\widehat\pi_\ell-\pi_0\|_2^2
+
P_0\!\left[
(H_0-\nu_r^\star)^2
I\left\{(H_0-\nu_r^\star)^2>M\right\}
\right].
\end{align*}
For any \(\varepsilon>0\), choose \(M<\infty\) such that
\[
P_0\!\left[
(H_0-\nu_r^\star)^2
I\left\{(H_0-\nu_r^\star)^2>M\right\}
\right]
<
\frac{\varepsilon}{2}.
\]
Since \(\|\widehat\pi_\ell-\pi_0\|_2^2=o_p(1)\),
\[
P\!\left(
P_0\!\left[
(\widehat\pi_\ell-\pi_0)^2
(H_0-\nu_r^\star)^2
\right]
>\varepsilon
\right)\le
P\!\left(
M\|\widehat\pi_\ell-\pi_0\|_2^2
>
\frac{\varepsilon}{2}
\right)
\rightarrow0.
\]
Therefore, for each \(r\in\{0,1\}\), \(P_0\!\left[ (\widehat\pi_\ell-\pi_0)^2 (H_0-\nu_r^\star)^2 \right] =o_p(1)\).

By \eqref{eq:final-score-difference-decomposition} and the triangle inequality,
\[
\begin{aligned}
\|\widehat\psi_\ell-\psi_0\|_2^2
&\lesssim
\|\widehat\mu_\ell-\mu_0\|_2^2
+
\|\widehat H_\ell-H_0\|_2^2
+
\|\widehat\nu_{1,\ell}-\nu_1^\star\|_2^2
+
\|\widehat\nu_{0,\ell}-\nu_0^\star\|_2^2
\\
&\quad+
P_0\!\left[
(\widehat\pi_\ell-\pi_0)^2
(H_0-\nu_1^\star)^2
\right]
+
P_0\!\left[
(\widehat\pi_\ell-\pi_0)^2
(H_0-\nu_0^\star)^2
\right]\\
&=o_p(1).
\end{aligned}
\]

Let \(\mathscr F_{\mathcal T_\ell\cup\mathcal J_\ell}\) denote the
sigma-field generated by the bridge-training sample
\(\mathcal T_\ell\) and the pseudo-outcome regression sample
\(\mathcal J_\ell\). Then,
\begin{align*}
&\E\!\left[
\sqrt{n_\ell}
(\mathbb P_\ell-P_0)
(\widehat\psi_\ell-\psi_0)
\Bigm|
\mathscr F_{\mathcal T_\ell\cup\mathcal J_\ell}
\right]\\
={}&
\E\!\left[
\frac{1}{\sqrt{n_\ell}}
\sum_{i\in\mathcal I_\ell}
\left\{
(\widehat\psi_\ell-\psi_0)(O_i)
-
P_0(\widehat\psi_\ell-\psi_0)
\right\}
\Bigm|
\mathscr F_{\mathcal T_\ell\cup\mathcal J_\ell}
\right]\\
={}&
\frac{1}{\sqrt{n_\ell}}
\sum_{i\in\mathcal I_\ell}
\left[
\E\!\left\{
(\widehat\psi_\ell-\psi_0)(O_i)
\mid
\mathscr F_{\mathcal T_\ell\cup\mathcal J_\ell}
\right\}
-
P_0(\widehat\psi_\ell-\psi_0)
\right]\\
={}&
\frac{1}{\sqrt{n_\ell}}
\sum_{i\in\mathcal I_\ell}
\left\{
P_0(\widehat\psi_\ell-\psi_0)
-
P_0(\widehat\psi_\ell-\psi_0)
\right\}=0,
\end{align*}
and 
\begin{align*}
&\operatorname{Var}\!\left[
\sqrt{n_\ell}
(\mathbb P_\ell-P_0)
(\widehat\psi_\ell-\psi_0)
\Bigm|
\mathscr F_{\mathcal T_\ell\cup\mathcal J_\ell}
\right]\\
={}&
\operatorname{Var}\!\left[
\frac{1}{\sqrt{n_\ell}}
\sum_{i\in\mathcal I_\ell}
\left\{
(\widehat\psi_\ell-\psi_0)(O_i)
-
P_0(\widehat\psi_\ell-\psi_0)
\right\}
\Bigm|
\mathscr F_{\mathcal T_\ell\cup\mathcal J_\ell}
\right]\\
={}&
\frac{1}{n_\ell}
\sum_{i\in\mathcal I_\ell}
\operatorname{Var}\!\left[
(\widehat\psi_\ell-\psi_0)(O_i)
\Bigm|
\mathscr F_{\mathcal T_\ell\cup\mathcal J_\ell}
\right]\\
={}&
P_0\!\left[
\left\{
(\widehat\psi_\ell-\psi_0)
-
P_0(\widehat\psi_\ell-\psi_0)
\right\}^2
\right]\\
={}&
P_0\!\left[
(\widehat\psi_\ell-\psi_0)^2
\right]
-
\left\{
P_0(\widehat\psi_\ell-\psi_0)
\right\}^2\\
\le{}&
P_0\!\left[
(\widehat\psi_\ell-\psi_0)^2
\right]=
\|\widehat\psi_\ell-\psi_0\|_2^2
=
o_p(1).
\end{align*}
For every \(\varepsilon,\delta>0\), the conditional Chebyshev's
inequality implies
\begin{align*}
\!&P\left(
\left|
\sqrt{n_\ell}
(\mathbb P_\ell-P_0)
(\widehat\psi_\ell-\psi_0)
\right|>\varepsilon
\right)\\
={}&
\E\!\left[
P\!\left(
\left|
\sqrt{n_\ell}
(\mathbb P_\ell-P_0)
(\widehat\psi_\ell-\psi_0)
\right|>\varepsilon
\Bigm|
\mathscr F_{\mathcal T_\ell\cup\mathcal J_\ell}
\right)
\right]\\
={}&
\E\!\left[
I\left\{
\|\widehat\psi_\ell-\psi_0\|_2^2>\delta
\right\}
P\!\left(
\left|
\sqrt{n_\ell}
(\mathbb P_\ell-P_0)
(\widehat\psi_\ell-\psi_0)
\right|>\varepsilon
\Bigm|
\mathscr F_{\mathcal T_\ell\cup\mathcal J_\ell}
\right)
\right]\\
&+
\E\!\left[
I\left\{
\|\widehat\psi_\ell-\psi_0\|_2^2\le\delta
\right\}
P\!\left(
\left|
\sqrt{n_\ell}
(\mathbb P_\ell-P_0)
(\widehat\psi_\ell-\psi_0)
\right|>\varepsilon
\Bigm|
\mathscr F_{\mathcal T_\ell\cup\mathcal J_\ell}
\right)
\right]\\
\le{}&
P\!\left(
\|\widehat\psi_\ell-\psi_0\|_2^2>\delta
\right)+
\E\!\left[
I\left\{
\|\widehat\psi_\ell-\psi_0\|_2^2\le\delta
\right\}
\frac{
\operatorname{Var}\!\left[
\sqrt{n_\ell}
(\mathbb P_\ell-P_0)
(\widehat\psi_\ell-\psi_0)
\Bigm|
\mathscr F_{\mathcal T_\ell\cup\mathcal J_\ell}
\right]
}{\varepsilon^2}
\right]\\
\le{}&
P\!\left(
\|\widehat\psi_\ell-\psi_0\|_2^2>\delta
\right)
+
\frac{\delta}{\varepsilon^2}.
\end{align*}
Therefore,
\[
\limsup_{n\to\infty}
P\left(
\left|
\sqrt{n_\ell}
(\mathbb P_\ell-P_0)
(\widehat\psi_\ell-\psi_0)
\right|>\varepsilon
\right)
\le
\frac{\delta}{\varepsilon^2}.
\]
Since \(\delta>0\) is arbitrary,
\[
\sqrt{n_\ell}
(\mathbb P_\ell-P_0)
(\widehat\psi_\ell-\psi_0)
=
o_p(1).
\]
Thus,
\[
T_2
=
\sum_{\ell=1}^L
\sqrt{\frac{n_\ell}{n}}
\left[
\sqrt{n_\ell}
(\mathbb P_\ell-P_0)
(\widehat\psi_\ell-\psi_0)
\right]
=
o_p(1).
\]

\paragraph{Analysis of \(T_3\).}
The Cauchy--Schwarz inequality gives
\begin{align*}
\left|
P_0\!\left[
(\widehat q_\ell-q_0)
\{T\widehat\mu_\ell-T\mu_0\}
\right]
\right|
&=
\left|
P_0\!\left[
(\widehat q_\ell-q_0)
T(\widehat\mu_\ell-\mu_0)
\right]
\right|\le
\|\widehat q_\ell-q_0\|_2
\|T\widehat\mu_\ell-T\mu_0\|_2,
\\
\left|
P_0\!\left[
(\widehat q_\ell-q_0)
\{T\widehat\mu_\ell-T\mu_0\}
\right]
\right|
&=
\left|
P_0\!\left[
\{T^\ast\widehat q_\ell-T^\ast q_0\}
(\widehat\mu_\ell-\mu_0)
\right]
\right|\le
\|T^\ast\widehat q_\ell-T^\ast q_0\|_2
\|\widehat\mu_\ell-\mu_0\|_2.
\end{align*}
Consequently,
\[
\left|
P_0\!\left[
(\widehat q_\ell-q_0)
\{T\widehat\mu_\ell-T\mu_0\}
\right]
\right|\le
\min\Bigl\{
\|\widehat q_\ell-q_0\|_2
\|T\widehat\mu_\ell-T\mu_0\|_2,\,
\|T^\ast\widehat q_\ell-T^\ast q_0\|_2
\|\widehat\mu_\ell-\mu_0\|_2
\Bigr\}.
\]

Applying Proposition~\ref{prop:nested-robustness} gives
\begin{align*}
&\sqrt n\,
\left|
P_0\widehat\psi_\ell-\theta_0
\right|\\
={}&
\sqrt n\,
\Bigg|
-
P_0\!\left[
(\widehat q_\ell-q_0)
\{T\widehat\mu_\ell-T\mu_0\}
\right]+
P_0\!\left[
(\pi_0-\widehat\pi_\ell)
\left\{
\frac{\bar\nu_{1,\ell}-\widehat\nu_{1,\ell}}
{\widehat\pi_\ell}
+
\frac{\bar\nu_{0,\ell}-\widehat\nu_{0,\ell}}
{1-\widehat\pi_\ell}
\right\}
\right]
\Bigg|\\
\le{}&
\sqrt n\,
\left|
P_0\!\left[
(\widehat q_\ell-q_0)
\{T\widehat\mu_\ell-T\mu_0\}
\right]
\right|+
\sqrt n\,
\left|
P_0\!\left[
\frac{\pi_0-\widehat\pi_\ell}
{\widehat\pi_\ell}
(\bar\nu_{1,\ell}-\widehat\nu_{1,\ell})
\right]
\right|+
\sqrt n\,
\left|
P_0\!\left[
\frac{\pi_0-\widehat\pi_\ell}
{1-\widehat\pi_\ell}
(\bar\nu_{0,\ell}-\widehat\nu_{0,\ell})
\right]
\right|\\
\lesssim{}&
\sqrt n\,
\min\Bigl\{
\|\widehat q_\ell-q_0\|_2
\|T\widehat\mu_\ell-T\mu_0\|_2,\,
\|T^\ast\widehat q_\ell-T^\ast q_0\|_2
\|\widehat\mu_\ell-\mu_0\|_2
\Bigr\}\\
&+
\sqrt n\,
\|\widehat\pi_\ell-\pi_0\|_2
\left\{
\|\widehat\nu_{1,\ell}-\bar\nu_{1,\ell}\|_2
+
\|\widehat\nu_{0,\ell}-\bar\nu_{0,\ell}\|_2
\right\}\\
={}&
O_p\!\left[
\sqrt n\,
\min\Bigl\{
\chi_{\mu,n}\tau_q(\chi_{q,n}),\,
\chi_{q,n}\tau_\mu(\chi_{\mu,n})
\Bigr\}
\right]+
O_p\!\left[
\sqrt n\,
\varepsilon_{\pi,n}
\left\{
\varepsilon_{\nu,1,n}
+
\varepsilon_{\nu,0,n}
\right\}
\right]\\
={}&
o_p(1).
\end{align*}
Therefore,
\[
|T_3|=
\left|
\sqrt n
\sum_{\ell=1}^L
\frac{n_\ell}{n}
\left(
P_0\widehat\psi_\ell-\theta_0
\right)
\right|\le
\sum_{\ell=1}^L
\frac{n_\ell}{n}
\left|
\sqrt n
\left(
P_0\widehat\psi_\ell-\theta_0
\right)
\right|=
o_p(1).
\]

Combining the results for \(T_1\), \(T_2\), and \(T_3\), we obtain
\[
\sqrt n(\widehat\theta_{\mathrm{eff}}-\theta_0)
=
\frac1{\sqrt n}\sum_{i=1}^n\phi_0(O_i)
+
o_p(1).
\]
Therefore, by Slutsky's theorem, \(\sqrt n(\widehat\theta_{\mathrm{eff}}-\theta_0) \xrightarrow{\cal D} N(0,V_{\mathrm{eff}})\).

\paragraph{Variance consistency.}
By \(\sum_{\ell=1}^L \sum_{i\in\mathcal I_\ell} \widehat\psi_\ell(O_i) = n\widehat\theta_{\mathrm{eff}}\), we have
\begin{align*}
\widehat V_{\mathrm{eff}}
&=
\frac{1}{n-1}
\sum_{\ell=1}^L
\sum_{i\in\mathcal I_\ell}
\left\{
\widehat\psi_\ell(O_i)-\widehat\theta_{\mathrm{eff}}
\right\}^2\\
&=
\frac{n}{n-1}
\left\{
\sum_{\ell=1}^L
\frac{n_\ell}{n}
\mathbb P_\ell\widehat\psi_\ell^2
-
\widehat\theta_{\mathrm{eff}}^2
\right\}.
\end{align*}

Fix an arbitrary fold \(\ell\). Conditional on \(\mathscr F_{\mathcal T_\ell\cup\mathcal J_\ell}\),
the observations in \(\mathcal I_\ell\) are i.i.d. from \(P_0\).
Since
\(\|\widehat\psi_\ell-\psi_0\|_2^2=o_p(1)\), by conditional Markov inequality, for every \(\varepsilon,\delta>0\),
\begin{align*}
P\!\left(
\mathbb P_\ell
(\widehat\psi_\ell-\psi_0)^2
>\varepsilon
\right)
&=
\E\!\left[
P\!\left(
\mathbb P_\ell
(\widehat\psi_\ell-\psi_0)^2
>\varepsilon
\Bigm|
\mathscr F_{\mathcal T_\ell\cup\mathcal J_\ell}
\right)
\right]\\
&=
\E\!\left[
I\left\{
\|\widehat\psi_\ell-\psi_0\|_2^2>\delta
\right\}
P\!\left(
\mathbb P_\ell
(\widehat\psi_\ell-\psi_0)^2
>\varepsilon
\Bigm|
\mathscr F_{\mathcal T_\ell\cup\mathcal J_\ell}
\right)
\right]\\
&\quad+
\E\!\left[
I\left\{
\|\widehat\psi_\ell-\psi_0\|_2^2\le\delta
\right\}
P\!\left(
\mathbb P_\ell
(\widehat\psi_\ell-\psi_0)^2
>\varepsilon
\Bigm|
\mathscr F_{\mathcal T_\ell\cup\mathcal J_\ell}
\right)
\right]\\
&\le
P\!\left(
\|\widehat\psi_\ell-\psi_0\|_2^2>\delta
\right)+
\E\!\left[
I\left\{
\|\widehat\psi_\ell-\psi_0\|_2^2\le\delta
\right\}
\frac{
\E\!\left[
\mathbb P_\ell
(\widehat\psi_\ell-\psi_0)^2
\Bigm|
\mathscr F_{\mathcal T_\ell\cup\mathcal J_\ell}
\right]
}{\varepsilon}
\right]\\
&=
P\!\left(
\|\widehat\psi_\ell-\psi_0\|_2^2>\delta
\right)+
\E\!\left[
I\left\{
\|\widehat\psi_\ell-\psi_0\|_2^2\le\delta
\right\}
\frac{
\|\widehat\psi_\ell-\psi_0\|_2^2
}{\varepsilon}
\right]\\
&\le
P\!\left(
\|\widehat\psi_\ell-\psi_0\|_2^2>\delta
\right)
+
\frac{\delta}{\varepsilon}.
\end{align*}
Therefore, for every fixed \(\varepsilon,\delta>0\),
\[
\limsup_{n\to\infty}
P\!\left(
\mathbb P_\ell
(\widehat\psi_\ell-\psi_0)^2
>\varepsilon
\right)
\le
\frac{\delta}{\varepsilon}.
\]
Since \(\delta>0\) is arbitrary, we obtain \(\mathbb P_\ell (\widehat\psi_\ell-\psi_0)^2 = o_p(1)\).
Also, since \(\psi_0\in L_2(P_0)\), the law of large numbers gives
\(\mathbb P_\ell\psi_0^2= P_0\psi_0^2+o_p(1)=O_p(1)\). Hence, by Cauchy--Schwarz inequality,
\begin{align*}
\left|
\mathbb P_\ell\widehat\psi_\ell^2
-
\mathbb P_\ell\psi_0^2
\right|
&=
\left|
\mathbb P_\ell
\left[
(\widehat\psi_\ell-\psi_0)
(\widehat\psi_\ell+\psi_0)
\right]
\right|\\
&=
\left|
\mathbb P_\ell
\left[
(\widehat\psi_\ell-\psi_0)
\left\{
(\widehat\psi_\ell-\psi_0)+2\psi_0
\right\}
\right]
\right|\\
&=
\left|
\mathbb P_\ell
(\widehat\psi_\ell-\psi_0)^2
+
2\mathbb P_\ell
\left[
(\widehat\psi_\ell-\psi_0)\psi_0
\right]
\right|\\
&\le
\mathbb P_\ell
(\widehat\psi_\ell-\psi_0)^2
+
2
\left|
\mathbb P_\ell
\left[
(\widehat\psi_\ell-\psi_0)\psi_0
\right]
\right|\\
&\le
\mathbb P_\ell
(\widehat\psi_\ell-\psi_0)^2
+
2
\left\{
\mathbb P_\ell
(\widehat\psi_\ell-\psi_0)^2
\right\}^{1/2}
\left\{
\mathbb P_\ell\psi_0^2
\right\}^{1/2}\\
&=
o_p(1)
+
2\,o_p(1)O_p(1)\\
&=
o_p(1).
\end{align*}
Because \(L\) is fixed and
\(\mathcal I_1,\ldots,\mathcal I_L\) partition the sample,
\begin{align*}
\sum_{\ell=1}^L
\frac{n_\ell}{n}
\mathbb P_\ell\widehat\psi_\ell^2
&=
\sum_{\ell=1}^L
\frac{n_\ell}{n}
\mathbb P_\ell\psi_0^2
+
o_p(1)=
\frac1n\sum_{i=1}^n\psi_0(O_i)^2
+
o_p(1)=
P_0\psi_0^2+o_p(1).
\end{align*}

Moreover, the asymptotic linearity established above implies \(\widehat\theta_{\mathrm{eff}}\xrightarrow{p}\theta_0\).
Therefore, 
\begin{align*}
\widehat V_{\mathrm{eff}}
&=
\frac{n}{n-1}
\left\{
\sum_{\ell=1}^L
\frac{n_\ell}{n}
\mathbb P_\ell\widehat\psi_\ell^2
-
\widehat\theta_{\mathrm{eff}}^2
\right\}\\
&=
\frac{n}{n-1}
\left[
\left\{
P_0\psi_0^2+o_p(1)
\right\}
-
\left\{
\theta_0^2+o_p(1)
\right\}
\right]\\
&=
P_0\psi_0^2
+
\theta_0^2
+
o_p(1)\\
&=
P_0\left\{
\psi_0^2
-
2\theta_0\psi_0
+
\theta_0^2
\right\}
+
o_p(1)\\
&=P_0\{(\psi_0-\theta_0)^2\}
+
o_p(1)\\
&=
P_0(\phi_0^2)+o_p(1)=
V_{\mathrm{eff}}+o_p(1),
\end{align*}
which completes the proof that \(\widehat V_\mathrm{eff}\xrightarrow{p} V_\mathrm{eff}\).
\end{proof}